\documentclass[11pt,a4paper]{article}
\usepackage[left=1in,top=1in,right=1in,bottom=1in]{geometry}

\usepackage[UKenglish]{babel}
\usepackage[utf8]{inputenc}
\usepackage[T1]{fontenc}
\usepackage[authoryear]{natbib} 
\usepackage{amsmath,amssymb}
\usepackage[thmmarks, amsmath, thref, amsthm]{ntheorem}
\usepackage{mathtools}
\usepackage{enumerate}
\usepackage{subcaption}
\usepackage[hidelinks]{hyperref} 
\usepackage{bm} 
\usepackage{graphicx}
\usepackage[ruled]{algorithm}
\usepackage[noend]{algpseudocode}  

\newtheorem{theorem}{Theorem}

\newtheorem{observation}[theorem]{Observation}
\newtheorem{lemma}[theorem]{Lemma}

\newtheorem{proposition}[theorem]{Proposition}
\newtheorem{example}{Example}

\theoremstyle{definition}
\newtheorem{definition}{Definition}

\newcommand{\bids}{{\mathcal{B}}}
\newcommand{\negbids}{{\mathcal{B}^-}}
\newcommand{\posbids}{{\mathcal{B}^+}}
\newcommand{\bid}{{\bm{b}}}
\newcommand{\mb}{\bm{m}}
\newcommand{\eb}{\bm{e}}
\newcommand{\pb}{\bm{p}}
\newcommand{\qb}{\bm{q}}

\newcommand{\rb}{\bm{r}}
\newcommand{\bs}{\bm{s}}
\newcommand{\tb}{\bm{t}}
\newcommand{\vb}{\bm{v}}
\newcommand{\wb}{\bm{w}}
\newcommand{\xb}{\bm{x}}
\newcommand{\yb}{\bm{y}}
\newcommand{\zb}{\bm{z}}
\newcommand{\Z}{\mathbb{Z}}
\newcommand{\R}{\mathbb{R}}
\newcommand{\conp}{\mbox{{\normalfont coNP}}}
\newcommand{\validbids}{{\sc Valid Bids}}
\newcommand{\threecnf}{{\sc 3-CNF Satisfiability}}
\newcommand{\mnd}{{\sc More Negatives Dominated}}
\newcommand{\true}{\mbox{\sc true}}
\newcommand{\false}{\mbox{\sc false}}

\newcommand{\uep}{u^{2[\varepsilon,i]}}
\newcommand{\uagg}{u^{\{1,2\}}}
\renewcommand{\vec}[1]{\bm{#1}}
\newcommand{\Walrasian}{competitive}
\newcommand{\linkgood}{link good}
\newcommand{\linkgoods}{link goods}
\newcommand{\keylist}{demand cluster}
\newcommand{\keylists}{demand clusters}

\DeclareMathOperator*{\argmin}{arg\,min}
\DeclareMathOperator*{\argmax}{arg\,max}

\begin{document}

\title{Solving Strong-Substitutes Product-Mix Auctions}
\author{Elizabeth Baldwin\thanks{Dept.\ of Economics, Oxford University, UK,
\href{mailto:elizabeth.baldwin@economics.ox.ac.uk}{elizabeth.baldwin@economics.ox.ac.uk}.}
\and
Paul W. Goldberg\thanks{Dept.\ of Computer Science, Oxford University, UK,
\href{mailto:paul.goldberg@cs.ox.ac.uk}{paul.goldberg@cs.ox.ac.uk}.}
\and
Paul Klemperer\thanks{Dept.\ of Economics, Oxford University, UK,
\href{mailto:paul.klemperer@nuffield.ox.ac.uk}{paul.klemperer@nuffield.ox.ac.uk}.}
\and
Edwin Lock\thanks{Dept.\ of Computer Science, Oxford University, UK,
\href{mailto:edwin.lock@cs.ox.ac.uk}{edwin.lock@cs.ox.ac.uk}.}
}
\date{5th July 2023}

\maketitle

\begin{abstract}
This paper develops algorithms to solve strong-substitutes product-mix auctions: it finds \Walrasian{} equilibrium prices and quantities for agents who use this auction's bidding language to truthfully express their strong-substitutes preferences over an arbitrary number of goods, each of which is available in multiple discrete units. Our use of the bidding language, and the information it provides, contrasts with existing algorithms that rely on access to a valuation or demand oracle.

We compute market-clearing prices using algorithms that apply existing submodular minimisation methods. Allocating the supply among the bidders at these prices then requires solving a novel constrained matching problem. Our algorithm iteratively simplifies the allocation problem, perturbing bids and prices in a way that resolves tie-breaking choices created by bids that can be accepted on more than one good. We provide practical running time bounds on both price-finding and allocation, and illustrate experimentally that our allocation mechanism is practical.
\end{abstract}

\begin{paragraph}{Keywords:}
bidding language, product-mix auction, competitive equilibrium, Walrasian
equilibrium, convex optimisation, strong substitutes, submodular minimisation
\end{paragraph}

\section{Introduction}
\label{sec:intro}
This paper develops algorithms that solve product-mix auctions in which participants can make bids that represent any strong-substitutes preferences for an arbitrary number of distinct goods. (These preferences are also known, in other literatures, as $M^\natural$-concave, as matroidal, and as well-layered maps). It thus allows bidders to express more general preferences than could previously be permitted in these auctions, and finds competitive equilibrium prices and quantities consistent with these, and the auctioneer's preferences.

Importantly, our algorithms for finding equilibrium differ from existing ones in that they directly use the information that the product-mix auction `language' provides. This information is in a very different form from the information  provided by a valuation or demand oracle. This creates additional complexities, as well as simplifications which we can exploit. However, the language is conceptually simple, and easy for bidders to use in a (product-mix) auction.

The \emph{product-mix auction} was developed in 2007-8 for the Bank of England to provide liquidity to financial institutions by auctioning loans to them \citep{Kle2008}. It is now used at least monthly by the Bank, and more often when institutions are more likely to be under stress. (After the 2016 vote for `Brexit', and starting again in March 2019, for instance, the auction was run weekly.)

The original implementation of the auction allowed bidders to each submit a list of bids, where each bid has $n+1$ elements: a price for each of the $n$ goods available, and a total quantity of goods sought by that bid.\footnote{In the Bank of England's auction, the bidders are commercial banks, etc., each good is a loan secured against one of $n$ different specified qualities of collateral (so the prices are interest rates), and the quantity is the amount of the loan (in \pounds).}
The auction sets a uniform price for each good (i.e.~every recipient of any particular good pays the same price per unit for that good), and gives each bid an allocation that maximises that bid's `utility', assuming preferences with quasilinear utilities. That is, each bid is allocated its desired quantity of the good on which its price strictly exceeds the auction's price by most, if there is a unique such good, and is allocated nothing if all its prices are strictly below the auction's corresponding prices. (So if a bid only wants one particular good, it simply sets prices of zero for all the other $(n-1)$ goods.) A bid that creates a tie (i.e.~two or more of its prices exceed the auction's corresponding prices by most, or none of its prices exceed, but at least one equals, the auction's corresponding prices) can be allocated in any way consistent with equilibrium (see below). Figure~\ref{fig:bid-acceptance} illustrates how the good that a bid is allocated depends on the bid's relative position to the price vector. In particular, we note that individual bids $(b_1, \ldots, b_n ; +1)$ that seek exactly one item of a good can be interpreted as a unit-demand bidder with valuation vector $\bid = (b_1, \ldots, b_n)$ and quasilinear utilities. We motivate the bidding language in the following example.

\begin{figure}
	\centering
	\includegraphics{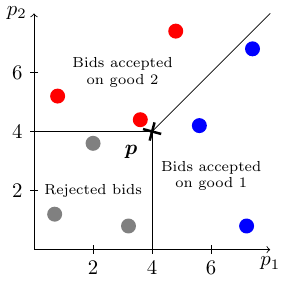}
	\caption{Price vectors divide $\R^n$ into regions that specify which good is allocated to bids lying in these region. This example works in the setting with two goods. The price vector $\pb = (4,4)$ (marked by a cross) divides $\R^2$ into three regions, each of which labelled according to the good that bids in this region are allocated.}
	\label{fig:bid-acceptance}
\end{figure}

\begin{example}[Figure~\ref{fig:Alice}]
\label{example:Alice}
Alice participates in an auction held with two goods, say apples (good 1) and bananas (good 2). For her breakfast, she requires a single item of fruit and is willing to pay up to \pounds 6 for either an apple or a banana. This is expressed by the bid $\bid \coloneqq (6,6;+1)$; the last element in the bid denotes the quantity of goods sought. Suppose the auction's prices are set at $\pb = (p_1, p_2)$. If at least one good is priced lower than the respective bid entry, she will demand the good $i \in \{1,2\}$ that maximises utility $b_i - p_i$. Otherwise, if both apples and bananas are priced higher than her valuation of~\pounds 6, she demands neither.
In addition, as she is fond of bananas, Alice wishes to pick up a banana as a \emph{second} fruit item, but only if the price of a banana does not exceed~\pounds 4. This is expressed by the bid $\bid' \coloneqq (0,4;+1)$.

To illustrate Alice's demand, consider the following two possible auction prices. If the auction prices are set at $\pb = (1,3)$, then Alice demands one apple and one banana: $\bid$ demands an apple, as its utility for an apple or a banana is 5 and 3, respectively, while $\bid'$ demands a banana, as its price is less than~\pounds 4. On other hand, if auction prices are $\pb = (6,5)$, Alice demands only one banana: $\bid$ demands a banana and $\bid'$ demands nothing, as the price of a banana is too high. Alice's two bids $\bid$ and $\bid'$, together with the demand correspondence they induce, are shown in Figure~\ref{fig:Alice}.
\end{example}

The auction's prices are chosen to maximise the sum of the bids' and auctioneer's welfare -- that is, it finds competitive equilibrium prices and quantities.%
\footnote{Specifically, the auction finds the lowest such price vector. Since the bids automatically express ``strong substitutes'' preferences for all bidders (see~{{\citep{BK-OS}}}), there exist equilibrium price vectors, and also a unique one among them at which every good’s price is lowest.}
Moreover, the auctioneer expresses her preferences about which goods to allocate in the form of supply functions which can themselves be equivalently represented as (and thus converted into) a list of bids of the kind described above.\footnote{See Appendix 1E of \citep{Kle2018} for how to convert supply functions into bid lists. (The bids that the auctioneer's supply functions would be converted into would ensure that it \emph{sells more} units on a good when prices are high, analogous to a buyer \emph{buying fewer} units in this case.) Although the auctioneer’s preferences \emph{could} equivalently have been represented as a list\ of bids of the kind made by the bidders, describing them simply as two-dimensional graphs of `supply schedules' was an important feature of the auction design: participants' ability to express their preferences in the ways that are most natural for them is crucial to getting an auction accepted for practical use, to getting bidders to participate, and to the auction working efficiently.

The Bank of England's original program restricted to $n=2$. Since 2014 its program permits much larger $n$ (it is currently being run with $n=3$), and it also allows richer forms of preferences to be expressed by the auctioneer (but not by the bidders, whose preference expression has, by contrast, been restricted in recent auctions).

A variant that allowed for bidders' budget constraints (hence non-quasilinear preferences) was programmed for the Government of Iceland in 2015-16. For additional information about the strong-substitutes product-mix auction and other variants of the product-mix auction that allow for divisible goods and budget constraints, etc., we refer to \url{http://pma.nuff.ox.ac.uk} and \citep{BK-OS,Kle2008,Kle2010,Kle2018}.}

\begin{figure}
	\centering
	\begin{subfigure}[b]{0.49\textwidth}
		\centering
		\includegraphics{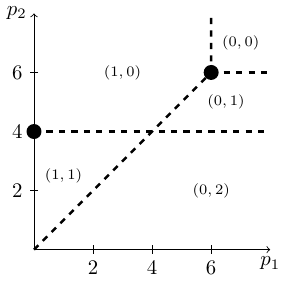}
		\caption{Alice's demand}
		\label{fig:Alice}
	\end{subfigure}
	\begin{subfigure}[b]{0.49\textwidth}
		\centering
		\includegraphics{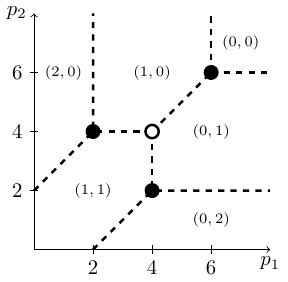}
		\caption{Bob's demand}
		\label{fig:Bob}
	\end{subfigure}
	\includegraphics{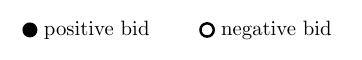}
	\caption{Examples of strong-substitutes demand correspondences on two goods \{1,2\} belonging to two bidders, Alice and Bob. Alice's demand is specified by bid list $\bids = \{ (6,6;+1), (0,4;+1) \}$ and Bob's bid list is $\bids' = \{(2,4;+1), (4,2;+1), (4,4;-1), (6,6;+1) \}$. Positive and negative bids are depicted as solid and hollow circles, respectively. Price space is divided into regions corresponding to demanded bundles $(x_1,x_2)$, with $x_i$ denoting the number of items of good $i$. At $\pb = (4,4)$, Alice demands bundles $\{(1,0), (0,1), (1,1), (0,2) \}$ and Bob demands bundles $\{(1,0), (0,1), (1,1) \}$, the discrete convex hulls of bundles demanded in the regions surrounding $\pb$.}
	\label{fig:ex1}
\end{figure}

\paragraph{Positive and negative bids.}
In versions of the product-mix auction thus far implemented, all bids are for positive quantities of goods. In the case with transferable utilities, the market-clearing prices can then be found by solving straightforward linear programs, and finding an allocation of the auctioneer's supply to the individual bidders is similarly straightforward. However, some strong-substitutes%
\footnote{\label{footnote3}
Strong-substitutes preferences are those that would be still be gross substitutes preferences if we treated every unit of every good as a separate good. Such preferences have many attractive properties; they mean, for example, that if the price of any one good increases, and the demand for it decreases, then the demand for all other goods can increase by at most the amount of that decrease. Strong substitutability is equivalent to $M^\natural$-concavity {\citep{ST15}} in our setting with valuations defined on non-negative and bounded domains.}
preferences (over bundles of non-negative quantities of goods) can only be expressed by using lists of bids for both positive \emph{and negative} quantities,%
\footnote{The ability to express such preferences was not thought necessary in the Bank of England's application, but might be useful in closely related environments (see \citep{Kle2018}).}
and it has been shown that \emph{any} strong-substitutes preferences can be represented using lists of positive and negative bids.%
\footnote{%
\label{footnote:expressivity}%
\Citet{Kle2010} stated this result for the case of multiple units of each of two goods. See \citep{BK-OS} and \citep{LinTran2017} for working papers proving (different) generalisations of this result. (As cited by \citet{LinTran2017}, a prior proof of the result itself was presented at a workshop; see \citep{BGK-slides}.)}
The rules for accepting, or not accepting, negative bids are identical to those for positive bids: if a negative bid is accepted at some prices, it is allocated a negative quantity of the good it demands, (partially) cancelling the allocated quantities of positive bids that are accepted at the same prices. 

To illustrate why negative bids are useful, we consider the setting with two distinct goods $1$ and~$2$ in which an agent has strong-substitutes valuation $u$ and is interested in at most 2 units. In other words, the agent only has a positive value for bundles $(0,1), (1,0), (1,1), (0,2)$ and $(2,0)$, where the 2-dimensional vectors $(x_1, x_2)$ denote bundles containing $x_1$ units of good 1 and $x_2$ units of good 2. It is known \citep{Shioura17} that the strong-substitutes condition is equivalent to M$^\natural$-concavity which, in our setting, is equivalent to discrete concavity together with the following three additional conditions:
\begin{enumerate}[i)]
    \item \label{item:negative-bids-i} $u \left ( (1,1)\right ) \leq u \left ( (0,1) \right ) + u \left ( (1,0) \right )$,
    \item \label{item:negative-bids-ii} $u \left ( (2, 0)\right ) - u \left ( (1, 0)\right ) \leq u \left ( (1,1) \right ) - u \left ( (0,1) \right )$.
    \item \label{item:negative-bids-iii} $u \left ( (0, 2)\right ) - u \left ( (0, 1)\right ) \leq u \left ( (1,1) \right ) - u \left ( (1,0) \right )$.
\end{enumerate}
It is easy to see (by infinitesimally and independently perturbing bundle values) that the generic case has strict inequalities in all three conditions. Intuitively, a strict inequality in the first condition implies that a single unit of good 1 is valued strictly less if the agent also has a unit of 2 than if he does not, reflecting strict substitutability as distinct from complementarity or independence between the two goods. A strict inequality in the second or third condition implies that the marginal utility of a good is reduced more by receiving an additional unit of the same good than by receiving a unit of a different good; `more variety' is valuable.
In contrast, it is easy to check that any demand correspondence generated only by positive bids is non-generic, in that at least one of the three conditions must be an equality. Thus a demand correspondence such as the one illustrated in Example~\ref{example:Bob} and Figure~\ref{fig:Bob}, in which all three conditions are strict inequalities, cannot be represented by positive bids only. We refer to \citep[p.~15]{Kle2008} and \citep[Appendix IC]{Kle2018} for further details and examples of the usefulness of negative bids.

\begin{example}[Figure~\ref{fig:Bob}]
\label{example:Bob}
A combination of positive and negative bids is required to express Bob's demand correspondence depicted in Figure~\ref{fig:Bob}, which divides price space into six demand regions for bundles with at most two units. The valuation function $u$ is given by $u((0,1)) = u((1,0)) = 6$, $u((0,2)) = u((2,0)) = 8$ and $u((1,1)) = 10$. The corresponding list of bids is given by $\bids' = \{(2,4;+1), (4,2;+1), (4,4;-1), (6,6;+1) \}$.
\end{example}

\paragraph{Our contributions.}
This paper addresses the computational challenges of determining uniform
component-wise minimal market-clearing prices (at which total market demands
equal the quantities of each good that are available) and allocating a fixed
collection of the goods at these prices, to bidders whose demands are defined
by lists of positive and negative bids that express strong-substitute
preferences. Our algorithms exploit the fact that the use of the product-mix bidding language allows for the efficient computation of a demanded bundle as well as of the indirect utility derived at any given prices.

Section~\ref{sec:prelims} introduces the product-mix auction's strong-substitutes bidding language in more detail, and develops some of its properties. A first contribution of our paper is to show that it is \conp-complete to determine whether a given list of positive and negative bids constitute a valid demand correspondence. However, when the number of goods, or the number of negative bids, is bounded by a constant, we present a polynomial-time algorithm for checking validity.

In Section~\ref{sec:price-finding-overview}, we consider algorithms for finding component-wise minimal market-clearing prices that have practical running time bounds. We adopt a discrete steepest descent method from the discrete optimisation literature {\citetext{\citealp[Chapter 10 in][]{murota-book}, \citealp{Shioura17}}} that employs submodular minimisation to find the steepest descent directions and discuss two techniques (one novel, and one suggested by \citet{Shioura17}) that reduce the number of iterations required by taking long steps in the steepest descent direction. Full details are provided in Appendix~\ref{apx:price-finding}. We show that the resulting algorithms run in time polynomial in the input size of the bid lists and, since submodular minimisation is rapid in practice \citep{CLSW}, we expect our approach to have a fast running time. Indeed, preliminary experiments (given in the online companion to this paper) suggest our method performs well in practice.

Our main contribution, given in Section~\ref{sec:main}, is an efficient polynomial-time algorithm that allocates the auctioneer's chosen supply among the bidders at given market-clearing prices. The difficulty in developing this lies in handling bids whose utility is maximised on more than one alternative good (or whose utility from its most-preferred good is exactly zero), as tie-breaking choices interact with each other. This gives rise to a novel matching problem. Our algorithm proceeds by iteratively simplifying the allocation problem at hand, allocating unambiguous bids and perturbing bids and prices in a way that resolves a subset of the tie-breaks and yields a simplified allocation problem. Progress is measured in terms of reductions in the number of edges of a multigraph associated with the allocation problems.

\paragraph{Software implementations.}
Our price-finding and allocation algorithms are conceptually simple. They are also straightforward to implement, which has allowed us to develop two practical implementations in Haskell and Python; these can be found at \url{http://pma.nuff.ox.ac.uk} and \url{https://github.com/edwinlock/product-mix}, respectively. Some effort was made to optimise for speed in the Python implementation by exploiting fast matrix operations provided by the NumPy package \citep{Numpy}. Furthermore, in an effort to implement submodular minimisation efficiently, the Fujishige-Wolfe algorithm \citep{Chakrabarty2014} was implemented in combination with a memoization technique to reduce the number of submodular function queries.

In order to evaluate the practical running time of our allocation algorithm, we used our Python implementation to run experiments on auctions with various numbers of goods, bidders and bids. The results of these experiments, given in the online companion to this paper, demonstrate that our allocation algorithm is efficient in practice for realistic auction sizes.

\paragraph{Related work.}
Our work continues a long literature, spanning economics and discrete convex analysis, on `gross' and `strong substitutes' \citep{KC82, MS09} and `$M^\natural$-concavity' \citep{Murota.Shioura1999}. Whereas the discrete convex analysis literature generally allows for multiple units of each good, much focus in economics has been on the case in which there is only one unit of each good; \citet{MS09} showed that `strong substitutes' provide the suitable generalisation of gross substitutes to the multi-unit case, retaining existence of equilibrium while insisting that any two units of the same good should have the same price.

There is now a substantial body of related work on the two algorithmic problems we address in this paper: finding market-clearing prices and computing an equilibrium allocation of supply to the buyers. 

Market-clearing prices are commonly found either by performing an iterative discrete steepest-descent search, or by solving a convex optimisation problem with a cutting plane method.

The method of steepest descent goes back to \citet{Aus06, KC82, MSY2013} and \citet{MSY2016}; see \citep[Chapter~11]{murota-book}, \citep{Murota2016} and \citep{PLsurvey} for recent surveys. However, while the steepest-descent methods in the above references run in pseudo-polynomial time (in the valuation and demand oracle settings), we describe two steepest descent algorithms based on \citep{Shioura17} that take longer steps by exploiting the bid representation of bidder demand, and show that they have polynomial running time in the size of the bid lists. (An algorithm is \textit{polynomial} if its running time is polynomial in the input size, and \textit{pseudo-polynomial} if its running time is required to be polynomial only in numeric values within the input. These differ when numerical values are encoded in binary.)

\Citet{baldwin2022strong} develop a conceptually similar descending-price algorithm using DC (difference of convex functions) programming.%
\footnote{This exploits a result that market-clearing prices minimise the difference between two linear programs (one for the positive bids and one for the negative bids).}
Extensive computational experiments in \citep{baldwin2022strong} suggest that this algorithm is competitive with the steepest descent algorithm we describe, and neither algorithm consistently outperforms the other. However, no polynomial running time bound is known for the algorithm of \citep{baldwin2022strong}, whereas we provide a polynomial bound (in the bid list representation size) for our method. 

\Citet{PLW} take the alternative approach of finding market-clearing prices by solving a convex optimisation problem. Their solution builds on an improved cutting plane method of \citet{LSW} to achieve a polynomial-time algorithm. However, to the best of our knowledge, the cutting plane method has not yet been implemented and is computationally expensive in practice; moreover, it is not guaranteed to find component-wise minimal prices.

The harder algorithmic problem is finding a valid allocation of supply to buyers. Most previous studies have failed to address this problem in the case when multiple units of each good are available. The two existing algorithms for finding multi-unit allocations of which we are aware, in \citep{MT2001} and \citep{PLW}, both work in the valuation oracle model. That is, their algorithms rely on oracle access to the valuation function of each bidder. In order to apply these algorithms in our bidding language setting, we would require the computational overhead of simulating a valuation oracle. Section~\ref{sec:simulating-valuation-oracle} describes a method for this simulation which takes time $O(n^2|\bids|^3 + n |\bids|T(n))$ per valuation query (or $O(n^2|\bids|^2 \log M + n |\bids|T(n))$ for an alternative variant), in which $n$ is the number of goods, $|\bids|$ is the number of bids, $M$ is an upper bound on the bid entries, and $T(n)$ is the time it takes to minimise an $n$-dimensional submodular set function.

The allocation algorithm of \citet{MT2001}, also described in \citep[Chapter 11]{murota-book}, works by reducing the allocation problem to an $M$-convex network flow problem using access to bidders' valuation functions. Following \citep[Chapters 10 and 11]{murota-book} and \citep{Shioura1998} (for a polynomial-time algorithm for minimising $M$-convex functions), this approach requires $O((nm)^6)$ queries to the valuation function, where $m$ is the number of bidders. With the additional overhead from simulating the valuation oracle this leads to a running time of $O(m^6n^8|\bids|^3 + m^6n^7 |\bids|T(n))$.

\Citet{PLW} provide a different algorithm, also in the valuation oracle setting, that solves the allocation problem only if the target bundle is \textit{uniquely} demanded at some prices. Such prices are guaranteed in the single-unit case (i.e.~when there is only one unit of every good) but need not exist in the multi-unit case. (Consider, for instance, the aggregate demand of Alice and Bob in Figure~\ref{fig:ex1}, also shown in Figure~\ref{fig:aggregate} of Section~\ref{sec:price-finding-overview}. The bundle $(1,1)$ is not uniquely demanded at any prices.) However, if the target bundle is uniquely demanded at some prices, \citep{PLW}'s algorithm makes $\tilde{O}(mn+n^3)$ queries to the valuation functions to allocate the target bundle. With the simulation of valuation queries, this leads to running-time complexity $\tilde{O}((mn^2 + n^4) (n |\bids|^3 + |\bids| T(n)))$.

In contrast to the algorithms of \citep{MT2001} and \citep{PLW}, our allocation algorithm directly exploits the specific representation provided by the product-mix auction's bidding language. Theorem~{\ref{thm:main}} establishes that our algorithm takes time $O(mn^3 \alpha(n)|\bids| + m n^2 T(n))$. As the near-constant Ackermann function~$\alpha$ satisfies $\alpha(n) \leq 4$ for any realistic number of goods $n$,%
\footnote{We have $\alpha(n) \leq 4$ for any $n \leq 2^{2^{2^{2^{16}}}}$.} 
we see that our algorithm outperforms those of \citep{MT2001} and \citep{PLW} (when applicable) in the product-mix auction setting. Moreover, our algorithm has been implemented and is efficient in practice.

We note that the valuation oracle setting of \citep{MSY2013, MSY2016} and others provides no straightforward way to compute a demanded bundle or the indirect utility function at given prices. In contrast, our bidding-language setting is somewhat analogous to, but more informative than, the demand oracle setting used in \citep{Aus06}, for example, which presupposes access to an oracle returning a demanded bundle at given prices.

This paper assumes that bidders are able to construct and submit bid lists that represent their strong-substitutes preferences. When the number of goods is large or the preferences are otherwise complex, this may present difficulties for bidders. \citet{Goldberg2022learningb} address this issue by presenting algorithms that generate bid lists on behalf of bidders using only access to a demand or valuation oracle to elicit bidders' preferences.

\section{Preliminaries}
\label{sec:prelims}
We denote $[n] \coloneqq \{1, \ldots, n\}$ and $[n]_0 \coloneqq \{0, \ldots, n\}$. In our
auction model, there are $n$ distinct goods $[n]$; a single copy of a good is an \emph{item}. A \emph{bundle} of goods, typically denoted by $\xb, \yb$ or $\zb$ in this paper, is a vector in $\Z^n$ whose $i$-th entry denotes the number of items of good $i$. The \emph{target bundle} $\tb$ is a bundle the auctioneer wants to allocate amongst
the bidders. Vectors $\pb, \qb \in \mathbb{R}^n$ typically denote vectors of
prices, with a price entry for each of the $n$ goods. We write $\pb \leq \qb$
when the inequality holds component-wise. It is often convenient to regard a
rejected bid as being accepted on a notional \emph{reject good} for which bids
and prices are always zero. Letting the reject good be 0, the set of goods is
then $[n]_0$. In this setting, we identify bundles and prices with the
$n+1$-dimensional vectors obtained by adding a $0$-th entry of value 0.

A \emph{valuation} $u$ is a function that maps (non-negative) bundles to non-negative real numbers. The valuations we consider are concave-extensible and have bounded effective domain.%
\footnote{\label{footnote:convex-extensible}A function is concave-extensible if its concave closure agrees with the function at every integral point. See \citep{ST15} for details, and how this assumption is natural for valuation functions.}
Bidders have \textit{quasilinear utilities}, i.e.~the utility derived from bundle $\xb$ at price $\pb$ by a bidder with valuation $u$ is given by
\begin{equation}\label{eq:utility}
	u(\xb) - \pb \cdot \xb.
\end{equation}
Any valuation $u$ is associated with a demand correspondence $D_u$ that maps prices $\pb$ to the set of bundles that maximise \eqref{eq:utility}. We omit the subscript $u$ when it is clear from context.

For any subset $X \subseteq [n]$, $\eb^X$ denotes the characteristic vector of $X$, i.e.~an $n$-dimensional vector whose $i$-th entry is 1 if $i \in X$ and 0 otherwise. Furthermore, $\eb^i$ denotes the vector whose $i$-th entry is~1 and other entries are~0.

A set function $f : 2^{[n]} \to \Z$ is \emph{submodular} if it satisfies $f(S) + f(T) \geq f(S \cup T) + f(S \cap T)$ for all $S,T \subseteq V$. Submodular function minimisation (SFM) is the task of finding a minimiser of a submodular function. It is well-known that the minimisers of submodular functions form a lattice; that is, if $S$ and $T$ are minimisers of $f$, then so are $S \cup T$ and $S \cap T$. SFM can be solved in polynomial time \citep{Iwata2008,Jiang2021, MCCormick2005}. For SFM that is efficient in practice, we refer to two SFM algorithms from the literature. The subgradient descent approach by~\citet{CLSW} determines a minimiser in time $O(nF^3\gamma \log n)$, while the Fujishige-Wolfe algorithm described by~\citet{Chakrabarty2014} takes time $O(F^2(n^2 \gamma + n^3))$, where~$F$ is an upper bound on the absolute value of $f$ and $\gamma$ denotes the time it takes to query $f$. Experimental results \citep{Chakrabarty2014} indicate that the running time of the Fujishige-Wolfe algorithm depends less on $F$ than suggested by the above bound.
For the price-finding algorithms described in Section~\ref{appendix:price-finding}, we require a subroutine that finds the \emph{inclusion-wise minimal minimiser}. Note that such a subroutine can be obtained by calling any SFM algorithm $n+1$ times. Indeed, let $S$ be a minimiser of $f$ and, for every $v \in S$, let $S_v$ be a submodular minimiser of the function~$f$ restricted to $[n] \setminus \{ v \}$. Then if $S_0$ denotes the minimal minimiser of $f$, we have $v \in S_0$ if and only if $f(S_v) > f(S)$, as the minimisers of a submodular function form a lattice. Hence, we obtain $S_0 \coloneqq \{v \in S \mid f(S_v) > f(S) \}$.

\subsection{Strong-substitutes valuation functions}
\label{sec:ss}
\label{sec:strong-substitutes}
In this paper, bidders have \emph{strong-substitutes} valuations. We review some basic properties of strong-substitutes (SS) valuations. A SS valuation $u$ divides price space into regions corresponding to bundles: any bundle $\xb$ has a region of prices in which $\xb$ is demanded, possibly along with other bundles (see Figure~\ref{fig:ex1}). It is known \citep[Theorem 11.16]{murota-book} that each such region is a convex lattice. When a region for some bundle $\xb$ has full dimensionality, then $\xb$ is uniquely demanded at any prices in the region's interior; we call this interior a \emph{unique demand region} (UDR).
Moreover, the \textit{Locus of Indifference Prices (LIP)} $\mathcal{L}_u$ comprising the set of prices at which demand is not unique consists of the union of $(n-1)$-dimensional linear pieces, or \textit{facets}, that separate the UDRs. In Figure~\ref{fig:ex1}, the LIPs associated with the demand correspondences of Alice and Bob in Examples~\ref{example:Alice} and \ref{example:Bob} are illustrated with dashed lines. We refer to Appendix~\ref{sec:missing-proofs} for further details on the LIP, including how every facet of the LIP can be endowed with a weight to capture the change in demand between any two price vectors.

\begin{definition}\label{def:subs}
A valuation $u$ is \emph{gross substitutes} if, for any prices $\pb' \geq \pb$ and any $\xb\in D_u(\pb)$, there exists $\xb'\in D_u(\pb')$ such that $x'_k\geq x_k$ for all $k$ such that $p_k=p'_k$. A valuation $u$ is \emph{strong substitutes (SS)} if, when we consider every unit of every good to be a separate good, $u$ is gross substitutes.%
\footnote{Formally, we can associate every multi-unit valuation function $u$ on a bounded effective domain with a single-unit valuation function $\hat{u}: \{0,1\}^N \to \R$ as follows. Let $U$ be an upper bound on the entries of the bundles in the domain of~$u$. We define the domain $\hat{u}$ as $N = \{ (i,\beta) \mid i \in [n], \beta \in [U]\}$, and let $\hat{u}(\hat{\xb}) = u(\xb)$, where $\xb$ is given by $x_i = \sum_{\beta=1}^U \hat{x}_{i,\beta}$. Valuation $u$ is strong substitutes if $\hat{u}$ is gross substitutes. See \citep{ST15} for details.}
\end{definition}

Importantly, {\citet[Theorem 4.1]{ST15}} establish that strong-substitutes is also equivalent to $M^\natural$-concavity {\citep{Murota.Shioura1999}} in our setting of valuations with bounded non-negative domains. (Both strong substitutes and $M^\natural$-concavity imply concave-extensibility, see {\citep[Theorem~9]{MS09}} and {\citep[Theorem~6.42]{murota-book} respectively.)}

While GS guarantees that a \Walrasian{} equilibrium exists in single-unit auction markets, the condition is not sufficient for the existence of such an equilibrium in the multi-unit case (\citet{ST15} give an example). SS represents a natural generalisation of single-unit GS that provides a general sufficient condition for the existence of an equilibrium. (See \citep{BK} and \citep{ST15} for discussion of the distinction between GS and~SS.)

\begin{definition}\label{def:iuf}
The {\em indirect utility function} $f_u$ of valuation $u$ maps a price vector
$\pb$ to the utility that a bidder with demand $D_u$ has for receiving her
preferred bundle at a given price vector~$\pb$ in the following way.
\begin{equation}
\label{eq:iuf}
    f_u(\pb) \coloneqq  \max_{\xb} (u(\xb) - \pb \cdot \xb) = u(\xb') - \pb \cdot \xb' \text{ with } \xb' \in D_u(\pb).
\end{equation}

\end{definition}
We note that $f_u(\pb)$ is convex, piecewise-linear and continuous for SS valuations $u$.

\subsection{Representing strong-substitutes valuations with weighted bids}
\label{sec:rep}
\label{sec:bids}
We describe how every strong-substitutes valuation function can be represented by a finite list $\bids$ of positive and negative bids. A \emph{bid} $\bid$ consists of an $(n+1)$-dimensional integral vector $(b_1, \ldots, b_n; b_{n+1})$, where the first $n$ components $b_1, \ldots, b_n$ of the vector denote the bid's `valuation' for the goods $[n]$ and the $(n+1)$-th component $b_{n+1}$ is interpreted as the bid's \textit{weight} (the quantity of goods sought by the bid, which may be positive or negative). For the reader's convenience, we also define the alternative notation $w(\bid) \coloneqq b_{n+1}$ to denote $\bid$'s weight. Moreover, we restrict ourselves to positive and negative unit weights $\pm 1$. This is without loss of generality, as any bid with a weight of $w(\bid) \in \Z$ can be represented by $|w(\bid)|$ unit bids with the same vector and of the same sign. Finally, when working with the notional reject good $0$ as described in Section~\ref{sec:prelims}, we identify the bid $\bid$ with the $(n+2)$-dimensional vector $(0, b_1, \ldots, b_n; b_{n+1})$ obtained by prepending a $0$-th entry of value 0 (and thus indexing from 0 instead of 1). This allows us to define a demand correspondence $D_\bids$ that maps each price vector $\pb$ to a set of bundles demanded at $\pb$, and an indirect utility function $f_\bids$ associated with $\bids$.

A bid $\bid$ \emph{demands good $i \in [n]_0$ at $\pb$} if the surplus $b_i - p_i$ at price $\pb$ is maximal, that is if we have $i \in \argmax_{i \in [n]_0} (b_i - p_i)$; recall that good 0 is the notional `reject' good and $b_0 = p_0 = 0$ by definition. We say that $\argmax_{i \in [n]_0} {(b_i - p_i)}$ is the set of \emph{demanded goods of $\bid$ at $\pb$}. A bid $\bid$ is \emph{marginal} (on the set of its demanded goods) at $\pb$ if $\bid$ demands more than one good at $\pb$. It is marginal \emph{on good $i$} at $\pb$ if it is marginal at $\pb$ and $i$ is in its set of demanded goods. Moreover, we say that price $\pb$ is marginal with respect to a given bid list $\bids$ if there are bids in $\bids$ that are marginal at~$\pb$, and non-marginal otherwise.
We illustrate these definitions using Example~\ref{example:Alice} (Figure~\ref{fig:Alice}). Recall that Alice's bid list $\bids$ consists of two positively-weighted bids $\bid \coloneqq (6,6;1)$ and $\bid' \coloneqq (0,4;1)$. At prices $\pb = (2,4)$, the bid $\bid$ is not marginal as it demands only good $1$ while $\bid'$ is marginal between the reject good and good $2$. At prices $\pb = (6,6)$, the bid $\bid$ is marginal between the reject good, good $1$, and $2$ while $\bid'$ is rejected (it only demands the reject good). Finally, at prices $\pb = (6,2)$, both bids are non-marginal, as they each demand only good 2. More generally, in Figure~{\ref{fig:Alice}}, all prices along the dashed lines are marginal, and all other prices are non-marginal. Note that these marginal prices coincide with the LIP $\mathcal{L}_u$, i.e.~the set of prices at which demand is not unique, of Alice's underlying valuation $u$. 

For any bid list $\bids$, we define the \emph{demand correspondence} $D_{\bids}(\pb)$ at prices $\pb$ as follows, distinguishing between the cases that $\pb$ is marginal and non-marginal. If $\pb$ is non-marginal, every bid $\bid \in \bids$ uniquely demands some good $i(\bid)$. In this case, $D_{\bids}(\pb)$ is a singleton set consisting of the bundle $\xb$ that is obtained by adding up an amount $w(\bid)$ of good $i(\bid)$ for each $\bid \in \bids$, i.e.~$D_{\bids}(\pb) = \{ \xb \}$ with $\xb \coloneqq \sum_{\bid \in \bids} w(\bid)\eb^{i(\bid)}$. In Example~\ref{example:Alice}, the demand correspondence at non-marginal prices $\pb = (6,2)$ is $D_{\bids}(6,2) = \{ (0,2) \}$, as both bids demand good 2. At $\pb = (1,2)$, the demand correspondence is $D_{\bids}(1,2) = \{ (1,1) \}$, as $\bid$ uniquely demands good 2 and $\bid'$ uniquely demands good 1. Figure~\ref{fig:Alice} shows the demand in the different polyhedron interiors corresponding to non-marginal prices.

If $\pb$ is marginal, $D_\bids(\pb)$ consists of the discrete convex hull of the bundles demanded at non-marginal prices arbitrarily close to $\pb$, where the discrete convex hull of a set of bundles $X$ is defined as conv$(X) \cap \Z^n$. In Figure~\ref{fig:Alice}, we see that bundles demanded in the local neighbourhood of $\pb = (0,4)$ are $(1,0)$ and $(1,1)$, which implies $D_\bids(\pb) = \{(1,0), (1,1) \}$. At $\pb = (4,4)$, we have $D_\bids(\pb) = \{ (1,0), (0,1), (1,1), (0,2) \}$.

For any bid list $\bids$, we can define the \emph{indirect utility function}
\begin{equation}
\label{eq:bid-utility-function}
	f_\bids(\pb) = \sum_{\bid \in \bids} w(\bid) \max_{i \in [n]_0} (b_i -
p_i).
\end{equation}

From {\eqref{eq:bid-utility-function}} it is clear that we can compute $f_\bids(\pb)$ in linear time in the number of goods and the size of the bid list. In our setting, we can also efficiently compute a bundle demanded at a given price~$\pb$. This \emph{demand oracle problem} is noted in~\citep{PLsurvey} as an algorithmic primitive needed to implement the Walrasian {t\^{a}tonnement} procedure. If $\pb$ is non-marginal, simply allocate each bid~$\bid$ a positive or negative item of the good $i(\bid)$ it uniquely demands and add up these items to obtain the demanded bundle $\xb = \sum_{\bid \in \bids} w(\bid)\eb^{i(\bid)}$.
However, if $\pb$ is marginal and the bid list contains both positive and negative bids, we emphasise that independently allocating to each bid one of the goods it demands may not result in a bundle that lies in $D_\bids(\pb)$. In Example~\ref{example:Bob} (depicted in Figure~\ref{fig:Bob}), for instance, we see that at prices $\pb = (4,4)$, allocating a negative item of good 1 to bid $(4,4;-1)$, a positive item of good 1 to bid $(4,2;1)$ and a positive item of good 2 each to bids $(2,4;1)$ and $(6,6;1)$ leads to a total bundle of $(0,2) \not \in D_\bids(\pb)$.
Instead, care must be taken to accept bids in a consistent way; one way to do this is to perturb entries of the price vector slightly so as to break ties, then note that the resulting bundle is demanded at the unperturbed prices.

Concurrent work by {\citet{BK-OS}} (see also Footnote~{\ref{footnote:expressivity}}) shows that any SS demand correspondence can be represented as a finite list of positive and negative bids, and this representation is essentially unique (up to redundancies). Conversely, however, not all lists of positive and negative bids induce a valuation function. We call a bid list \emph{valid} if the indirect utility function $f_\bids$ defined in \eqref{eq:bid-utility-function} is convex; Theorem~\ref{thm:validity} gives two further equivalent characterisations of validity.%
\footnote{Equivalent conditions for validity are provided in concurrent work \citep[Proposition 4.14]{BK-OS}, in a more general setting.  There, also, it is shown that condition \eqref{eq:valid-cover} is equivalent to the existence of a strong-substitutes valuation $u$ satisfying $D_\bids=D_u$ (that is, the second part of condition \eqref{eq:valid-SS}).  \citep{BK-OS} does not work with the indirect utility function, but does additionally show that validity is equivalent to $D_\bids$ satisfying the ``law of demand'': an increase in the price of only one good, between prices at which demand is unique, leads to a weak decrease in demand for the good whose price has changed, with equality in this demand if and only if the full demand vector is unchanged.
\citet{baldwin2022strong} additionally show that, when $\bids$ is valid, the demand $D_\bids$ can be given as the Minkowski difference between the demand associated with the set of positive bids in $\bids$ and the demand associated with the set of negative bids in $\bids$. (The work of \citet{LinTran2017}, mentioned in Footnote~\ref{footnote:expressivity}, does not consider characterisations of validity as, in their framework, invalid combinations do not arise.)}  
In particular, it establishes that the indirect utility function and demand correspondence of a valid bid list as defined in Section~{\ref{sec:bids}} coincide with the indirect utility function and demand correspondence for some strong-substitutes valuation as defined in Sections~{\ref{sec:prelims}} and~{\ref{sec:strong-substitutes}}. In Section~\ref{sec:main}, we will also introduce a weaker local notion of validity that holds in an open, convex neighbourhood of a price~$\pb$.

\begin{theorem}
\label{thm:validity}
Let $\bids$ be a list of positive and negative bids and $f_\bids$ be its associated indirect utility function as defined in \eqref{eq:bid-utility-function}. The following conditions are equivalent.
\begin{enumerate}
	\item \label{eq:valid-convexity} $f_{\bids}$ is convex.
	\item \label{eq:valid-cover} For every price vector $\vec{p}$ and pair of distinct goods $i,i' \in [n]_0$ it holds that $\sum_{\bid\in\bids_{ii'}} w(\bid)\geq 0$, where $\bids_{ii'}\subseteq\bids$ is the subset of bids marginal on both $i$ and $i'$ at $\vec{p}$.
    \item \label{eq:valid-SS} There exists a strong-substitutes valuation $u$ satisfying $f_{\bids} = f_u$ and $D_{\bids} = D_u$.
\end{enumerate}
\end{theorem}
The equivalence of \eqref{eq:valid-convexity} and \eqref{eq:valid-cover} is a special case of Proposition \ref{prop:local-validity} in Section \ref{sec:main}, which is proved in Appendix \ref{appendix:local-validity}. It is well-known that \eqref{eq:valid-SS} implies \eqref{eq:valid-convexity}. To show that \eqref{eq:valid-cover} implies \eqref{eq:valid-SS}, we use results of \citet{Mikhalkin} to verify the existence of a suitable valuation $u$, and  check that $u$ indeed has the properties required. We present the full proof of Theorem~\ref{thm:validity} in Appendix~\ref{appendix:validity}.

\subsection{Simulating a valuation oracle}
\label{sec:simulating-valuation-oracle}
A bid list $\bids$ does not directly provide the value $u(\xb)$ that the bidder  has for a given bundle $\xb$. However, the ability to evaluate $u(\xb)$ for arbitrary bundles is a prerequisite for the allocation algorithms in \citep{Murota2003} and \citep{PLW}, as mentioned in the related work section. An efficient way to determine $u(\xb)$ is to compute prices $\pb$ at which the bidder demands~$\xb$ using the price-finding algorithm from Section~{\ref{sec:price-finding-overview}}; then computing the bidder's indirect utility $f_u(\pb) = f_\bids(\pb)$ (cf.~Definitions~\ref{def:iuf} and~\ref{eq:bid-utility-function}) at these prices; and finally recovering the bundle value $u(\xb) = f_{\bids}(\pb) + \pb \cdot \xb$. The first step, finding the prices~$\pb$, takes time $O(n^2|\bids|^3 + n |\bids|T(n))$ with the demand change method (and $O(n^2|\bids|^2 \log M + n |\bids|T(n))$ with the binary search method) for determining the step lengths of the steepest descent method described in Section~\ref{sec:price-finding-overview} and Appendix~\ref{appendix:price-finding}. Here $n$ is the number of goods, $M$ is an upper bound on the bid entries and $T(n)$ is the time it takes to minimise an $n$-dimensional submodular set function. The second and third steps,  computing $f_{\bids}(\pb)$ and solving for $u(\xb)$, take time $O(n|\bids|)$, so the first step dominates the computation time.

\subsection{Deciding validity of bid lists}
\label{subsec:coNP}
In our auction, bidders submit their lists of positive and negative bids to the auctioneer prior to the auction. We show in Theorem~\ref{thm:np} that checking the validity of a given list of bids is computationally hard. The proof, given in Appendix~\ref{sec:testing}, exploits definition \eqref{eq:valid-cover} of validity in Theorem~\ref{thm:validity}. However, in Appendix~\ref{sec:testing} we provide a straightforward algorithm to verify the validity of a given bid list in polynomial time if the number of goods, or the number of negative bids, is bounded by a constant.

\begin{theorem}\label{thm:np}
	Deciding the validity of a given list of positive and
negative bids is \conp-complete.
\end{theorem}

\begin{theorem}
\label{thm:polynomial-checking}
	Let $\bids$ be a list of positive and negative bids. If the number of goods or the number of negative bids is bounded by a constant, there exists a polynomial-time algorithm for checking the validity of $\bids$.
\end{theorem}

Furthermore, many useful subclasses of bid lists (of arbitrary size) are easily checkable for validity in practice. Hence, while the \conp-completeness result is disappointing, it does not seem to be an important limitation on the auction, as sensible restrictions on the permitted bids can allow us to check validity efficiently in practice.

\subsection{The computational challenges}
\label{sec:stmt}
In this subsection we state the computational problems to be solved. For any bidder $j \in J$, where~$J$ is the set of bidders, $\bids^j$ denotes the bids of bidder~$j$. We assume each $\bids^j$ is \emph{valid}, as defined in Section~\ref{sec:rep}, and provided as a list of vectors encoded in binary. Let $\bids$ be the aggregate list of bids obtained by concatenating the lists of all bidders. Theorem~{\ref{thm:validity}} part {\eqref{eq:valid-cover}} implies that $\bids$ is valid. Figure~\ref{fig:aggregate} depicts the aggregate bid list of Alice and Bob from Figure~\ref{fig:ex1}. The running times of our algorithms are given in terms of $n$, $|J|$, $|\bids|$ and $M \coloneqq \max_{\bid \in \bids} \| \bid \|_\infty$, the maximum bid vector~entry.

Suppose the auctioneer intends to sell target bundle $\tb$. Our aim is to
compute a \Walrasian{} equilibrium: a market-clearing price $\pb$ at
which $\tb$ is demanded and an allocation of $\tb$ to the various bidders so
that every bidder receives a bundle they demand at $\pb$. In the event
that not enough bids are made for the target bundle, some items can
go unsold, which is equivalent to the auctioneer buying them back from the
market. To reflect this, the auctioneer places a total of $\| \tb \|_1 + 1$
positive bids at $\mathbf{0}$.\footnote{Of course, any positive reserve price
can be reflected by appropriately locating the auctioneer's bids.} The computation of a \Walrasian{} equilibrium separates into two problems.

\begin{description}
\item[\textbf{The price-finding problem.}] Given the aggregate bid list $\bids$ and
target bundle $\tb$, find the coordinate-wise minimal \emph{market-clearing}
price $\pb$ at which $\tb$ is demanded, that is $\tb \in D_\bids(\pb)$.
\item[\textbf{The allocation problem.}] Given a valid bid list $\bids^j$ for each
bidder $j \in J$, a target bundle $\tb$ and market-clearing price $\pb$,
allocate $\tb$ to the bidders so that each bidder receives a bundle they demand
at $\pb$, that is, a partition $(\tb^j)_{j \in J}$ of $\tb$ with $\tb^j \in
D_{\bids^j}(\pb)$ for all $j \in J$.
\end{description}
Note that we do not ask for a breakdown of which of $j$'s bids are accepted on which goods. Indeed, we show that a solution to the allocation problem can be obtained without this knowledge.

\begin{theorem}
\label{thm:exists}
Suppose that each bidder $j \in J$ has SS demand correspondence $D^j$.
Given a target allocation $\tb$ and a
market-clearing price $\pb$ for $\tb$, there exists a partition
$\tb=\sum_{j \in J} \tb^j$ such that $\tb^j \in D^j(\pb)$.
\end{theorem}

Theorem~\ref{thm:exists} (cf.~\citep{DKM01,murota-book,MS09}) ensures that if the
target bundle $\tb$ is aggregately
demanded at market-clearing price~$\pb$, then there exists an allocation of
$\tb$ among the bidders so that every bidder receives a bundle she demands.
However, care must be taken in finding such an allocation, as bidders cannot
simply be allocated an arbitrary bundle they demand at the market-clearing
price. (Indeed, in Figure~\ref{fig:ex1} the lowest prices at which Alice and
Bob aggregately demand $(1,1)$ is given by $\pb = (4,4)$. If we allocate Alice
the whole bundle $(1,1)$, Bob does not demand the empty remainder $(0,0)$ at
$\pb$.)

If we had (oracle) access to the valuations of the bidders, we could perturb the valuations such that the target bundle is uniquely demanded at some price by subtracting carefully constructed functions that are discrete-convex in one variable.\footnote{See \citep[Theorem 6.13 (4)]{murota-book}.} However, it is not clear how to perturb the bids individually to achieve a suitable `indirect' perturbation of the valuations. Instead, our allocation algorithm takes the approach of perturbing bids bidder by bidder.
Some intuition can be gained by considering the addition of a small `random' perturbation vector $\vec{v}^j$ to all the bids in bid set $\bids^j$ (a different $\vec{v}^j$ for each $j \in J$) and recomputing the new market-clearing prices. This has the effect of breaking ties between bids in two different bid lists. (This also uses the plausible fact, established in Proposition~\ref{prop:shift-valid}, that small perturbations make only small changes to the market clearing price(s).) However, the perturbations need to be exponentially small in the number of goods in order to ensure that all possible ties are broken. (One can check, noting the pigeonhole principle, that if perturbations are multiples of some inverse-polynomial, then they will have distinct subsets having the same sum, losing the guarantee that all possible ties are broken.) Our algorithm of Section~\ref{sec:main} avoids this by systematically perturbing and un-perturbing bid sets, making decisions on marginal bids as it proceeds.

\section{Finding market-clearing prices}
\label{sec:price-finding-overview}

\begin{figure}[t!]
	\centering
	\begin{subfigure}[t]{0.49\textwidth}
		\centering
		\includegraphics{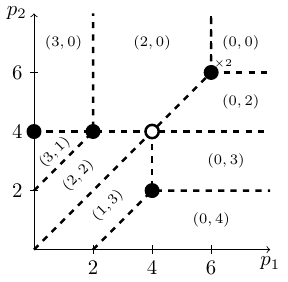}
		\caption{Aggregate demand of Alice and Bob}
		\label{fig:aggregate}
	\end{subfigure}
	\begin{subfigure}[t]{0.49\textwidth}
		\centering
		\includegraphics{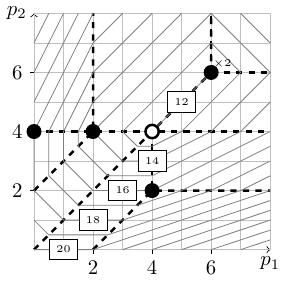}
		\caption{Lyapunov function $g$ for target bundle $\tb = (1,1)$}
		\label{fig:lyapunov}
	\end{subfigure}
	\caption{Figure (a) shows the demand of the aggregate bid list obtained by aggregating Alice's and Bob's bids from Figure~\ref{fig:ex1}, while (b) depicts a contour plot of the Lyapunov function $g$ for target bundle $\tb = (1,1)$. Note that $\pb = (4,4)$ is the minimal market-clearing price for $\tb$. Starting at $\mathbf{0}$, \textsc{MinUp} repeatedly moves in direction $d = (1,1)$ until it reaches $\pb$, whereas \textsc{LongStepMinUp} makes a single long step from $\mathbf{0}$ to $\pb$.}
	\label{fig:price-finding}
\end{figure}

In order to determine component-wise minimal prices at which the market is cleared, we apply an iterative steepest-descent approach from the discrete convex optimisation literature {\citep{Shioura17}}. This approach generalises an algorithm used by Ausubel's ascending auction design~{\citep{Aus06}} and {\citet{GS00}} for the task of finding equilibrium prices in single-unit markets. {\Citet{Aus06}} and others define the \emph{Lyapunov function} with regard to indirect utility function $f_u$ for strong-substitutes valuation $u$ and target bundle $\tb$ as $g_{\tb}(\pb) \coloneqq f_u(\pb) + \tb \cdot \pb$.

A function $f:\Z^n \to \mathbb{R}$ is \textit{$L^\natural$-convex} if it satisfies the \emph{translation submodularity} property,
\begin{equation}\label{eq:translation-submodularity-main}
	f(\pb) + f(\qb) \geq f( (\pb - \alpha \mathbf{1}) \vee \qb) + f(\pb \wedge (\qb + \alpha \mathbf{1})) \quad (\forall \pb, \qb \in \Z^n_+, \forall \alpha \in \Z_+).
\end{equation}
Here, $\vee$ and $\wedge$ denote the component-wise maximum and minimum, respectively. It is well-known that the Lyapunov function $g_{\tb}$ is $L^\natural$-convex if the valuation is strong substitutes (see \citep{MSY2013, MSY2016}, or Proposition~\ref{prop:L-natural-convex} in Appendix~{\ref{appendix:price-finding}}) and thus has a unique component-wise minimal minimiser {\citep{Shioura17}}. Moreover, this minimiser is integral and corresponds to the component-wise minimal market clearing prices (cf.~Lemma~{\ref{lemma:lyapunov-equivalence}} in Appendix~{\ref{appendix:price-finding}}), and can be found using the following discrete steepest descent procedure. Given a starting point $\pb \coloneqq \pb^0$ that is dominated by the market-clearing prices~$\pb^*$, repeatedly find the component-wise minimal subset $S \subseteq [n]$ minimising $g_{\tb}(\pb + \eb^S) - g_{\tb}(\pb)$ and move a unit step in this direction by updating $\pb \coloneqq \pb + \eb^S$. The method terminates once a local minimum is reached. As a local minimum of an $L^\natural$-convex function is a global minimum (see, e.g., {\citep{Shioura17}}), we have found~$\pb^*$. It is straightforward that the $L^\natural$-convexity of the Lyapunov function as defined in \eqref{eq:translation-submodularity-main} guarantees the submodularity of function $g'(\pb;S) \coloneqq g_{\tb}(\pb + \eb^S) - g_{\tb}(\pb)$ with respect to argument $S$ (for any fixed~$\pb$). Hence we can apply standard submodular minimisation algorithms (cf.~Section~{\ref{sec:prelims}}) to find the component-wise minimal steepest descent direction at any prices $\pb$.

Adapting this method to our bidding-language setting, we can use our direct knowledge of the bids and \eqref{eq:bid-utility-function} to express $g_{\tb}$ for
any valid bid list $\bids$ as
\begin{equation}\label{eq:lyapunov-overview}
	g_{\tb}(\pb) \coloneqq f_{\bids}(\pb) + \tb \cdot \pb = \sum_{\bid \in \bids} w(\bid) \max_{i \in [n]_0} (b_i - p_i) + \tb \cdot \pb.
\end{equation}

Theorem~\ref{thm:validity} guarantees that the function $f_{\bids}$ defined in \eqref{eq:bid-utility-function} is the indirect utility function $f_u$ of some strong-substitutes valuation $u$, so \eqref{eq:lyapunov-overview} is a well-defined Lyapunov function. (Figure~\ref{fig:lyapunov} depicts a contour plot of the Lyapunov function defined for the aggregate demand of Alice and Bob as given in Examples~\ref{example:Alice} and~\ref{example:Bob} and shown in Figure~\ref{fig:aggregate}.) Moreover, \eqref{eq:lyapunov-overview}
implies that $g_{\tb}$ can be evaluated at any $\pb$ in time $O(n|\bids|)$.

In order to improve the running time of the steepest-descent method, \citet{Shioura17} proposes aggregating several steps in the same direction into a single `long' step. After finding the component-wise minimal direction $\eb^{S}$, we move $\lambda(\pb, S)$ unit steps in this direction at once by updating $\pb \coloneqq \pb + \lambda(\pb,S) \eb^S$, where $\lambda(\pb, S)$ is defined as
\begin{equation}
\label{eq:steplength-overview}
\lambda(\pb, S) \coloneqq \max \{ \lambda \in \Z_+ \mid g'(\pb;S) =
g'(\pb + (\lambda-1)\eb^{S};S) \}
\end{equation}
and $g'(\pb; S) \coloneqq g_{\tb}(\pb + \eb^S) - g_{\tb}(\pb)$.
\citet{Shioura17} shows that the number of long steps before reaching~$\pb^*$ is at most $n \max \{-g'(\bm{0};S) \mid S \subseteq [n] \}$. In our setting, we see that $-g'(\bm{0};S) \leq |\bids|$, and obtain an upper bound of $n |\bids|$ on the number of long steps that our steepest-descent method takes (see Appendix~\ref{appendix:price-finding}).

In Appendix~\ref{appendix:price-finding}, we describe two methods for computing the step length \eqref{eq:steplength-overview} at every iteration of the steepest-descent procedure, and analyse their complexity in our bidding-language setting. The first method uses binary search, as suggested by \citet{Shioura17}, to find each $\lambda(\pb, S)$; we show that determining each $\lambda(\pb, S)$ takes time $O(n |\bids| \log M)$. (Recall that $M = \max_{\bid \in \bids} \| \bid \|_\infty$.) Our second method for determining $\lambda(\pb, S)$ (the `demand change' method) directly exploits our knowledge of the individual bids to determine the step length in time $O(n|\bids|^2)$. We see that the best method in practice depends on the magnitudes of $\log M$ and $|\bids|$. Overall, both methods lead to price-finding algorithms that are polynomial in the size of the bid lists, as we state in Theorem~\ref{thm:long-step-analysis}.

\begin{theorem}
\label{thm:long-step-analysis}
The steepest descent approach takes time $O(n^2|\bids|^2\log M + n|\bids|T(n))$ and $O(n^2|\bids|^3 + n|\bids|T(n))$ to find the component-wise smallest market-clearing prices using the \emph{binary search} and \emph{demand change} long step methods, respectively.
\end{theorem}

The computation time of the steepest descent method, even with long steps, is dominated by the task of finding a component-wise minimal submodular minimiser. In practice, we can exploit our knowledge of the bids to decrease the dimensionality of the submodular function to be minimised, which leads to performance improvements. For full details on these practical improvements, and on the `unit-step' and `long-step' steepest-descent methods for computing market-clearing prices in the bidding-language setting, we refer to Appendix~\ref{appendix:price-finding}.

\section{Allocations to the separate bidders}
\label{sec:main}
Suppose we are given a market-clearing price $\pb$ for target bundle $\tb$. We now present an algorithm that solves the allocation problem, i.e.~finds a partition $(\tb^j)_{j \in J}$ of the target bundle $\tb$ such that $\tb^j$ is demanded by bidder $j$ at price $\pb$. We note that while the market-clearing price $\pb$ returned by our steepest descent approach in Section~\ref{sec:price-finding-overview} is component-wise minimal, our allocation algorithm works for any integral market-clearing price.

\begin{algorithm}[htb]
  \caption{\textsc{Allocate}}
  \label{alg:allocate}
  \begin{algorithmic}[1]
    \State \textbf{Input:} Initial allocation problem $\mathcal{A}$.
    \State \textbf{Output:} Target bundle allocation.
    \State Allocate all non-marginal bids by applying \textsc{NonMarginals} (Algorithm~\ref{alg:proc1}).
    \While {$\mathcal{A}$ has a non-empty bid list}
      \State Run the \textsc{FindParams} subroutine (Algorithm~\ref{alg:findparams}).
      \If {\textsc{FindParams} returns a \keylist{} $(I,j^*)$ with at most one \linkgood{} $i^*$}
        \State Process unambiguous marginal bids with \textsc{UnambiguousMarginals} (Algorithm~\ref{alg:proc2}).
      \Else
        \State Simplify the allocation problem with \textsc{ShiftProjectUnshift} (Algorithm~\ref{alg:proc3}).
      \EndIf
      \State Allocate all (newly) non-marginal bids with \textsc{NonMarginals} (Algorithm~\ref{alg:proc1}).
    \EndWhile
    \State \Return the partial allocation bundle $\mb^j$ for each bidder $j \in J$.
  \end{algorithmic}
\end{algorithm}

Our algorithm \textsc{Allocate}, stated in Algorithm~\ref{alg:allocate}, repeatedly simplifies the problem until it becomes vacuous. It generates a sequence of \emph{allocation problems} by iteratively allocating parts of the target bundle to bidders and removing satisfied bids from the bid lists until the residual supply bundle (initially the target bundle) is empty. In order to capture this formally, we define a generalised allocation problem in Section~\ref{subsec:generalised-allocation-problem}; the definition introduces a weaker notion of local validity for bid lists and also features a partial allocation bundle for each bidder as well as a residual target bundle. With every allocation problem we associate a corresponding \emph{marginal bids multigraph} and \emph{derived graph}. These graphs, and key terminology such as \keylists{} and \linkgoods{}, are defined in Section~\ref{subsec:graphs}. The marginal bids multigraph is used to quantify progress and establish the running time of the algorithm, while the derived graph is required for \textsc{Allocate} to decide whether to apply the \textsc{UnambiguousMarginals} or \textsc{ShiftProjectUnshift} procedure.

The algorithm starts by removing all non-marginal bids and transferring their demanded items from the residual supply to the partial allocation bundle of the bid's owner. This is performed by the procedure \textsc{NonMarginals}, which is defined in Section~\ref{sec:unambiguous}. Next, the algorithm calls the \textsc{FindParams} subroutine (defined in Section~\ref{subsec:graphs}), which constructs and works on the derived graph. If \textsc{FindParams} identifies a subset of marginal bids that can be allocated unambiguously, \textsc{Allocate} invokes \textsc{UnambiguousMarginals} (defined in Section~\ref{sec:unambiguous}) to process these bids. Otherwise \textsc{FindParams} identifies parameters that \textsc{Allocate} uses to call the \textsc{ShiftProjectUnshift} procedure. \textsc{ShiftProjectUnshift} is defined in Section~\ref{sec:shift}; it does not allocate items and delete bids but instead simplifies the allocation problem by shifting and projecting the bids in order to strictly reduce the demand ties between bids. Finally, \textsc{NonMarginals} is invoked again to process any bids that may have become non-marginal. The algorithm repeats this process until the allocation problem become vacuous in the sense that all bid lists are empty, at which point the partial allocation bundles constitute a valid allocation of the target bundle $\tb$ at price vector $\pb$. Section~\ref{subsec:allocate} gives the running time and a proof of correctness for \textsc{Allocate}.

\subsection{The generalised allocation problem}
\label{subsec:generalised-allocation-problem}
We now introduce allocation problems formally and define the terminology and technical tools required to reason about how \textsc{Allocate} successively reduces the initial allocation problem using the three procedures \textsc{NonMarginals}, \textsc{UnambiguousMarginals} and \textsc{ShiftProjectUnshift}. In general, we note that applying these procedures may result in bid lists that are no longer valid in the sense of Theorem~\ref{thm:validity}. (Figure~\ref{fig:ex3} in Section~\ref{sec:unambiguous} gives an example of how applying \textsc{NonMarginals} can result in an invalid bid list.) Instead, we introduce the notion of local validity that requires bid lists to be valid only in the neighbourhood of the market-clearing price $\pb$. For prices $\pb$ in this neighbourhood, we can define the demand correspondence $D_\bids(\pb)$ in the same way as for globally valid bid lists. Let $B(\pb, \varepsilon) \coloneqq \left \{ \yb \in \mathbb{R}^n \mid \|\yb-\pb\|_\infty < \varepsilon \right \}$ denote the open $\varepsilon$-ball with regard to the $L_\infty$-norm.

\begin{definition}\label{def:local-validity}
A list of bids $\bids$ is \textit{$\varepsilon$-valid at $\pb$} if the indirect utility function $f_{\bids}$ restricted to $B(\pb;\varepsilon)$ is convex. The bid list is \textit{locally valid at $\pb$} if there exists some $\varepsilon > 0$ for which it is $\varepsilon$-valid.
\end{definition}

This generalises the notion of valid bid lists introduced in Section~{\ref{sec:bids}}, and we correspondingly generalise Theorem~{\ref{thm:validity}} with the following proposition (proved in Appendix~{\ref{appendix:local-validity}}).

\begin{proposition}\label{prop:local-validity}
Let $\bids$ be an $\varepsilon$-valid bid list at $\pb$, and let $P = B(\pb;\varepsilon)$. The following are equivalent.
\begin{enumerate}
	\item The indirect utility function $f_\bids$ restricted to $P$ is convex.
	\item There is no price vector $\pb \in P$ and pair of distinct goods $i,i' \in [n]_0$ at which the weights $w(\bid)$ of the bids $\bid \in \bids$ marginal on $i$ and $i'$ sum to a negative number.
\end{enumerate}
\end{proposition}

\begin{definition}\label{defn:allocation-prob}
An \emph{allocation problem} is a 4-tuple $\mathcal{A}=[\pb,(\bids^j)_{j\in
J},(\mb^j)_{j\in J},\rb]$, where
\begin{enumerate}
\item $\pb \in \R^{n+1}$ is a price vector with $p_0=0$,
\item for each bidder $j \in J$, $\bids^j$ is a list of bids in
$\{0\} \times \left(\Z+\{0,\frac{1}{10}\}\right)^{n} \times \Z$ that is locally valid at $\pb$,
\item $\mb^j\in\Z^{n+1}_+$ for each $j\in J$ is the partial allocation,
\item $\rb\in\Z^{n+1}_+$ is the residual supply,
\end{enumerate}
and there exists a valid allocation $(\rb^j)_{j \in J}$ of the residual supply
$\rb$ to bidders, i.e.~$\rb^j \in D_{\bids^j}(\pb)$ and $\sum_{j \in J} \rb^j =
\rb$. For any such valid allocation, we say that $\tb^j \coloneqq \rb^j + \mb^j$ is a
\emph{solution} to $\mathcal{A}$.
\end{definition}

Note that the bids need not be integral and the bid lists for the bidders are only required to be locally valid at $\pb$. The intuition behind this definition is to capture the action of successively allocating items of the target bundle $\tb$ to bidders, which may break global validity but preserves local validity at $\pb$. For each bidder $j \in J$, the partial allocation $\mb^j$ is a bundle denoting the subset of $\tb$ already allocated to $j$, while $\rb$ denotes the remaining part of $\tb$ not yet allocated to any
bidder. The $0$-th coordinates of $\mb^j$ and $\rb$ denote the number of reject goods (not yet) allocated.

\begin{definition}\label{defn:reduction}
	We say that $\mathcal{A}'$ is a \emph{valid reduction from $\mathcal{A}$} if
$\mathcal{A}'$ is an allocation problem and all solutions to $\mathcal{A}'$ are
also solutions to $\mathcal{A}$.
\end{definition}

In the \emph{initial problem}, the residual supply is given by the target
bundle and the partial allocation vectors are $\mathbf{0}$. Note that $t_0$
denotes the total number of rejected bids and can be computed as the total
weight of bids minus the total number of items in $\tb$. In the \emph{vacuous
problem}, the bid lists are empty and the residual supply is $\mathbf{0}$.

\paragraph{Technical tools.}
The following observations and lemma will be used in Section~\ref{sec:shift}
to prove that applying the procedure \textsc{ShiftProjectUnshift} to an allocation problem $\mathcal{A}$ produces an allocation problem that is a valid reduction from $\mathcal{A}$. We define the \emph{surplus gap} of a bid $\bid$ at $\pb$ as the
difference between the utility derived from a demanded good and the maximum
utility derived from a non-demanded good. This allows us to describe the goods
that $\bid$ demands when we perturb the price or the bid by a small amount.
Moreover, we see in Lemma~\ref{obs:validity-margin} that the size of the
neighbourhood around $\pb$ in which a bid list is valid can be lower bounded by
the surplus gaps of the bids in $\bids$.
Formally, let $I = \argmax_{i \in [n]_0} b_i - p_i$ denote the goods demanded
by $\bid$ at $\pb$. Then if $I \not = [n]_0$, we define the \emph{surplus gap}
of $\bid$ at $\pb$ as
\[
	\max_{i \in [n]_0} (b_i - p_i) - \max_{i \in [n]_0 \setminus I} (b_i - p_i).
\]

\begin{observation}
\label{obs:no-new-marginals}
	If bid $\bid$ demands goods $I$ at $\pb$ with a surplus gap of at
least $\varepsilon$, then for any price $\pb' \in B(\pb, \varepsilon / 2)$, the
goods $I'$ demanded by $\bid$ at $\pb'$ form a subset of $I$.
\end{observation}

\begin{observation}
\label{obs:no-new-marginals-2}
	If bid $\bid$ demands goods $I$ at $\pb$ with a surplus gap of at least $\varepsilon$, then for any bid $\bid' \in B(\bid, \varepsilon /4)$ and price $\pb' \in B(\pb, \varepsilon/4)$, the goods $I'$ demanded by $\bid'$ at $\pb'$ form a subset of $I$.
\end{observation}

\begin{lemma}
\label{obs:validity-margin}
	If $\bids$ is $\delta$-valid at $\pb$ for some $\delta > 0$ and
all bids in $\bids$ have a surplus gap of at least~$\varepsilon$ at $\pb$, then $\bids$ is $\varepsilon / 2$-valid at $\pb$.
\end{lemma}

We prove this lemma in Appendix~\ref{appendix:surplus-validity}.

\subsection{The marginal bids graph and the derived graph}
\label{subsec:graphs}
With every allocation problem~$\mathcal{A}$ we associate a marginal bids
multigraph and a derived graph. As mentioned above, the marginal bids multigraph allows us to provide upper bounds on the running time of our allocation algorithm. Moreover, in order to decide whether to apply \textsc{UnambiguousMarginals} or \textsc{ShiftProjectUnshift} and to determine the input parameters for these procedures, our \textsc{Allocate} algorithm calls the subroutine \textsc{FindParams}, defined in this section, which first constructs the derived graph and then attempts to find a maximal open or closed path in this graph.

\begin{definition}\label{def:mbg}
	The \emph{marginal bids graph} $G_\mathcal{A}$ associated with
allocation problem $\mathcal{A}$ is an undirected edge-labelled multigraph
whose vertices are the goods $[n]_0$ (including the `reject' good).
$G_\mathcal{A}$ has an edge $(i,i')$ labelled with $j$ if bidder $j$ has a bid
that is marginal at $\pb$ between $i$ and $i'$.
\end{definition}

For any bidder $j$, let $G_\mathcal{A}[j]$ denote the subgraph of $G_\mathcal{A}$ induced by all $j$-labelled edges; note that this graph is simple, as there is at most one edge labelled
$j$ connecting two goods $i$ and $i'$ in $G_\mathcal{A}$.
A pair $(I,j)$ with $I \subseteq [n]_0$
and $j \in J$ is a \emph{\keylist{}} if $I$ is the set of vertices of some
connected component of $G_\mathcal{A}[j]$. Intuitively, this captures the dependencies between the demands of (and thus the possible allocations to) bidder $j$'s bids. A vertex is a \emph{\linkgood{}} if its incident edges are labelled with more than one bidder. (Effectively, a \linkgood{} expresses a dependency between \keylists{} of different bidders.) We call a cycle in $G_\mathcal{A}$ a \emph{multi-bidder cycle} if it contains edges labelled by different bidders. A vertex in such a cycle is called a \emph{cycle-\linkgood{}} if its two edges in the cycle are labelled differently. Note that any multi-bidder cycle has at least two cycle-\linkgood{}s{}. Like \linkgoods{}, multi-bidder cycles and their cycle-\linkgoods{} capture the dependencies between demands of (and thus possible allocations to) different bidders.

\begin{definition}\label{def:ubergraph}\label{def:derived-graph}
The \emph{derived graph} $D_\mathcal{A}$ associated with allocation problem
$\mathcal{A}$ is a simple bipartite graph whose two disjoint and independent
vertex sets are the set of \linkgood{}s{} and the set of \keylists{} of
$G_\mathcal{A}$. There is an edge between \linkgood{} $i$ and \keylist{} $(I,j)$ if $i \in I$.
\end{definition}

Note that by the definition of \linkgood{}s{}, every \linkgood{} is adjacent to at
least two \keylists{} in the derived graph. This is consistent with our observation that a link good can be interpreted as a link between connected components of $G_\mathcal{A}[j]$ for different bidders $j$. Figure~\ref{fig:graphs-ex} gives an example of a marginal bids multigraph and its corresponding derived graph.

\begin{figure}[t!]
	\centering
	\begin{subfigure}[b]{0.4\textwidth}
	\centering
	\includegraphics{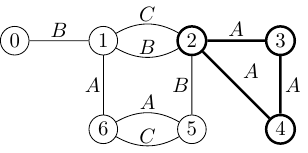}
	\caption{A marginal bids multigraph.}
	\label{fig:marginal-bids-graph}
	\end{subfigure}
	\begin{subfigure}[b]{0.59\textwidth}
	\centering
	\includegraphics{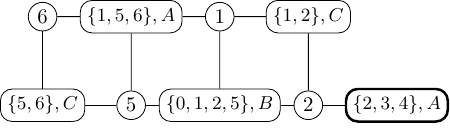}
	\caption{A derived graph.}
	\label{fig:derived-graph}
	\end{subfigure}
	\caption{Example of a marginal bids multigraph (a) and derived graph
(b) with three bidders $A,B,C$ and six goods $1,\ldots,6$ together with reject good $0$. Goods $1,2,5,6$ are \linkgood{}s{}. Goods $2,3,4$ form a \keylist{} $(\{2,3,4\},A)$ with one \linkgood{}, represented by a leaf \keylist{} in the derived graph (bolded in both graphs).}
	\label{fig:graphs-ex}
\end{figure}

We describe a procedure to construct the derived graph in near-linear time
$O(\alpha(n)n|\bids|)$ using a disjoint-union data structure. Here
$\alpha(\cdot)$ is the inverse Ackermann function, which grows extremely slowly
and is near-constant in our context, as $\alpha(n) \leq 4$ for any $n \leq
2^{2^{2^{2^{16}}}}$. The disjoint-union data structure
(cf.~\citep{TarjanvanLeeuwen1984}) maintains a representation of a set
partition and admits an $\alpha(n)$-time operation to merge two subsets of the
partition. Internally, it assigns a distinct label to each subset of the
partition and provides $\alpha(n)$-time access to the label of the subset in
which a given element lies.

First we show how to compute the \keylists{} for each bidder. Fix a bidder $j$
and initialise the disjoint-union data structure. For each bid $\bid \in
\bids^j$, compute the marginal goods $S$ at prices~$\pb$ and, fixing any $i \in
S$, merge the sets containing $i$ and $i'$ for all $i' \in S \setminus \{i\}$.
Now the data structure has learnt the vertex partition induced by the connected
components of $G_\mathcal{A}$ in time $O(\alpha(n) n |\bids^j|)$.
In order to recover this partition and express it as a family of sets, we initialise an empty family $\mathcal{K}$ of sets, where each set in $\mathcal{K}$ will have an associated label. For each good $i \in [n]$, determine the label $l$ of $i$'s subset in the data structure. If there already is a set in $\mathcal{K}$ with label $l$, add $i$ to this set. Otherwise, add the new singleton set $\{i\}$ with label $l$ to $\mathcal{K}$. Finally, iterate through $\mathcal{K}$ and delete all singleton sets. This takes time $O(\alpha(n) n)$. Now $\mathcal{K}$ represents the family of \keylists{} of bidder $j$. Hence in total it takes time $O(\alpha(n)n|\bids|)$ to compute the \keylists{} for all bidders.

In order to compute the \linkgood{}s{}, iterate through the \keylists{} of all bidders and count the number of times each good appears. If a good appears at least twice, it is a \linkgood{}. Once we know which good is a \linkgood{}, we can compute the edges of the derived graph: for each \keylist{} $(I,j)$, add an edge between $(I,j)$ and each \linkgood{} in $I$. Iterating through the \keylists{} of all bidders takes $O(nm) = O(n|\bids|)$ time.

\paragraph{The \textsc{FindParams} subroutine.}
In every iteration, the \textsc{Allocate} algorithm employs a subroutine to decide, using the derived graph, whether to invoke \textsc{UnambiguousMarginals} or \textsc{ShiftProjectUnshift}, and to compute the input parameters for the respective procedure. This subroutine \textsc{FindParams}, given in Algorithm~\ref{alg:findparams}, takes as input an allocation problem and returns one of three possible outputs: a \keylist{} $(I,j)$ with no \linkgood{}s{}, a \keylist{} $(I,j)$ with one \linkgood{} denoted~$i^*$, or a cycle-\linkgood{} $i^*$ and the label $j^*$ of one of its incident edges in the multi-bidder cycle. The \textsc{Allocate} algorithm invokes \textsc{UnambiguousMarginals} if \textsc{FindParams} returns a \keylist{} and \textsc{ShiftProjectUnshift} if it returns a cycle-\linkgood{} and edge~label.

\textsc{FindParams} works by walking through the derived graph $D_{\mathcal{A}}$ to find a maximal path. A \emph{path} in a graph is a sequence of distinct vertices $v_1, \ldots, v_k$ such that $v_i$ and $v_{i+1}$ are connected by an edge for all $i \in [k-1]$; it is \emph{maximal} in a given graph if we cannot extend the path with any vertex $v_{k+1}$ of the graph. A path is \emph{closed} if there is an edge between $v_1$ and $v_k$, and \emph{open} otherwise. (Hence a closed path is a cycle.) Note that a path in $D_{\mathcal{A}}$ alternates between link good and demand cluster vertices.

\begin{algorithm}[bht]
\caption{\textsc{FindParams}}
\label{alg:findparams}
\begin{algorithmic}[1]
  \State Compute the derived graph $D_\mathcal{A}$.
  \If {$D_\mathcal{A}$ contains an isolated \keylist{}}
    \State \Return this \keylist{} (without a \linkgood{}).
  \EndIf
  \State Starting from any \linkgood{}, take a walk in $D_\mathcal{A}$ to find a maximal path.
  \State Let $i^*$ and $(I,j^*)$ be the last \linkgood{}
and \keylist{} visited, respectively.
  \If {the path is open}
    \State \Return $(I,j^*)$ and $i^*$.
  \Else
    \State \Return $i^*$ and $j^*$.
  \EndIf
\end{algorithmic}
\end{algorithm}

\begin{lemma}
	In time $O(\alpha(n)n|\bids|)$, \textsc{FindParams} returns a \keylist{}
$(I,j^*)$ with no \linkgood{}s{}, a \keylist{} $(I,j^*)$ with one \linkgood{} $i^*$, or
a cycle-\linkgood{} $i^*$ and the label $j^*$ of one of its incident edges in the
multi-bidder cycle.
\end{lemma}
\begin{proof}
Note that any path in the derived graph $D_{\mathcal{A}}$ has length at most $2n$. Hence constructing the derived graph, which takes time $O(\alpha(n)n|\bids|)$, dominates the running time. By construction of the derived graph, an isolated \keylist{} in
$D_\mathcal{A}$ contains no \linkgood{}s{}, while a leaf \keylist{} contains exactly
one \linkgood{}. Hence if there is an isolated \keylist{}, the subroutine returns a
\keylist{} without a \linkgood{}.
If the path found is open, the last vertex must be a leaf \keylist{} and the subroutine returns a \keylist{} with single \linkgood{}. Now suppose the path is closed and consider only the cycle formed by alternating \linkgood{} and \keylist{} vertices. Firstly, note that there is a path between consecutive \linkgood{}s{} in the marginal bids graph $G_\mathcal{A}$ using only $j$-labelled edges. Secondly, consecutive \keylists{} have different bidder
labels. These two observations imply that $G_\mathcal{A}$ contains a simple
multi-bidder cycle with cycle-\linkgood{} $i^*$ and an incident edge labelled
with $j^*$.
\end{proof}

\subsection{Allocating unambiguous bids}
\label{sec:unambiguous}
\begin{figure}[t!]
	\centering
	\begin{subfigure}[b]{0.49\textwidth}
		\centering
		\includegraphics{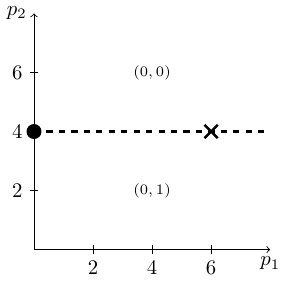}
		\caption{Alice's bid list after applying \textsc{NonMarginals}.}
		\label{fig:Alice-proc1}
	\end{subfigure}
	\begin{subfigure}[b]{0.49\textwidth}
		\centering
		\includegraphics{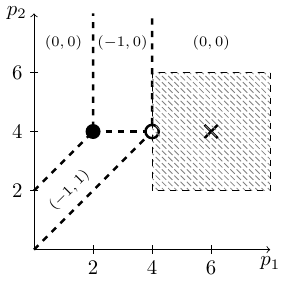}
		\caption{Bob's list after applying \textsc{NonMarginals}.}
		\label{fig:Bob-proc1}
	\end{subfigure}
	\caption{The resulting bid lists after applying \textsc{NonMarginals}
	(Algorithm~\ref{alg:proc1}) to
Alice's and Bob's bid lists from Figure~\ref{fig:ex1} at prices $\pb = (6,4)$
(marked by a cross). Note that Alice's bid list remains globally valid, whereas
Bob's bid list is only locally $2$-valid at $\pb$ (indicated by the hatched
square).}
	\label{fig:ex3}
\end{figure}

\paragraph{Non-marginal bids.}
Suppose bidder $j$ has a non-marginal bid $\bid$ at market-clearing
prices~$\pb$ that demands good $i \in [n]_0$ in the allocation problem
$\mathcal{A}$. Then this bid contributes exactly $w(\bid) \in \{-1,1\}$ items
of good $i$ to any solution of $\mathcal{A}$. Hence we can unambiguously
allocate these items to $\mb^j$ and remove them from the residual supply $\rb$,
thus accepting the non-marginal bid on the appropriate good.
\textsc{NonMarginals}, given in Algorithm~\ref{alg:proc1}, processes all
non-marginal bids in this way. Note that while this operation may not preserve
global validity of bid lists, the resulting lists remain locally valid at $\pb$,
so that the result is a valid allocation problem. Figure~\ref{fig:ex3} gives an
example. \textsc{Allocate} calls \textsc{NonMarginals} in every iteration in order
to process non-marginal bids that can be allocated unambiguously.

\begin{algorithm}[ht!]
\caption{\textsc{NonMarginals} (accept non-marginal bids)}
\label{alg:proc1}
\begin{algorithmic}[1]
\State \textbf{Input:} Allocation problem $\mathcal{A} = [\pb,
(\bids^j)_{j \in J}, (\mb^j)_{j \in J}, \rb]$.
\State \textbf{Output:} Reduced allocation problem $\mathcal{A}'$ without
non-marginal bids.
\ForAll{bidders $j \in J$}
	\ForAll{non-marginal bids $\bid \in \bids^j$}
  	\State Determine unique good $i$ demanded by $\bid$ and remove $\bid$ from $\bids^j$.
	  \State Increment $m^j_i$ by $w(\bid)$.
	  \State Decrement $r_i$ by $w(\bid)$.
	\EndFor
\EndFor
\end{algorithmic}
\end{algorithm}

\begin{lemma}\label{lem:obvious}
	Given a valid allocation problem $\mathcal{A}$, \textsc{NonMarginals} (Algorithm~\ref{alg:proc1})
outputs a reduction $\mathcal{A}'$ of $\mathcal{A}$ in linear time. Moreover,
we have $G_\mathcal{A} = G_{\mathcal{A}'}$.
\end{lemma}
\begin{proof}
	Let $\mathcal{A}' = [\pb, (\bids^j)'_{j \in J}, (\mb^j)'_{j \in J},
\rb']$ be the output of \textsc{NonMarginals}. First we show that~$\mathcal{A}'$
is a valid allocation problem. Fix some $j\in J$. Since $\bids^j$ is locally
valid by assumption, the indirect utility function $f_{\bids^j}$ is convex in
some small neighbourhood of $\pb$. Allocating a non-marginal bid to bidder $j$
corresponds to subtracting from the utility function $f_{\bids^j}(q)$ the term
$\max_{i \in [n]_0} (b_i - q_i)$. In a sufficiently small neighbourhood of
$\pb$, this term is linear, so the resulting utility function is also convex in
some open neighbourhood of $\pb$.
	It is straightforward to see that $(\tb^j)_{j \in J}$ is a solution to
$\mathcal{A}'$ if and only if it is a solution to $\mathcal{A}$,
so~$\mathcal{A}'$ is a reduction of~$\mathcal{A}$. To see that $G_\mathcal{A} =
G_{\mathcal{A}'}$, note that the marginal bids are unchanged.
\end{proof}

\paragraph{Unambiguous marginal bids.}
Let $\mathcal{A}$ be an allocation problem without non-marginal bids and suppose \textsc{FindParams} returns a \keylist{} $(I,j)$ with at most one \linkgood{}. \textsc{UnambiguousMarginals} is invoked by \textsc{Allocate} to unambiguously allocate the bids specified by this \keylist{}. By the definition of a \keylist{}, none of bidder $j$'s bids are marginal between items in both $I$ and $[n]_0 \setminus I$. Let $\bids^j_I$ denote the bids marginal on goods in $I$. All bids $\bid \in \bids^j_I$ contribute an item of a good in $I$ to bidder $j$'s final allocation bundle, while all other bids contribute an item of a good not in~$I$.

If the \keylist{} $(I,j)$ has no \linkgood{}s{}, none of the bids by bidders other
than $j$ are marginal on~$I$, so all items of goods $i \in I$ in the residual
supply must be allocated to $j$. Hence we reduce our allocation problem by
transferring $r_i$ items from the residual supply $\rb$ to bidder $j$'s partial
allocation bundle $\mb^j$ for each $i \in I$ and removing all bids marginal on items in $I$ from $\bids^j$.

If the \keylist{} $(I,j)$ has a single \linkgood{} $i^*$, all items of goods in $I
\setminus \{ i^* \}$ must be allocated to $j$ by the same argument as above.
Secondly, the bids in $\bids^j_I$ must be allocated a total of
$\sum_{\bid \in \bids^j_I} w(\bid)$ units of goods from $I$. Hence the difference
$\sum_{\bid \in \bids^j_I} w(\bid) - \sum_{i \in I \setminus \{ i^* \} } r_i$
gives us the number of items of $i^*$ that must be allocated to bidder $j$. This yields Algorithm~\ref{alg:proc2}.

\begin{algorithm}[htb]
\caption{\textsc{UnambiguousMarginals} (process all unambiguous marginal bids)}
\label{alg:proc2}
\begin{algorithmic}[1]
\State \textbf{Input:} Allocation problem $\mathcal{A} = [\pb, (\bids^j)_{j \in
J}, (\mb^j)_{j \in J}, \rb]$ and \keylist{} $(I,j)$ with at most one \linkgood{}
$i^*$.
\State \textbf{Output:} Valid reduction $\mathcal{A}'$ of $\mathcal{A}$ with no bids marginal on goods in $I$.
\If {$(I,j)$ has no \linkgood{}s{}}
  \ForAll {$i \in I$}
    \State Increment $m^j_i$ by $r_i$ and set $r_i$ to $0$.
  \EndFor
\Else
  \State Compute $d = \sum_{\bid \in \bids^j_I} w(\bid) - \sum_{i \in I
\setminus \{ i^* \} } r_i$.
  \State Increment $m^j_{i^*}$ by $d$ and decrement $r_{i^*}$ by $d$.
  \ForAll {$i \in I \setminus \{ i^* \}$}
    \State increment $m^j_i$ by $r_i$ and set $r_i$ to $0$.
  \EndFor
\EndIf
\State Remove all bids marginal on $I$ from $\bids^j$.
\end{algorithmic}
\end{algorithm}

\begin{lemma}\label{lemma:procedure2}
	\textsc{UnambiguousMarginals}, given in Algorithm~\ref{alg:proc2},
	returns a valid reduction of the input allocation
	problem in time $O(n|\bids^j|)$. Moreover, the marginal bids graph
$G_{\mathcal{A}'}$ of $\mathcal{A}'$ has strictly fewer edges than
$G_\mathcal{A}$.
\end{lemma}
\begin{proof}
	In order to see that $\mathcal{A}'$ is a valid allocation problem, we
verify that the new bid list $(\bids^j)'$ of bidder $j$ after applying
\textsc{UnambiguousMarginals} is locally valid by checking criterion~(2) of Definition~\ref{def:local-validity}.
Fix a price $\pb$ and goods $i,i'$. Recall that bids in $\bids^j$ cannot be marginal on goods in $I$ and $[n]_0 \setminus I$, by the definition of \keylist{}s. If $i \in I$ and $i' \in [n]_0 \setminus I$, then $\bids^j$ has no bids marginal on~$i$ and~$i'$. If both $i, i'$ are goods in $I$, all bids marginal on $i$ and $i'$ are removed by \textsc{UnambiguousMarginals}. If neither $i$ nor~$i'$ is a good in $I$, then none of the bids marginal on $i$ and~$i'$ are removed. As $\bids^j$ is locally valid by assumption, this implies that $(\bids^j)'$ is also locally valid.
It is straightforward to see that a solution to $\mathcal{A}'$ is a solution to
$\mathcal{A}$, as the partial allocation performed by \textsc{UnambiguousMarginals} is
unambiguous. Finally, let $i, i' \in I$ and note that the marginal bids graph
$G_\mathcal{A}$ contains an edge between $i$ and $i'$ that is not present in
$G_{\mathcal{A}'}$. As removing bids does not introduce new edges to the graph,
$G_{\mathcal{A}'}$ has strictly fewer edges.
\end{proof}

\subsection{The shift-project-unshift reduction}
\label{sec:shift}

Suppose that \textsc{FindParams} returns a cycle-\linkgood{} $i^*$ and the label $j^*$ of one of its adjacent edges in the cycle. In this case \textsc{Allocate} invokes the \textsc{ShiftProjectUnshift} procedure to obtain a reduction. Unlike \textsc{NonMarginals} and \textsc{UnambiguousMarginals}, this procedure does not make progress by allocating items and deleting the satisfied bids. Instead, it temporarily shifts the bids of bidder $j^*$ and perturbs the market-clearing price $\pb$ using SFM. We show that the marginal bids graph for this new allocation problem has strictly fewer edges, reducing the overall dependencies between the demands of bids. Finally, the procedure projects the bids of bidder $j^*$ in such a way that it can unperturb the market-clearing price while retaining the dependency structure of the new marginal bids graph. We now introduce the shift and project operations.

\paragraph{Shifting bids.} Suppose all bids are integral and we shift some
bidder's bids by a small quantity $\varepsilon < 1/4$ (we use $\varepsilon =
1/10$ for concreteness) in the direction of $i$ by adding $\varepsilon \eb^i$
to each bid vector. Then Lemma~\ref{lemma:shift-subset} and
Proposition~\ref{prop:shift-valid} (proved in Appendix~\ref{appendix:shift-valid}) together show that we can use SFM to find a
price ${\pb^\varepsilon \in \left \{ \pb + \varepsilon \eb^S, S \subseteq [n]
\right \}}$, at which the new allocation problem $\mathcal{A}'$ with shifted
bids and price vector~$\pb^\varepsilon$ is a valid reduction.

\begin{lemma}\label{lemma:shift-subset}
Fix $\varepsilon$ with $|\varepsilon| < 1/4$, as well as $i \in [n]$ and $j \in J$. Let $\mathcal{A}$
be an allocation problem with integral bids and prices, and let $\mathcal{A}'$
be obtained by replacing $\bids^j$ with $(\bids^j)' \coloneqq \left \{\bid +
\varepsilon \eb^i \mid \bid \in \bids^j \right \}$. Then the bid lists of all
bidders in $\mathcal{A}'$ are locally valid at any price $\pb^\varepsilon \in
B(\pb, \varepsilon)$. Moreover, the bundles demanded by any bidder at
$\pb^\varepsilon$ in $\mathcal{A}'$ form a subset of the bundles they demand
at $\pb$ in $\mathcal{A}$.
\end{lemma}
\begin{proof}
As the bids and prices in $\mathcal{A}$ are integral, each bid either demands
all goods $[n]_0$ or has a surplus gap of at least 1 at $\pb$. By
Lemma~\ref{obs:validity-margin}, all bid lists are 1/2-valid at $\pb$. This
implies that every unshifted bid list is $(1/2-\varepsilon)$-valid at
$\pb^\varepsilon$ and the shifted bid list is $(1/2-2\varepsilon)$-valid at
$\pb^\varepsilon$.
By Observations~\ref{obs:no-new-marginals} and~\ref{obs:no-new-marginals-2}, a
non-marginal bid demands the same good for all $\qb \in B(\pb^\varepsilon,
\delta)$ and sufficiently small $\delta > 0$. Hence a bidder will demand the
same set of bundles at all non-marginal prices in $B(\pb^\varepsilon, \delta)$.
As the bundles a bidder demands at $\pb^\varepsilon$ are by definition the
discrete convex hull of bundles they demand at non-marginal prices
infinitesimally close to $\pb^\varepsilon$, we are done.
\end{proof}

\begin{proposition}\label{prop:shift-valid}
Let $\mathcal{A}$ be an allocation problem with price vector $\pb$ and fix $\varepsilon$ with $|\varepsilon| < 1/4$. Suppose we shift all bids of bidder $j$'s in $\bids$ by $\varepsilon$ to obtain (aggregate) bid list $\bids'$. Then there exists a price $\pb^\varepsilon \in \left \{ \pb \pm \varepsilon \eb^S, S \subseteq [n] \right \}$ at which the residual supply~$\rb$ is demanded by $\bids'$. We can determine~$\pb^\varepsilon$ using submodular minimisation on the Lyapunov function $g$ with regard to list $\bids'$ and the residual supply $\rb$.
\end{proposition}

\paragraph{Projecting bids.}
Next we define a projection operation on bids. The idea of this operation is to modify a bid $\bid$ so that its surplus gap at $\pb$ is increased. As a
consequence, if that bid or the price is perturbed slightly, it `remembers'
its preferred goods.
Let $\bid$ be a bid and $I = \argmax_{i \in [n]_0} (b_i - p_i)$ be the set of
demanded goods at price $\pb$. The projection $\bid'$ of $\bid$ w.r.t.~$\pb$ is
defined as follows.
	\[
		\bid' \coloneqq
		\begin{cases}
			\bid - \eb^{[n]_0 \setminus I} & \text{if } 0 \in I, \\
			\bid + \eb^I & \text{ otherwise.}
		\end{cases}
	\]
Note that we allow for bid vector entries to be negative.
Figure~\ref{fig:project-example} illustrates the projection operation.

\begin{observation}\label{obs:project-margin}
	If bid $\bid$ demands all goods $[n]_0$, we have $\bid' = \bid$.
Otherwise, the projection operation on $\bid$ w.r.t.~$\pb$ increases the bid's
surplus gap at $\pb$ by 1.
\end{observation}

\begin{lemma}
\label{lemma:project-invariants}
Projecting all bids in a bidder's bid list $\bids$ w.r.t.~$\pb$ does
not change the set of bundles demanded at $\pb$. Moreover, if $\bids$ is
locally valid at $\pb$, the projected bid list is locally $1/2$-valid.
\end{lemma}
\begin{proof}
Suppose all bids in $\bids$ have a surplus gap of at least $\varepsilon$ at
$\pb$ and $\bids$ is locally $\delta$-valid at $\pb$. Note that at any price
$\qb$ in the open ball $B(\pb, \varepsilon)$ centred at $\pb$, every bid
demands the same set of goods before and after projecting. Hence the projected
bid list is locally $\min\{\delta, \varepsilon \}$-valid at $\pb$, and the set
of bundles demanded at $\pb$, consisting of the discrete convex hull of bundles
demanded in UDRs bordering $\pb$ remains the same. The second statement follows
from Observations \ref{obs:validity-margin} and~\ref{obs:project-margin}.
\end{proof}

\begin{figure}
	\centering
	\begin{subfigure}[b]{0.49\textwidth}
		\centering
		\includegraphics{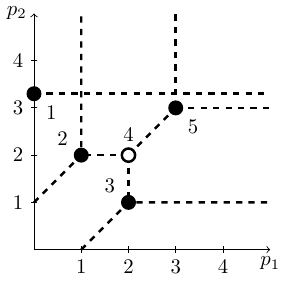}
		\caption{Five numbered bids before projecting.}
		\label{fig:before}
	\end{subfigure}
	\begin{subfigure}[b]{0.49\textwidth}
		\centering
		\includegraphics{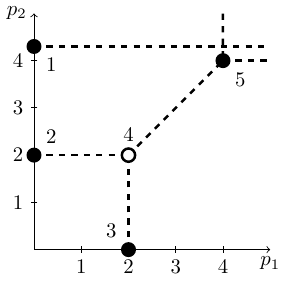}
		\caption{The result of projecting w.r.t.~prices $\pb=(4,4)$.}
		\label{fig:after}
	\end{subfigure}

	\caption{Example of the project operation w.r.t.~prices $\pb = (4,4)$
on five bids numbered from $1$ to $5$ shown in (a). The projected bids are
shown in (b). Note that only the negative bid is unchanged, as it demands all
goods, including the reject good.}
	\label{fig:project-example}
\end{figure}

\paragraph{The procedure.} The procedure \textsc{ShiftProjectUnshift} is stated in Algorithm~\ref{alg:proc3}. It shifts and projects to reduce the number of edges in marginal goods graph. This effectively reduces the number of demand ties among the bids in our allocation problem. Lemmas~\ref{lemma:procedure3-reduce} (proved in Appendix~\ref{appendix:procedure3-reduce}) and~\ref{lemma:procedure3-main} establish that the reduction obtained by \textsc{ShiftProjectUnshift} is valid and makes progress.

\begin{algorithm}[bth!]
  \caption{\textsc{ShiftProjectUnshift}}
  \label{alg:proc3}
  \begin{algorithmic}[1]
  \State \textbf{Input:} Allocation problem $\mathcal{A} = [\pb, (\bids^j)_{j \in J}, (\mb^j)_{j \in J}, \rb]$, cycle-\linkgood{} $i^*$ and edge label $j^*$.
  \State \textbf{Output:} Allocation problem $\mathcal{A}'$.
  \State Add $\frac{1}{10}\eb^{i^*}$ to each of $j^*$'s bids and
compute a new market-clearing price $\pb^\varepsilon = \pb +
\frac{1}{10}\eb^{S^*}$ by solving $S^* = \argmin_{S \subseteq [n]} g_{\rb}(\pb
\pm \frac{1}{10}\eb^S)$ using SFM.
  \State Project every bid list w.r.t.~$\pb^\varepsilon$.
  \State Subtract $\frac{1}{10}\eb^{i^*}$ from each of $j^*$'s
bids and reset price to $\pb$.
  \end{algorithmic}
\end{algorithm}

\begin{lemma}\label{lemma:procedure3-reduce}
	The marginal bids graph $G_{\mathcal{A}'}$ of $\mathcal{A}'$ has
strictly fewer edges than $G_\mathcal{A}$.
\end{lemma}

\begin{lemma}\label{lemma:procedure3-main}
	\textsc{ShiftProjectUnshift}, given in Algorithm~\ref{alg:proc3}, returns a reduction $\mathcal{A}'$ of $\mathcal{A}$ in time $O(T(n) + n|\bids|)$, where $T(n)$ is the time it takes to minimise an $n$-dimensional submodular function.
\end{lemma}
\begin{proof}
Let $\mathcal{A}'$ and $\mathcal{A}''$ denote the allocation problems after executing lines 3 and 4, respectively. First we~show that~$\mathcal{A}'$ is a reduction of
$\mathcal{A}$. Indeed, by Proposition~\ref{prop:shift-valid} and Step 1, the
residual bundle~$\rb$ is aggregately demanded at $\pb^\varepsilon$
in~$\mathcal{A}'$. Furthermore, for any allocation $(\rb^j)_{j \in J} = (\tb^j
- \mb^j)_{j \in J}$ of the residual bundle $\rb$ at price $\pb^\varepsilon$ in
$\mathcal{A}'$, Lemma~\ref{lemma:shift-subset} implies that $(\mb^j +
\rb^j)_{j \in J}$ is a solution of $\mathcal{A}$.
Secondly, Lemma~\ref{lemma:project-invariants} implies that $\mathcal{A}''$ is a reduction of~$\mathcal{A}'$. Finally, after line 5, the bid lists are locally valid at $\pb$ by Lemma~\ref{lemma:shift-subset}.
Proposition~\ref{prop:shift-valid} implies that the residual bundle $\rb$ is
demanded at some price
$\pb' \in \left \{ \pb^\varepsilon \pm \varepsilon \eb^S \mid S \subseteq [n] \right \}$.
Note that if a bundle is demanded at $\pb'$, it is also demanded at $[\pb']$ due to
the structure of our integral bid vectors, where $[ \cdot ]$ denotes the
operation of rounding each component to the nearest integer. As $[\pb'] = \pb$,
the result follows.
\end{proof}

\subsection{The main algorithm}
\label{subsec:allocate}
The algorithm \textsc{Allocate}, stated in Algorithm~\ref{alg:allocate} above, combines the procedures \textsc{NonMarginals}, \textsc{UnambiguousMarginals} and \textsc{ShiftProjectUnshift}
in order to solve the allocation problem described in
Section~\ref{sec:stmt}. Recall that it uses \textsc{FindParams} as a subroutine to decide whether to call \textsc{UnambiguousMarginals} or \textsc{ShiftProjectUnshift} in each iteration of the loop.
Theorem~\ref{thm:main} proves correctness and gives a running time bound for
\textsc{Allocate}. We note that this bound is likely to be pessimistic.


\begin{theorem}
\label{thm:main}
	\textsc{Allocate} solves the allocation problem in time $O \left
(n^2|J|(\alpha(n)n|\bids| + T(n)) \right )$, where $T(n)$ is the time required
to minimise an $n$-dimensional submodular set function.
\end{theorem}
\begin{proof}
By construction, the marginal bids graph of the initial allocation problem has
at most $|J| \binom{n+1}{2}$ edges. Every call to \textsc{UnambiguousMarginals}
or \textsc{ShiftProjectUnshift}
strictly reduces the number of edges (by Lemmas~\ref{lemma:procedure2}
and~\ref{lemma:procedure3-reduce}). Hence after at most $|J| \binom{n+1}{2}$
iterations of Step 2, the marginal bids graph of the current allocation problem
is an empty graph, implying that there are no more marginal bids. In
particular, at this point a single call to \textsc{NonMarginals} allocates all
remaining non-marginal bids and returns a vacuous allocation problem. As all
three procedures \textsc{NonMarginals}, \textsc{UnambiguousMarginals} and
\textsc{ShiftProjectUnshift} return a reduction in the sense of
Definition~\ref{defn:reduction}, the solution to the final vacuous allocation
problem is also a solution to the original allocation problem.
To see the running time guarantee, note that \textsc{FindParams} and
the procedures \textsc{NonMarginals}, \textsc{UnambiguousMarginals} and
\textsc{ShiftProjectUnshift} are each called at most $|J| \binom{n+1}{2} =
O(n^2|J|)$ times, and \textsc{FindParams} dominates \textsc{NonMarginals}
and \textsc{UnambiguousMarginals}.
\end{proof}

\paragraph{Incorporating priorities.}
\label{subsubsec:priorities}
The \textsc{FindParams} subroutine as stated in Section~\ref{subsec:graphs} is not fully specified, and different implementations may lead to different inputs for \textsc{UnambiguousMarginals} or \textsc{ShiftProjectUnshift} when given the same allocation problem. If the input for \textsc{ShiftProjectUnshift} depends on the implementation used, ties may be broken differently and thus the target bundle is allocated differently among the bidders.

In order to control which input for \textsc{ShiftProjectUnshift} is returned, we propose the use of a priority list consisting of a permutation of all good-bidder pairs $(i,j) \in [n]_0 \times J$. The priority list can be chosen to favour certain types of bidder; for example, if it is considered desirable to bring `small' bidders (having low demand) into the market, we
can prioritise their bids, making it slightly more likely that they will be allocated. Moreover, the list may be given as an additional input parameter at the start of \textsc{Allocate} or be generated and updated dynamically as \textsc{Allocate} runs. The subroutine \textsc{FindParams} is replaced by
\textsc{PriorityParams}, which returns the highest pair $(i,j)$ in the priority list that constitutes a valid input to \textsc{ShiftProjectUnshift}, or an input to \textsc{UnambiguousMarginals} if no such pair exists. Recall that a pair $(i,j)$ is a valid input for \textsc{ShiftProjectUnshift} if $i$ is a cycle-\linkgood{} in some multi-bidder cycle and $j$ is the label of one of its edges in this cycle. This is the case if the derived graph $D_\mathcal{A}$ has a cycle containing the edge from~$i$ to the \keylist{} $(I,j)$ satisfying $i \in I$. For any (undirected) graph $G$ and edge $vw$, one can check whether $G$ has a cycle containing $vw$ by determining whether its endpoints $v$ and $w$ are connected in the graph $G-vw$ obtained by deleting $vw$, using breadth first search (BFS). This implies the following subroutine.

\begin{algorithm}[htb!]
  \caption{\textsc{PriorityParams}}
  \label{alg:priorityparams}
  \begin{algorithmic}[1]
    \State Compute the derived graph $D_\mathcal{A}$.
    \ForAll {pairs $(i,j)$ in the priority list}
      \If {$i$ is a \linkgood{} and $i \in I$ for some \keylist{} $(I,j)$}
        \State Check whether the edge from $i$ to $(I,j)$ lies in a cycle by temporarily removing the edge and verifying (using BFS) whether the two endpoints are still connected in the graph.
        \State \Return $(i,j)$ if the edge lies in a cycle.
      \EndIf
    \EndFor
    \State \Return some leaf \keylist{} $(I,j)$ and the adjacent \linkgood{} $i$.
  \end{algorithmic}
\end{algorithm}

To see that \textsc{PriorityParams} is well-defined, note that the priority
list is a permutation of all possible pairs $(i,j)$. Hence if the derived graph
contains a cycle, the subroutine will return it, otherwise there exist at least
two leaves for Step 2 to choose from. In practice, we can prune the priority
list dynamically by removing pairs $(i,j)$ once $i$ is no longer a \linkgood{}.
This is possible because \textsc{UnambiguousMarginals} and
\textsc{ShiftProjectUnshift} do not add edges to the
derived graph, and so a good cannot become a \linkgood{} again at some later
point.

\section{Conclusions and further work}
This paper provides a practical process for running auctions in which bidders can express strong-substitutes preferences using positive and negative bids. Recent working papers \citep{BK-OS, LinTran2017} show that \emph{all} strong-substitutes preferences can be represented using appropriate combinations of these bids. While we show that verifying the validity of bid lists is \conp-complete, useful subclasses of collections of bids are efficiently checkable for validity in practice, and can be checked offline, bidder by bidder, prior to the auction.

Our approach to finding market-clearing prices builds on the long-step steepest-descent procedure proposed by \citet{Shioura17}. We both analyse his binary search method for determining long steps in our product-mix auction bidding language context, and also propose a novel second method for determining step lengths.

The structure of the product-mix bidding language also facilitates the design of an efficient algorithm for computing an equilibrium allocation of goods to bidders. Our approach contrasts with the previous literature on the problem of allocating supply to bidders at given prices (\citep{murota-book, MT2001} and \citep{PLW}), which instead assumes that bidders' preferences are accessed through an abstract valuation oracle.

All our algorithms approaches are polynomial-time in the size of the bid list inputs and, as we show in the online companion, efficient in practice.

One question for future research is whether we can also exploit the information provided by the bidding language to improve the computational efficiency of the pre-existing submodular function minimisation subroutine that both our price-finding and allocation
algorithms use. Another obvious question is the extent to which our methods can be extended to broader classes of valuations.

\section*{Acknowledgements}
Baldwin and Klemperer were supported by ESRC Grant ES/L003058/1. Goldberg and Lock were supported by a JP Morgan faculty fellowship during the work on the final version of the paper.
We thank the reviewers for valuable comments and suggestions, which have been helpful in improving this paper. We are particularly grateful to a referee who both identified a problem in an earlier draft of our paper, and took the trouble to develop an elegant solution to it.
We also thank our colleagues -- including Meg Meyer, Duncan Coutts, Martin Bichler and Maximilian Fichtl -- and Tobias Dammers and Andres Löh, who undertook the Haskell implementation and provided much helpful feedback as the algorithms were developed.

\bibliographystyle{abbrvnat}
\bibliography{library}

\appendix

\section{Experiments}
\label{appendix:experiments}
In order to evaluate the practical running time of our allocation algorithm, we
run experiments on various numbers of goods, bidders and bids. We use our own
Python implementation of the product-mix auction, available at
\url{https://github.com/edwinlock/product-mix}. Some effort was made to
optimise for speed by exploiting fast matrix operations provided by the NumPy
package
\citep{Numpy}.
Furthermore, in an effort to implement submodular
minimisation efficiently, the Fujishige-Wolfe algorithm was implemented in
combination with a memoization technique to reduce the number of submodular
function queries.

\subsection{Generating test data}
We describe a procedure to generate a valid list of positive and negative bids
at points within the lattice $[M]^n$ and a bundle $\vec{x}$ that is demanded in
aggregate by these bids at prices $\vec{p} = \frac{1}{2} M \eb^{[n]}$ in the
centre of the lattice. In our experiments, we fix $M=100$. The bids are
generated in such a way that for parameter $q$ at least $q$ bids are marginal
between two or more goods at $\vec{p}$ and we pick $\vec{x}$ such that
$\vec{p}$ is the component-wise minimal market-clearing price vector.
For any permutation $\pi$ of $[n]$, let $\vec{p}_\pi = \vec{p} + \sum_{i \in
[n]} \frac{\varepsilon}{2i} \eb^{\pi (i)}$ for some $\varepsilon < 0.1$. Note
that a unique bundle is demanded at $\vec{p}_\pi$ for any permutation $\pi$.

\begin{algorithm}[h!]
  \caption{\textsc{GenerateList}(n,M,q)}
  \label{alg:generatelist}
  \begin{algorithmic}[1]
    \State \textbf{Initialise:} Empty bid list $\bid$ and bundle $\xb = \bf{0}$.
    \State \textbf{Repeat} the following $q$ times:
    \State Pick a subset $S$ of goods with $|S| \geq 2$ from $[n]_0$ and
flip a fair coin.
    \If {the coin lands on heads}
      \State Pick a positive bid that is marginal on goods $S$ at $\pb$ and add it to $\bids$. \State Pick any good $i \in S$ uniformly at random and increment $x_i$ by one.
    \Else
		  \State Generate a negative bid $\bid$ that is marginal on goods $S$ at $\pb$, as well as the following three positive bids.
      \State Pick two goods $i, j \in [n]$ and add a bid at points $\bid - \lambda_i \eb^i$ and $\bid - \lambda_j \eb^j$ for some $1 \leq \lambda_i \leq b_i - 1$ and $1 \leq \lambda_j \leq b_j - 1$.
      \State Add a bid at $\bid + \lambda \vec{1}$ for some $1 \leq \lambda \leq \min_{i \in [n]} {M-b_i}$.
      \State Pick a permutation $\pi$ of $[n]$ and increment $\vec{x}$ by the bundle aggregately demanded at $\vec{p}_\pi$ by the four bids just generated.
    \EndIf
    \If {$\vec{p}$ is the component-wise minimal market-clearing price vector of $\vec{x}$}
      \State \Return $\bids$.
    \Else
      \State Repeat the algorithm.
    \EndIf
  \end{algorithmic}
\end{algorithm}

Generated in this way, every bid list has $2.5q$ bids in expectation. The
procedure \textsc{GenerateList} is repeated for each bidder, so the total
number of bids generated is $B = 2.5qm$.

\begin{lemma}
	The bid list generated by \textsc{GenerateList} is valid.
\end{lemma}
\begin{proof}
	Note that the union of finitely many valid bid lists is again a valid
bid list. Hence it suffices to show that the list consisting of one negative
bid $\bid$ and three positive bids $\bid - \lambda_i \eb^i, \bid - \lambda_j
\eb^j$ and $\bid + \lambda \vec{1}$ is valid. To see this, apply Theorem~19 from the main paper.
\end{proof}

\subsection{Testing the algorithm}
We run \textsc{Allocate} on allocation problems with bid lists generated by
\textsc{GenerateList} for different numbers of goods $n$, bidders $m$ and bids
$B$. It should be noted that, as expected, the value of $M$ has no discernable
impact on the running time of the allocation algorithm and thus can be fixed to
$M=100$. We perform three pairs of experiments in which we vary one of the
parameters $n, m, q$ and fix the other two to realistic values.

\begin{enumerate}
	\item For the first pair of experiments, we vary the average number of
bids by running \textsc{GenerateList} with values $q=20, 40, 60, \ldots, 500$.
We fix the number of bidders to $m=5$ and the number of goods to $n=2$ and
$n=10$, respectively.
	\item For the second pair of experiments, we vary the number of goods
from $n=10$ to $50$ in steps of $5$ and fix $q$ to 50 and 100, respectively.
	\item Finally, the last pair of experiments varies the number of
bidders from $m=2$ to $m=20$ in steps of 1 with two goods $n=2$, while the bid
numbers are fixed by setting $q$ to $50$ and $100$, respectively.
\end{enumerate}

For each data point $(n,m,q)$, 50 allocation problems with $n$ goods, $m$
bidders and $2.5q$ bids per bidder (in expectation) are generated. The
\textsc{Allocate} algorithm is then timed on each allocation problem and the
average over all 50 times is recorded.

\subsection{Results}
The outcomes of the three pairs of experiments are shown in
Figures~\ref{fig:exp1} to \ref{fig:exp3}. The experimental data corroborates
the running time bounds for \textsc{Allocate} given in Theorem~9:
the algorithms is linear in the number of bids and quadratic in the number of
goods. Figure~\ref{fig:exp3} suggests that our algorithm runs in quadratic time
on our generated allocation problems, which is in line with our theoretical
bound, as the total number of bids $B=2.5qm$ is linear in $m$. Overall, we see
that our allocation algorithm runs quickly even when presented with a large
number of bids.

\begin{figure}[htb!]
	\centering
	\begin{subfigure}{.49\textwidth}
		\centering

\includegraphics[scale=0.45]{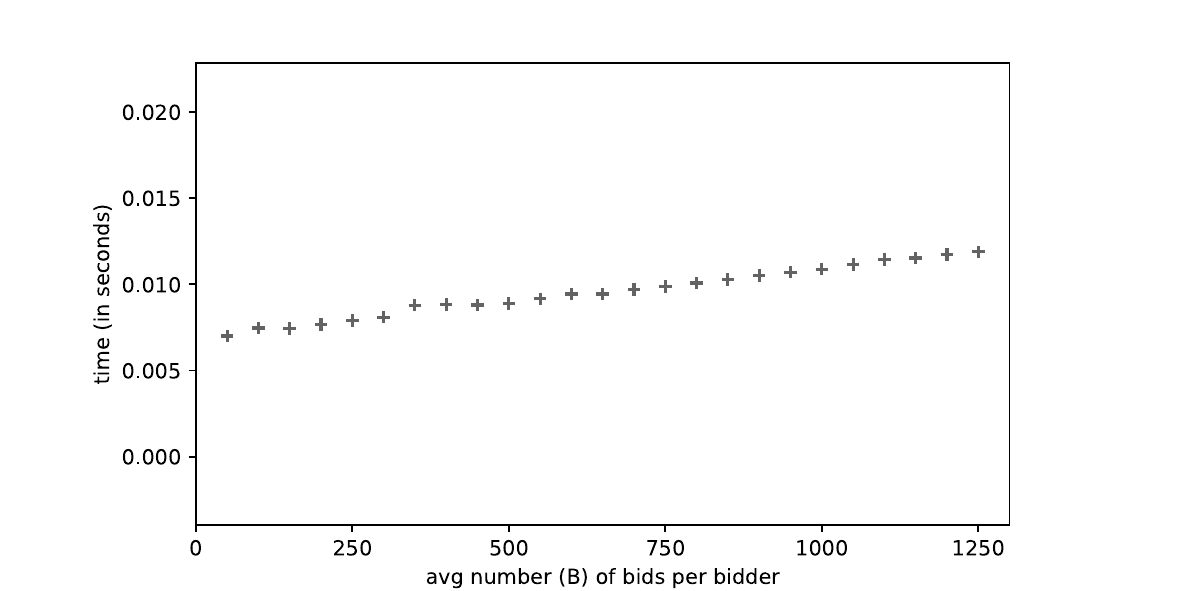}
		\caption{Parameters: $n=2, m=5, M=100$.}
	\end{subfigure}
	\begin{subfigure}{.49\textwidth}
		\centering

\includegraphics[scale=0.45]{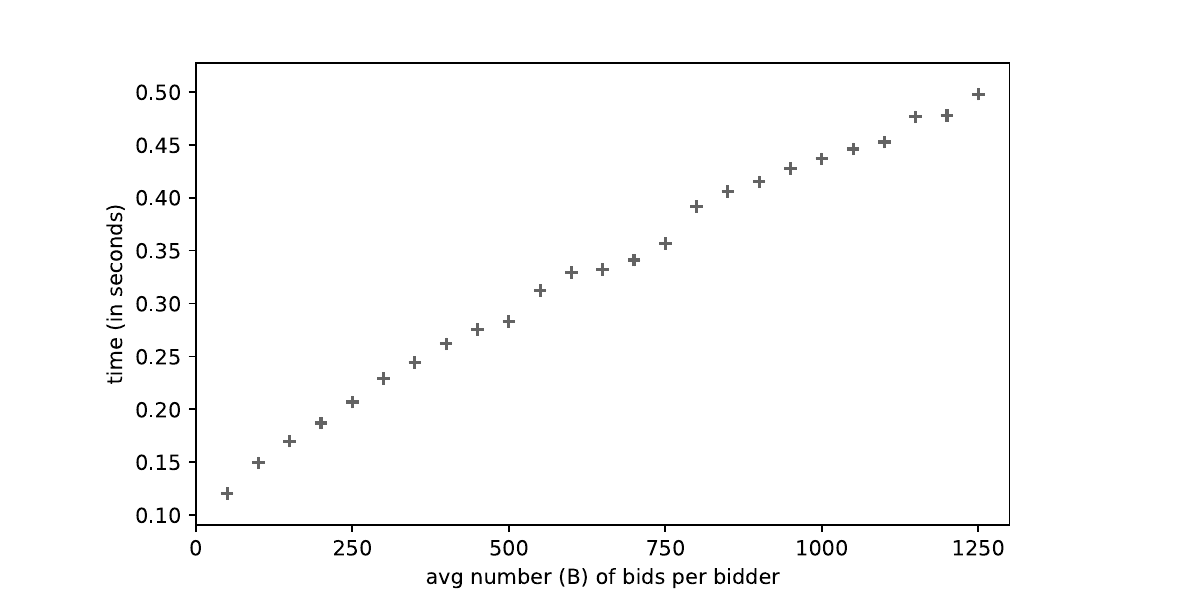}
		\caption{Parameters: $n=10, m=5, M=100$.}
	\end{subfigure}
	\caption{Testing \textsc{Allocate} by varying the number $B$ of bids
for each bidder while keeping all other parameters fixed. We increase $q$ from
$20$ to $500$ in steps of $20$, fix the number of bidders at $m=5$ and set the
number of goods to $n=2$ (a) and $n=10$ (b), respectively.}
	\label{fig:exp1}
\end{figure}

\begin{figure}[htb!]
	\centering
	\begin{subfigure}{.49\textwidth}
		\centering

\includegraphics[scale=0.45]{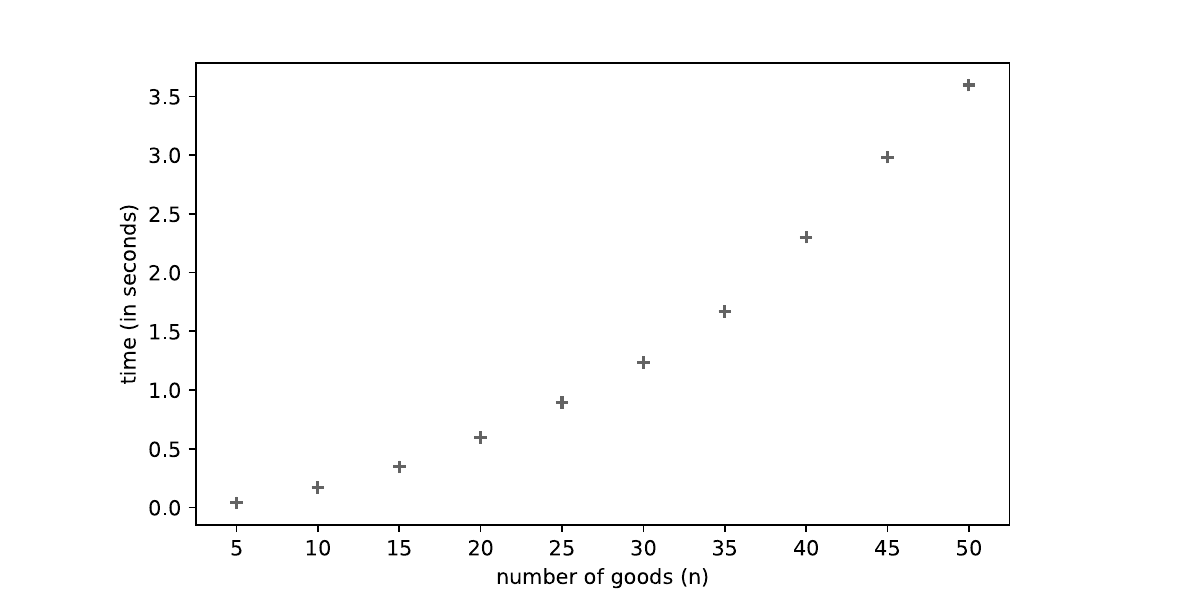}
		\caption{Parameters: $m=5, M=100, q=50$.}
	\end{subfigure}
	\begin{subfigure}{.49\textwidth}
		\centering

\includegraphics[scale=0.45]{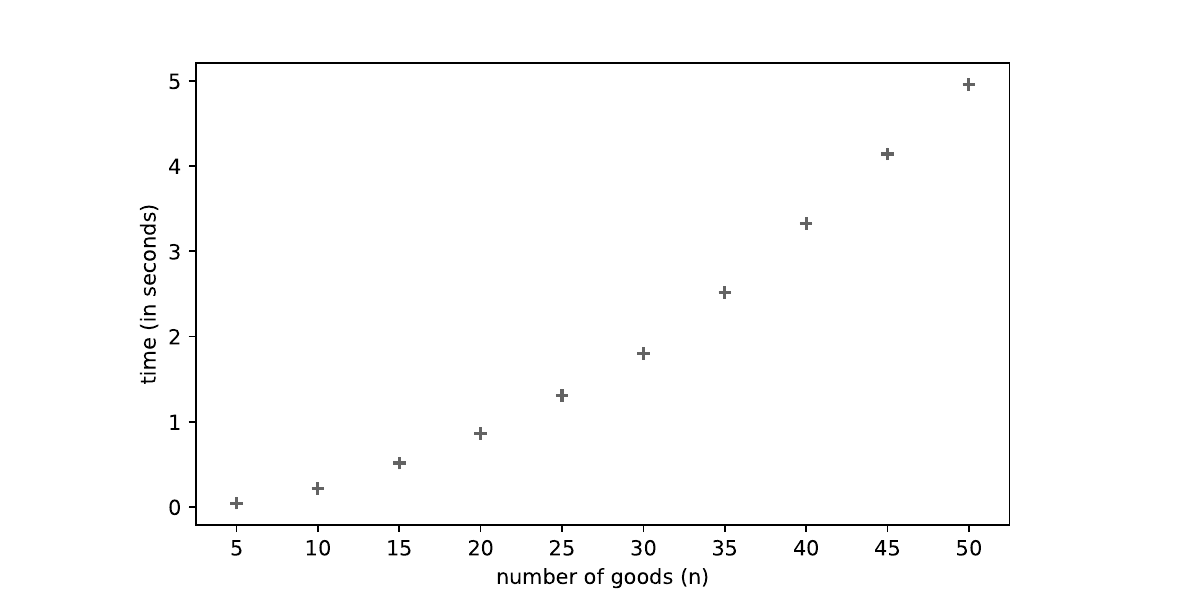}
		\caption{Parameters: $m=5, M=100, q=100$.}
	\end{subfigure}
	\caption{Testing \textsc{Allocate} by varying the number $n$ of goods.
We increase $n$ from $5$ to $50$ in steps of $5$, fix the number of bidders at
$m=5$ and $M=100$ and set the number of bids to $q=50$ (a) and $q=100$ (b),
respectively.}
	\label{fig:exp2}
\end{figure}

\begin{figure}[htb!]
	\begin{subfigure}{.49\textwidth}
		\centering

\includegraphics[scale=0.45]{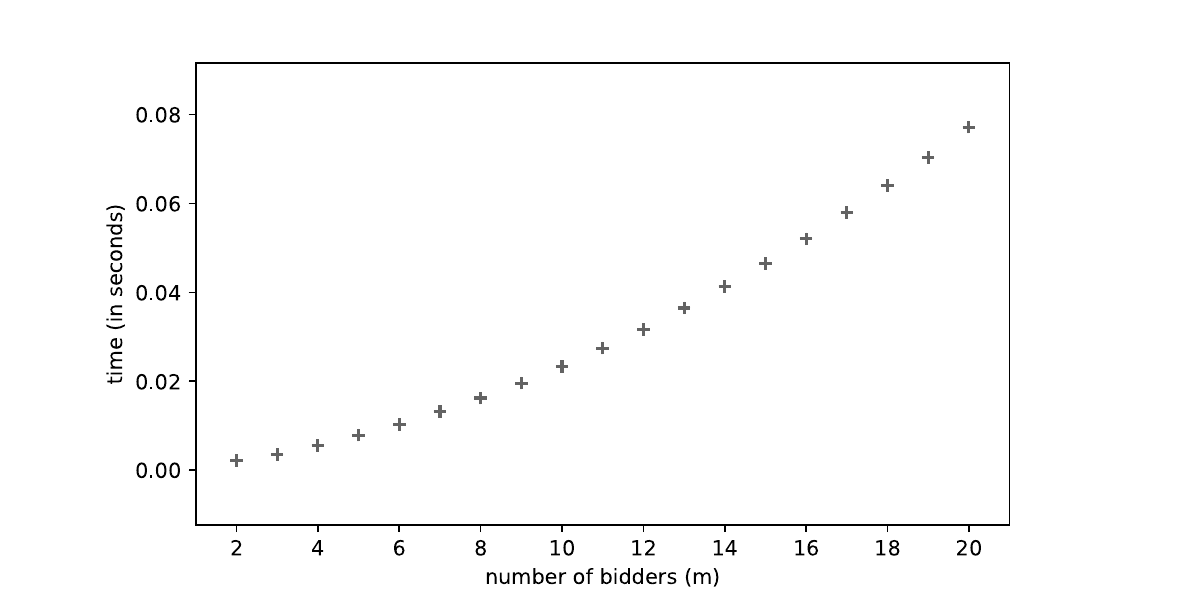}
		\caption{Parameters: $n=2, M=100, q=50$.}
	\end{subfigure}
	\begin{subfigure}{.49\textwidth}
		\centering

\includegraphics[scale=0.45]{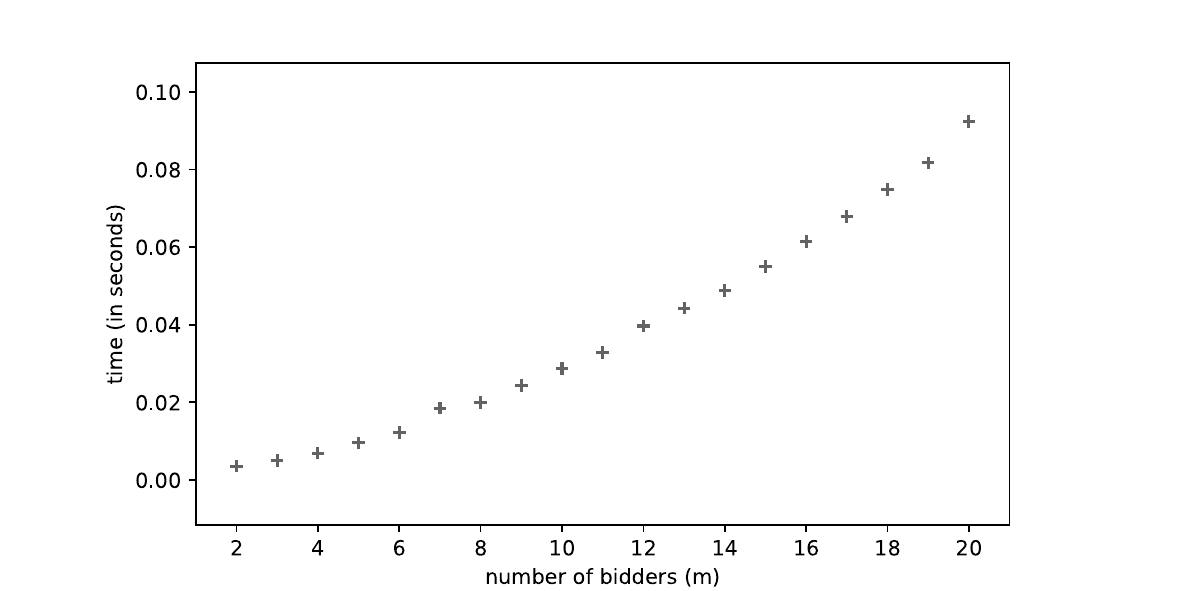}
		\caption{Parameters: $n=2, M=100, q=100$.}
	\end{subfigure}
	\caption{Testing \textsc{Allocate} by varying the number $m$ of
bidders. We increase $m$ from $2$ to $20$ in steps of $1$, fix the number of
goods at $n=2$ and set $q$ to $50$ (a) and $100$ (b), respectively.}
	\label{fig:exp3}
\end{figure}

\section{Additional proofs}
\label{sec:missing-proofs}

We give some definitions that are used in the proofs below. Here and in the following, we use terminology and material developed in \citep{BK}, drawing on the literature in tropical geometry, especially \citep{Mikhalkin}; see also the textbooks by \citet{Maclagan.Sturmfels2015} and \citet{Joswig}. 

\textbf{Geometric objects for valuations.}
The Locus of Indifference Prices (LIP) $\mathcal{L}_u$ for valuation~$u$ with demand correspondence $D_u$ is defined as
\[
    \mathcal{L}_{u} = \{ \pb \in \R^n \mid |D_u(\pb)| > 1 \}.
\]
This set was identified as an object of mathematical interest by \citet{Mikhalkin} and is also known as a ``tropical hypersurface'' \citep{Maclagan.Sturmfels2015, Joswig}.

Recall that a \emph{rational polyhedral complex} is a finite collection of polyhedrons, closed under taking faces; non-empty intersections of these polyhedrons must be faces of them, and these faces must have rational slope \citep{Maclagan.Sturmfels2015,Joswig}.  $\mathcal{L}_u$ can be decomposed into an $(n-1)$-dimensional \emph{rational polyhedral complex} $\Pi_u$ \citep[Proposition 2.1]{Mikhalkin}. We call the $(n-1)$-dimensional polyhedra in this complex its \textit{facets}. Let $w_u$ be the weight function obtained as follows: if $\xb$ and $\yb$ denote the two bundles demanded on either side of a facet $F$, then the weight $w_u(F)$ of $F$ is the greatest common divisor of the entries of $\xb - \yb$. We recall the following proposition from~{\citep{BK}}.

\begin{proposition}[\mbox{\citep[Proposition 2.4]{BK}}]
\label{prop:demand-change}
    The change in demand, as prices change between the UDRs on either side of facet $F$, is $w_u(F)$ times the primitive integer vector that is normal to $F$ and points in the opposite direction to the change in price.
\end{proposition}
In particular we identify the \emph{strong-substitutes vectors}: those non-zero vectors in $\Z^n$ which have at most one $+1$ entry, at most one $-1$ entry, and no other non-zero entries. (These vectors are also distinguished in the literature on discrete convex analysis, being the edges of discrete generalised polymatroids; see e.g.~\citep[Theorem 3.9]{ST15}). Recall that we assume throughout this paper that valuations are concave-extensible. This implies the following theorem (see also \mbox{\citep[Proposition 3.8]{BK}}).

\begin{theorem}[\mbox{\citep[Theorems 4.1 and 4.10]{ST15}}]
\label{thm:ss-demand-type}
A valuation~$u$ is strong-substitutes if and only if the facets of $\mathcal{L}_u$ are normal to strong-substitutes vectors.
\end{theorem}

Following \citep[Definition 3]{Mikhalkin}, we say that a weighted polyhedral complex $(\Pi, w)$ is \emph{balanced} if, for every $(n-2)$-cell $G$ of $\Pi $, the weights $w(F_{j})$ on the facets $F_{1},\ldots,F_{l}$ that contain $G$, and primitive integer normal vectors $\mathbf{v}_{F_{j}}$ for these facets that are defined by a fixed rotational direction about~$G$, satisfy $\sum_{j=1}^{l}w(F_{j})\mathbf{v}_{F_{j}}=0$.  Balancing is essential to what \citep[Theorem 2.14]{BK} call the ``Valuation-Complex Equivalence Theorem'':
\begin{theorem}[{{\citep[Remark 2.3 and Proposition 2.4]{Mikhalkin}}}]
\label{thm:valuation-complex}
Suppose that $(\Pi,w)$ is an $(n-1)$-dimensional weighted rational polyhedral complex in $\R^n$ with positive weights only, and that $\mathcal{L}$ is the union of the cells in $\Pi$.  There exists a finite set $A \subsetneq \Z^n$ and a concave-extensible valuation $u:A \to \R$ such that $\mathcal{L}_u = \mathcal{L}$ and $w_u = w$, if and only if $(\Pi, w)$ balanced.  Moreover in this case, then for any price vector $\pb$ not contained in $\mathcal{L}$, there exists a unique concave-extensible such $u$ satisfying also $u(\bm{0}) = 0$ and $D_u(\pb)  = \{\bm{0}\}$.
\end{theorem}

\textbf{Geometric objects for bids.}
In analogy to the definitions above, we now define geometric objects for product-mix bids and bid lists.
The Locus of Indifference Prices (LIP) for a bid~$\bid$ is
\[
	\mathcal{L}_{\bid} = \{ \pb \in \R^n \mid |\argmax_{i \in [n]_0} (b_i -
p_i)| > 1 \}.
\]
Similarly, the LIP of a list of bids $\bids$ is given by
\[
	\mathcal{L}_\bids = \{ \pb \in \R^n \mid |D_\bids(\pb)| > 1\}.
\]
Any $\mathcal{L}_\bids$ can be decomposed into an $(n-1)$-dimensional rational
polyhedral complex $\Pi_\bids$. We call the $(n-1)$-dimensional polyhedra in
this complex its \emph{facets}. Endow every facet $F$ of $\Pi_\bids$ with a
weight, as follows:
\[
	w_\bids(F) \coloneqq \sum \{ w(\bid) \mid \bid \in \bids \text{ and }
\mathcal{L}_{\bid} \supseteq F\}.
\]

\subsection{Proof of Theorem~\ref{thm:validity}}
\label{appendix:validity}
\begin{proof}[Proof of Theorem~\ref{thm:validity}]
The equivalence of \eqref{eq:valid-convexity} and \eqref{eq:valid-cover} follows from Proposition \ref{prop:local-validity} if we set $P = \mathbb{R}^n$. The implication $\eqref{eq:valid-SS} \Rightarrow \eqref{eq:valid-convexity}$ is well-known; \citet[Theorem 7.1]{ST15} and \citet[Theorem~4.3]{MSY2016}, for instance, show that $f_u$ is a polyhedral $L^\natural$-convex function and these functions are defined as a subclass of convex functions (cf.~{\citep{MSY2016}}).
We now show $\eqref{eq:valid-cover} \Rightarrow \eqref{eq:valid-SS}$ using Theorem~{\ref{thm:valuation-complex}}. Let $(\Pi_{\bids},w_\bids)$ be the weighted polyhedral complex associated with $\bids$. Note that as well as showing that balancing holds, we must also demonstrate that all weights of the polyhedral complex are positive, which is the setting for Theorem~{\ref{thm:valuation-complex}}.

Observe first that every facet of $\Pi_\bids$ has a strong-substitutes vector as its normal vector, as this property holds for any individual $\mathcal{L}_\bid$. Moreover, the weighted polyhedral complex $(\Pi_\bid,w_\bid)$ associated with a single bid is balanced. If $w(\bid)=1$, then this follows because
$(\Pi_\bid,w_\bid)$ is the weighted polyhedral complex associated with a simple
valuation for at most one unit of any good, and apply Theorem \ref{thm:valuation-complex}.
This then extends to the case $w(\bid)=-1$ because changing the sign of the
weights of all facets does not affect balancing. Then $(\Pi_\bids,w_\bids)$ is
also balanced, as it is the complex corresponding to the union of balanced LIPs.

Next we see that all weights of $(\Pi_\bids,w_\bids)$ are non-negative. Let $F$ be a facet of $\Pi_\bids$ and fix a point~$\vec{p}$ in the relative interior of $F$. Then, since the normal vector to $F$ is a strong-substitutes vector and by Proposition~{\ref{prop:demand-change}}, every bid is either non-marginal at $\vec{p}$ or marginal between the same two distinct goods, $i,i'\in[n]_0$. Let $\bids_{ii'}$ denote the bids that are marginal between $i,i'$ at $\vec{p}$. Note that we have $F \subseteq \mathcal{L}_{\bid}$ if and only if $\bid$ is marginal on $i, i'$, so $w_\bids(F) = \sum_{\bid \in \bids_{ii'}} w(\bid) \geq 0$ (by assumption~\eqref{eq:valid-cover}).

$(\Pi_\bids,w_\bids)$ may have zero-weighted facets: let the set of these be $\Pi_\bids^0$ and write $\Pi_\bids^+\coloneqq\Pi_\bids\setminus\Pi_\bids^0$. Then $\Pi_\bids^+$ inherits from $\Pi_\bids$ the structure of a polyhedral complex. Moreover, if we write $w_\bids^+$ for the restriction of $w_\bids$ to the facets of $\Pi_\bids^+$, then $(\Pi_\bids^+,w_\bids^+)$ is balanced: removing $0$-weighted facets does not affect this criterion.

So, we write $\mathcal{L}_\bids^+$ for the union of the cells in $\Pi_\bids^+$.
Let $\overline{\bm{p}}$ be a price vector whose components are strictly larger
than the components of any bid, so that $D_\bids (\overline{\bm{p}}) = \bm{0}$.
By Theorem~{\ref{thm:valuation-complex}}, there exists a unique concave-extensible
valuation $u$ such that $\mathcal{L}_u = \mathcal{L}_\bids^+$, $w_u =
w_\bids^+$, $u(\bm{0}) = 0$ and $D_u(\overline{\bm{p}}) = \vec{0}$.  We observed above that all facets of $\Pi_\bids$, and hence of $\mathcal{L}_u$, are normal to a strong-substitutes vector, so $u$ is a strong-substitutes valuation by Theorem \ref{thm:ss-demand-type}.

By Proposition~{\ref{prop:demand-change}}, the facets and weights of $\mathcal{L}_u$ define all changes in demand between UDRs associated with valuation~$u$.
If $w(\bid)$ is positive then $\mathcal{L}_\bid$ has the property that facets and weights define all changes in demand between UDRs (see also \citep[Section 2.2]{BK}), and it is easy to see that the same follows for negative-weighted bids, and thus extends by aggregation to our positive-weighted polyhedral complex $(\Pi_\bids^+,w_\bids^+)$. 
So $D_u(\overline{\bm{p}}) = D_\bids(\overline{\bm{p}})$ implies
$D_u(\vec{p}) = D_\bids(\vec{p})$ for all UDR prices $\vec{p} \in \mathbb{R}^n$, and hence for all prices, since $u$ is concave-extensible and $D_\bids(\vec{p})$ is defined to be discrete convex.

We now show that $f_\bids = f_u$. For any price $\vec{p}$ at which demand is unique for every bid $\bid \in \bids$, let $\vec{x}$ be the bundle demanded and let $(i(\bid))_{\bid \in \bids}$ be the list of goods allocated to the bids. That is, each bid $\bid \in \bids$ uniquely demands $i(\bid)$ and $\xb = \sum_{\bid \in \bids} w(\bid) \vec{e}^{i(\bid)}$. Then we claim that
\begin{equation}\label{eq:char-valuation}
	u(\xb) = \sum_{\bid \in \bids} w(\bid) b_{i(\bid)}.
\end{equation}
This then immediately implies $f_\bids(\vec{p}) = f_u(\vec{p})$ for any price
$\vec{p} \in \mathbb{R}^n$ at which demand is unique for every bid, since
\[
f_\bids(\vec{p}) = \sum_{\bid \in \bids} w(\bid) \max_{i \in [n]_0} (b_i -
p_i) = \sum_{\bid \in \bids} w(\bid) (b_{i(\bid)} - p_{i(\bid)}) = u(\vec{x}) -
\vec{x} \cdot \vec{p} = f_u(\vec{p}).
\]
But these prices are dense in $\R^n$ so it follows by continuity of these piecewise-linear functions that $f_\bids(\vec{p})=f_u(\vec{p})$ for all $\vec{p}\in\R^n$.

To see that \eqref{eq:char-valuation} holds, first observe that it holds for the demand and allocation at price  $\overline{\vec{p}}$: here, by definition $i(\bid)=0$ for all $\bid\in\bids$, so $\xb=\bm{0}$; and $u(\bm{0})=0$. 

Next, consider price vectors $\vec{p}'$, $\vec{p}''$ at which bundles $\vec{x}'$, $\vec{x}''$ are uniquely demanded, via unique allocations to bids 
$(i'(\bid))_{\bid \in \bids}$ and $(i''(\bid))_{\bid \in \bids}$ respectively.  Suppose moreover that there exists a price $\vec{p}$ satisfying $\vec{x}',\vec{x}''\in D_\bids(\vec{p})=D_u(\vec{p})$ with $i'(\bid),i''(\bid)\in\argmax_{i\in[n]_0}(b_i-p_i)$ for all $\bid\in\bids$: this holds if $\vec{p}'$ and $\vec{p}''$ are in the same or neighbouring connected components of the complement of $\Pi_\bids$ in $\R^n$.  Then we have:
\begin{equation}\label{eq:abc1}
	u(\vec{x}') - \vec{p} \cdot \vec{x}' = u(\xb'') - \vec{p} \cdot \vec{x}''
\end{equation}
and
\begin{equation}\label{eq:abc2}
	f_\bids(\vec{p}) = \sum_{\bid \in \bids} w(\bid) (b_{i'(\bid)} -
p_{i'(\bid)}) = \sum_{\bid \in \bids} w(\bid) (b_{i''(\bid)} - p_{i''(\bid)}).
\end{equation}
Suppose \eqref{eq:char-valuation} holds for bundle $\vec{x'}$ and
$(i'(\bid))_{\bid \in \bids}$. Then \eqref{eq:abc1} and \eqref{eq:abc2} imply
that \eqref{eq:char-valuation} holds for $\vec{x}''$ and $(i''(\bid))_{\bid \in
\bids}$. But now we can use the fact that \eqref{eq:char-valuation} holds for the demand and allocation at price $\overline{\vec{p}}$ to infer that 
\eqref{eq:char-valuation} for the demand and allocation at any price at which demand is unique, by induction.
\end{proof}

\subsection{Proof of Proposition~\ref{prop:local-validity}}
\label{appendix:local-validity}
We make use of the following technical observation to prove Proposition~\ref{prop:local-validity}.

\begin{observation}\label{obs:one-dim-convex}
	Let $f:[0,1] \to \mathbb{R}$ be a continuous, piecewise-linear
function. If its slope is non-decreasing as we move from $0$ to $1$, $f$ is
convex. This implies that if $f$ is not convex, there exists a point $\lambda$
at which two linear segments meet and the slope decreases.
\end{observation}

\begin{proof}[Proof of Proposition~\ref{prop:local-validity}]
Recall that $f_\bids$ restricted to $P$ is convex if and only if
\begin{equation}\label{eq:iuf-convex}
	f_\bids(\lambda \vec{p} + (1-\lambda)\vec{q}) \leq \lambda
f_\bids(\vec{p}) + (1-\lambda)f_\bids(\vec{q}), \quad \forall \vec{p}, \vec{q}
\in P \text{ and } \lambda \in [0,1].
\end{equation}

Suppose that $\vec{p}$ and $\vec{q}$ are two non-marginal prices in
neighbouring UDRs such that for some $\lambda \in (0,1)$, the point $\vec{r} =
\lambda \vec{p} + (1-\lambda) \vec{q}$ lies on the interior of a facet
separating the UDRs. This implies two possibilities for each bid $\bid \in
\bids$. Either $\bid$ non-marginally demands the same good at $\vec{p}$ and
$\vec{q}$ (and thus also at $\vec{r}$). Or $\bid$ non-marginally demands $i$ at
$\vec{p}$ and $i'$ at $\vec{q}$, and is marginal (only) on $i$ and $i'$ at
$\vec{r}$. Moreover, all bids of the latter kind are marginal on the same two
goods.

Let $\bids_{ii'}$ denote the bids that demand the set $\{i, i'\}$ at $\vec{r}$.
We show that $f_{\bids}$ satisfies \eqref{eq:iuf-convex} for $\vec{p}, \vec{q}$
and any $\lambda \in [0,1]$ if and only if the weights of the bids marginal on
$i$ and $i'$ sum to a non-negative number, i.e.~if $\sum_{\bid \in \bids_{ii'}}
w(\bid) \geq 0$.
By construction, we have
\begin{align*}
	f_\bids (\vec{p}) = \sum_{\bid \in \bids_{ii'}} w(\bid) (b_i - p_i) +
\sum_{\bid \not \in \bids_{ii'}} w(\bid) \max_{j \in [n]_0} (b_j - p_j), \\
	f_\bids (\vec{q}) = \sum_{\bid \in \bids_{ii'}} w(\bid) (b_{i'} -
q_{i'}) + \sum_{\bid \not \in \bids_{ii'}} w(\bid) \max_{j \in [n]_0} (b_j -
q_j), \\
	f_\bids (\vec{r}) = \sum_{\bid \in \bids_{ii'}} w(\bid) (b_i - r_i) +
\sum_{\bid \not \in \bids_{ii'}} w(\bid) \max_{j \in [n]_0} (b_j - r_j).
\end{align*}

Hence the inequality
\begin{equation}\label{eq:specific-nonconvex}
	f_\bids(\vec{r}) \leq \lambda f_\bids(\vec{p}) + (1-\lambda)
f_\bids(\vec{q})
\end{equation}
holds if and only if
\begin{equation}\label{eq:valid-convex}
	\sum_{\bid \in \bids_{ii'}}w(\bid)[(b_i - q_i) - (b_{i'} - q_{i'})]
\leq 0.
\end{equation}

Note that the term $(b_i - q_i) - (b_{i'} - q_{i'})$ is strictly negative, as $\bid$ demands $i'$ and not $i$ at $\vec{q}$. Moreover, this term is the same for all bids $\bid \in \bids_{ii'}$, as $i$ and $i'$ are both demanded at $\vec{r}$. This implies by \eqref{eq:valid-convex} that $\sum_{\bid \in \bids_{ij}} w(\bid) \geq 0$ if and only if \eqref{eq:specific-nonconvex} holds.

Now we prove our main statement. Suppose $f_\bids$ restricted to $P$ is not convex. Then there exist $\vec{p}, \vec{q} \in P$ and $\lambda \in [0,1]$ that violate \eqref{eq:iuf-convex}. Due to continuity of $f_\bids$, we can perturb $\vec{p}$ and $\vec{q}$ slightly so that \eqref{eq:iuf-convex} is still violated, $\vec{p}, \vec{q}$ are non-marginal and the line segment $[\vec{p}, \vec{q}]$ crosses only interiors of facets of the LIP $\mathcal{L}_\bids$. We can assume without loss of generality that $\vec{p}$ and $\vec{q}$ are in neighbouring UDRs; otherwise, we can apply to find two non-marginal prices on the interior of the line segment for which \eqref{eq:iuf-convex} fails. Thus we can apply the above to see the implication $2 \Rightarrow 1$.

For the converse ($1 \Rightarrow 2$), suppose there exist $\vec{r}$ and $i,i' \in [n]_0$ such that the weights of bids marginal on $i,i'$ at $\vec{r}$ sum to a negative number. Without loss of generality, we can assume that all marginal bids at $\vec{r}$ are marginal only on goods $i,i'$, by subtracting $\delta \vec{e}^{\{i,i'\}}$ from $\vec{r}$ for some infinitesimal positive value of $\delta$ if necessary. Hence $\vec{r}$ lies in the interior of a facet $F$ and for small enough $\varepsilon > 0$, the points $\vec{p} = \vec{r} + \varepsilon \vec{e}^i - \varepsilon \vec{e}^{i'}$ and $\vec{q} = \vec{r} - \varepsilon \vec{e}^i + \varepsilon \vec{e}^{i'}$ lie in $P$ and in neighbouring UDRs separated by $F$. By the above, this implies that $\vec{p}$ and $\vec{q}$ with $\lambda = 1/2$ violate \eqref{eq:iuf-convex} and we are done.
\end{proof}

\subsection{Proof of Lemma~\ref{obs:validity-margin}}
\label{appendix:surplus-validity}

\begin{proof}[Proof of Lemma~\ref{obs:validity-margin}]
We show that $f_\bids$ is convex on $B(\pb,\varepsilon/2)$ by verifying that it satisfies midpoint convexity. Note that $f_\bids$ is linear on the line segment connecting $\pb$ to $\qb$, for any $\qb \in B(\pb, \varepsilon/2)$. Indeed, $f_\bids$ is the sum of terms $w(\bid)\max_{i \in [n]_0} (b_i - p_i)$ and it suffices to show that each term is linear. Hence fix $\bid$ and define $h(\pb) = w(\bid) \max_{i \in [n]_0} (b_i - p_i)$ as well as $\rb = \theta \pb + (1-\theta) \qb$. By Observation~\ref{obs:no-new-marginals}, the goods demanded by $\bid$ at $\qb, \rb$ and $\pb$ satisfy $I_q \subseteq I_r \subseteq I_{\pb}$. Hence, for any $i^* \in I_q$, we have $h(\pb) = w(\bid) (b_{i^*} - p_{i^*}) = \theta h(\pb) + (1-\theta)h(\qb)$.
Now fix $\qb, \qb' \in B(\pb, \varepsilon /2)$ and choose $\theta > 0$ so
that $\rb = \theta \pb + (1-\theta)\qb$, $\rb' = \theta \pb + (1-\theta)\qb'$
and $(\rb+\rb')/2$ are in $B(\pb, \delta)$. As $f_\bids$ is convex on
$B(\pb,\delta)$, by assumption, we have
	$f \left( \frac{1}{2}(\rb + \rb') \right) \leq \frac{1}{2}f(\rb)
+ \frac{1}{2}f(\rb')$.
	Secondly, we have $f(\rb) = \theta f(\pb) + (1-\theta)f(\qb)$, $f(\rb')
= \theta f(\pb) + (1-\theta)f(\qb')$ and $f(\frac{1}{2}(\rb+\rb')) = \theta
f(\pb) + (1-\theta)f(\frac{1}{2}(\qb+\qb'))$ due to the linearity of $f_\bids$
on the line segments connecting $\pb$ to $\qb, \qb'$ and $(\qb+\qb')/2$ . This
implies midpoint convexity,
		$f \left (\frac{1}{2}(\qb+\qb') \right ) \leq \frac{1}{2}f(\qb)
+ \frac{1}{2}f(\qb').$
\end{proof}

\subsection{Proof of Proposition~\ref{prop:shift-valid}}
\label{appendix:shift-valid}

In order to prove Proposition~{\ref{prop:shift-valid}}, we first prove in Lemma~{\ref{lemm:local-valuation}} that the demand of a bid list that is locally valid around integral prices coincides, in the local neighbourhood of these prices, with the demand of some strong-substitutes valuation (in analogy to the third statement in Theorem~{\ref{thm:validity}}). Working directly with the valuations that correspond locally to our bid lists, we then establish two technical lemmas (Lemmas~{\ref{lemma:correct-shift}} and {\ref{lemma:shift-subs:old}}) that illustrate the robustness of demand when we perturb valuations by some small value. We recall that a vector is strong-substitutes if it has at most one $+1$ entry, at most one $-1$ entry, and no other non-zero entries. In Lemma~{\ref{lemma:shift-subs:old}} and the proof of Proposition~{\ref{prop:shift-valid}}, we will refer to the \textit{aggregate valuation} of two valuations. The aggregation of two bid lists is defined as the concatenation of the two individual lists. For any two strong-substitutes valuations $u^1$ and $u^2$ associated with valid bid lists $\bids^1$ and $\bids^2$, we define the \textit{aggregate valuation} $u^{1,2}$ as the strong-substitutes valuation associated with the aggregation of the two bid lists $\bids^1$ and $\bids^2$; note that the existence of $u^{1,2}$ is guaranteed by Theorem~{\ref{thm:validity}}.%
\footnote{Note that $u^{1,2}$ can be expressed as $u^{1,2}(\xb) = \max_{\yb \in \Z^n} (u^{1}(\yb) + u^{2}(\yb - \xb))$ for all $\xb \in \Z^n$. This is stated, for instance, in footnote 20 of \citep{BK}, and is related to the \textit{aggregate cost function} in \citep[Chapter 11, page 335]{murota-book}. We do not make use of this way of expressing $u^{1,2}$ in the following proofs.}

\begin{lemma}\label{lemm:local-valuation}
    Let $\bids$ be an integral bid list that is $\varepsilon$-valid at integral prices $\pb$ with $\varepsilon < 1/2$. There exists a strong-substitutes valuation $u$ such that $D_{\bids}(\qb) = D_{u}(\qb)$ for all $\qb \in B(\pb;\varepsilon)$.
\end{lemma}

\begin{proof}[Proof of Lemma \ref{lemm:local-valuation}]
Let $\mathcal{H}$ be the collection of hyperplanes obtained by taking the affine hull of all facets in the polyhedral complex $\Pi_{\bids}$ of $\bids$ that intersect $B(\pb;\varepsilon)$. Let $H \in \mathcal{H}$ be a hyperplane obtained from facet $F$ of $\Pi_{\bids}$. As $F$ is normal to a strong-substitutes vector and the bid entries are integral, we can write $H = \{\qb \in \R^n \mid q_i - q_j = c \}$ for some $i,j\in[n]_0$ and some integer $c$. For any $\qb \in B(\pb;\varepsilon) \cap H$ we have $p_i - p_j - 2\varepsilon < q_i - q_j < p_i - p_j + 2\varepsilon$. Due to integrality of $c$ and $\varepsilon < 1/2$, we have $q_i - q_j = p_i - q_j = c$. Hence every $H \in \mathcal{H}$ contains $\pb$.

We now construct a second rational polyhedral complex $(\Pi,w)$ as follows. Take the union of all hyperplanes in $\mathcal{H}$, and decompose this union into a polyhedral complex to obtain $\Pi$.
The cells of this complex are intersections of hyperplanes in $\mathcal{H}$. Every cell (and so every facet and $(n-2)$-dimensional cell) of $\Pi$ contains $\pb$, as $\pb$ is contained in every $H \in \mathcal{H}$. So, in particular, $F\cap P\neq\emptyset$ for every facet $F$ of $\Pi$. Moreover, for each facet $F$ of $\Pi$, either $F\cap P\subseteq F'\cap P$ for some facet $F'$ of $\Pi_\bids$, or $F\cap P$ has at most $(n-2)$-dimensional intersection with the facets of $\Pi_\bids$; this holds because the facets of $\Pi_\bids$ meet facets with distinct affine spans at every point on their boundaries. This allows us to define the weight function $w$ as follows. For any facet $F$ of $\Pi$, we let $w(F) = w_{\bids}(F')$ if $F \cap P\subseteq F' \cap P$ for some facet $F'$ of $\Pi_{\bids}$, and we let $w(F) = 0$ otherwise. Note that $(\Pi,w)$ is balanced within $P$, since $(\Pi_\bids,w_\bids)$ is balanced (see also the proof of Theorem~{\ref{thm:validity}}): dividing a facet of a balanced polyhedral complex into two equally-weighted facets, and adding zero-weighted facets, does not affect balancing. But every $(n-2)$-cell of $\Pi$, and every facet of $\Pi$ containing this $(n-2)$-cell, has non-zero intersection with $P$. So $(\Pi,w)$ is balanced.

The facets of $\Pi_{\bids}$ that intersect with $P$ have non-negative weight due to the local validity of $\bids$ at $\pb$ (by the same argument as in the proof of Theorem~{\ref{thm:validity}}). Hence, the weight function $w$ is non-negative by construction.
Now let $(\Pi^+, w^+)$ be obtained from $(\Pi, w)$ by removing all zero-weighted facets, and let $\mathcal{L}^+$ be the union of the facets of $\Pi^+$. Note that $\mathcal{L}_{\bids}\cap P = \mathcal{L}^+\cap P$; and if $F'$ is a facet of $\mathcal{L}_\bids$ and $F^+$ is a facet of $\mathcal{L}^+$ such that $\dim(F\cap F^+\cap P)=n-1$, then $w(F')=w^+(F^+)$. Fix some prices $\qb\in P$, $\qb\notin\mathcal{L}^+$, and let $\xb$ be the bundle uniquely demanded by $\bids$ at $\qb$.
By Theorem~{\ref{thm:valuation-complex}}, there exists a valuation function $u$ such that $D_u(\qb) = \{ \xb \}$, $\mathcal{L}_u = \mathcal{L^+}$ and $w_u = w^+$. Using the same argument as in the proof of Theorem~{\ref{thm:validity}}, we have that the normal vectors and weights of the facets of $\mathcal{L}_{\bids}$ fully specify how demand of $\bids$ changes in $P$. Hence, by Proposition~{\ref{prop:demand-change}}, it follows that $D_\bids(\qb')=D_u(\qb')$ for all $\qb'\in P$. Moreover, as the facets of $\mathcal{L}_u$ are normal to strong-substitutes vectors, the valuation~$u$ is strong-substitutes by Theorem \ref{thm:ss-demand-type}.
\end{proof}

\begin{lemma}
\label{lemma:correct-shift}
Suppose that $G = \begin{pmatrix}G_1\\G_2\end{pmatrix}$ is an invertible $n \times n$ matrix with strong-substitutes rows. Moreover, let $H_i$ denote the $i$-th column of the matrix $H \coloneqq \begin{pmatrix}0 \\ G_2 \end{pmatrix}$. Then, for any $i \in [n]$, we have either (i) $G^{-1}H_i=\eb^S$ for some $S\subseteq [n]$ with $i\in S$, or (ii) $G^{-1}H_i=-\eb^S$ for some $S\subseteq [n]$ with $i \not \in S$.
\end{lemma}
\begin{proof}
Suppose $G$ is a lower triangular matrix with diagonal values $g_{ii} = 1$ and $g_{ij} \in \{0, -1\}$ for all $i > j$. This is without loss of generality, as we can permute columns and multiply by a diagonal matrix~$U$ with entries $\pm 1$ on its diagonals otherwise; in this case, replace $G$ by $PUG$ (where $P$ denotes the permutation matrix), below and adjust the proof slightly. Note that the rows of $PUG$ remain strong-substitutes.

In order to prove our lemma, we associate $G$ with a directed graph $\mathcal{D}$ on vertices $[n]_0$ as follows. As $G$ is invertible and the rows are strong-substitutes, each row takes the form $\eb^j - \eb^i$ for some distinct $i, j \in [n]_0$, where we define $\eb^0 = \vec{0}$. Thus, for each row $\eb^j - \eb^i$ we add an arc from $i$ to $j$.

We establish some properties of $\mathcal{D}$. Firstly, $\mathcal{D}$ is an arborescence rooted at $0$; that is, it is a directed graph in which there is exactly one path from $0$ to any other vertex. This can be seen by induction on $k \in [n]$ if we define $G_{kk}$ as the sub-matrix of $G$ consisting of the first $k$ rows and columns and let $\mathcal{D}_{kk}$ denote its graph. It is immediate that $\mathcal{D}_{11}$ is an arborescence. Secondly, consider $\mathcal{D}_{kk}$ and assume the claim holds for $k-1$. The last row (and column) of $G_{kk}$ adds a single arc from a good $j < k$ to~$k$, so $\mathcal{D}_{kk}$ remains an arborescence.

Secondly, we can express $G$ as $G = I_n - L$, where $I_n$ is the identity matrix and $L$ is a strictly lower triangular matrix. Moreover, $L$ admits the intuitive interpretation of being the transpose of the adjacency matrix of $\mathcal{D}$. (To see this, note that $\eb^j - \eb^i$ is the $j$-th row in $G$, so $L_{ji} = 1$ if, and only if, there exists an arc from $i$ to $j$.) In the following, let $L^k$ be the $k$-th power of $L$ (with $L^0 = I_n$). We recall the well-known result that $L^k_{ji}$ denotes the number of directed walks of length~$k$ from vertex $i$ to $j$. In particular, as $\mathcal{D}$ is a tree, all walks are paths. This implies that, for all $i,j \in [n]$, we have $L^{k}_{ij} \in \{0,1\}$ and $L^k_{ij} = 1$ for at most one value of $k$. Moreover, $L^n = \bm{0}$, as no path can have length $n$. It follows that the inverse of $G$ is
\begin{equation}
\label{eq:G-inverse-L}
    G^{-1} = \sum_{k=0}^{n-1}L^k,
\end{equation}
and $G^{-1}$ has only entries $\{0,1\}$.

Fix $i \in [n]$ and let $J_i$ be any subset of the in-neighbours of $i$. (The \textit{in-neighbours} of vertex $i$ in $\mathcal{D}$ are the vertices for which there exists an arc from $i$ to $j$.) We claim that for $\vec{h} = -\sum_{j \in J_i}\eb^j$, there exists a subset $S$ such that $i \in S$ and $G^{-1} \vec{h} = -\eb^S$, and for $\vec{h} = \eb^i - \sum_{j \in J_i} \eb^j$, there exists a subset $S$ such that $i \in S$ and $G^1 \vec{h} = \eb^S$. This claim immediately implies the main statement of Lemma~{\ref{lemma:correct-shift}}, as $H_i$ takes the form $H_i = -\sum_{j \in J_i}\eb^j$ or $H_i = \eb^i - \sum_{j \in J_i} \eb^j$ for some $J_i$.

It remains to prove our claim. Note first that $J_i \subseteq \{i+1, \ldots, n\}$. For any $j \in [n]_0$, let $S(j)$ denote the vertices reachable from vertex $j$ in~$\mathcal{D}$. As $L^{k}_{jl} = 1$ if there exists a path of length $k$ from vertex $l$ to $j$, and $L^k_{jl} = 0$ otherwise, \eqref{eq:G-inverse-L} implies that $G^{-1} \eb^i = \eb^{S(j)}$ is the indicator vector denoting, for each $l \in [n]$, whether~$l$ is reachable from $j$ (along a path of arbitrary length).
Next, observe that $J_i$ contains endpoints of arcs from $i$, so the arborescence property of~$\mathcal{D}$ implies that the sets $S(j)$ for $j \in J_i$ are pairwise disjoint. Moreover, we see that $S(j) \subseteq S(i)$, as there is an arc from $i$ to every $j \in J_i$. This allows us to say the following.
If $\vec{h} = \sum_{j \in J_i} \eb^j$, then we have $G^{-1} \vec{h} = -\sum_{j \in J_i} \eb^{S(j)} = \eb^S$, where we define $S \coloneqq \bigcup_{j \in J_i} S(j)$, and we note that $i \not \in S$.
On the other hand, if $\vec{h} = \eb^i - \sum_{j \in J_i} \eb^j$, then we have $G^{-1} \vec{h} = \eb^{S(i)} - \sum_{j \in J_i} \eb^{S(j)} = \eb^{S}$, where we define $S \coloneqq S(i) \setminus \bigcup_{j \in J_i} S(j)$, and we note that $i \in S$. The last equality holds as $S(j) \subseteq S(i)$ for all $j \in J_i$.
\end{proof}

\begin{observation}
\label{obs:Hstep}
Suppose that $G$ is an $n'\times n$ matrix with strong-substitutes rows and that $S\subseteq [n]$. Then $G\eb^S\in\{-1,0,1\}^n$ and, in particular, $G\eb^S \geq -\bm{1}$.
\end{observation}

\begin{lemma}
\label{lemma:shift-subs:old}
Suppose we have two strong-substitutes integer-valued valuations, $u^1$ and~$u^2$, and let $\uagg$ denote the aggregate valuation. Fix $i \in [n]$ and $\varepsilon \in (0,1)$, and write $\uep$ for the valuation given by $\uep(\xb) = u^2(\xb) + \varepsilon x_i$. For any $\rb^1, \rb^2 \in \Z^n$, suppose there exists a price~$\pb^*$ such that $\rb^1 \in  D_{u^1}(\pb^*)$, and $\rb^2\in D_{\uep}(\pb^*)$. Then there exists a price~$\pb^0$ such that $\rb^1 \in D_{u^1}(\pb^0)$, and $\rb^2\in D_{u^2}(\pb^0)$. 
Moreover, for any price $\pb$ such that $\rb \coloneqq \rb^1 + \rb^2 \in D_{\uagg}(\pb)$, we have $\rb^1 \in D_{u^1}(\pb)$, and $\rb^2 \in D_{u^2}(\pb)$.
\end{lemma}
\begin{proof}
Fix bundles $\rb^1$ and $\rb^2$ that are, by assumption, demanded at $\pb^*$ under $u^1$ and~$\uep$. For $j \in \{1, 2\}$, we define $P^j \coloneqq \{\pb \geq 0 \mid \rb^j \in D_{u^j}(\pb)\}$, the set of prices at which $\rb^j$ is demanded under $u^j$. It follows from the definition of $\uep$ that the set $P^{2[\varepsilon, i]}$ of prices at which $\rb^2$ is demanded by $\pb$ under $\uep$ is equal to $P^2 + \varepsilon \eb^i$. Moreover, by assumption we have $\pb^* \in P^1 \cap P^{2[\varepsilon,i]}$, so this intersection is non-empty.

As $P^1$ and $P^2$ are polyhedra in $\R^n$ (cf.~\citep[Proposition 2.7]{BK}), we can describe each region $P^j$ by the system of linear inequalities
\begin{equation}\label{eqn:GpA}
	H^j \pb \geq A^j,
\end{equation}
where $H^j$ is an $a_j \times n$ matrix for some $a_j \in \Z_+$ and $A^j \in \R^{a_j}$.
Moreover, Proposition~3.10 in \citep{BK} tells us that we can choose $H^j$ such that each row of $H^j$ is a strong-substitutes vector. Additionally, as valuation $u^j$ is integer-valued, we have $A^j \in \Z^{a_j}$. It follows that $P^{2[\varepsilon,i]} = P^2 + \varepsilon \eb^i$ is defined by
\begin{equation}\label{eqn:HpB-shift}
H^2(\pb - \varepsilon \eb^i) \geq A^2 \quad \Leftrightarrow \quad H^2 \pb \geq A^2 + \varepsilon H^2_i,
\end{equation}
where $H^2_i$ is the $i$-th column of $H^2$.
This allows us to express the intersection $P^1 \cap P^{2[\varepsilon,i]}$ as
\begin{equation}\label{stacked}
H\pb\geq A+\varepsilon\hat{H}_i,
\end{equation}
where we define $H = \begin{pmatrix}H^1 \\ H^2 \end{pmatrix}$, $A = \begin{pmatrix} A^1 \\ A^2 \end{pmatrix}$, and $\hat{H} = \begin{pmatrix} 0 \\ H^2 \end{pmatrix}$.

Now we argue similarly to Theorem 19.1 in \citep{Schrijver-book}. We know that $P^1 \cap P^{2[\varepsilon,i]} \neq \emptyset$. Let $F = \{ \pb \geq 0 \mid H'\pb = A' + \varepsilon \hat{H}_i' \}$ be a minimal face of this intersection, where $H'\pb\geq A'+\varepsilon \hat{H}_i'$ is a subsystem of \eqref{stacked} with linearly independent rows. Without loss of generality, we can write $H' = \begin{pmatrix}G & \tilde{G}\end{pmatrix}$, where~$G$ is invertible (permuting the columns of $H'$ if necessary, which is equivalent to relabelling goods).
Moreover, we can choose this permutation such that the $i$-th coordinate remains within the first square matrix -- that is, if $G$ is $n'\times n'$ then $i \leq n'$. In addition, $G$ inherits from $H^1,H^2$ the property that its rows are strong-substitutes vectors. Now define the two price vectors
\[
\pb^\varepsilon\coloneqq
\begin{pmatrix}
	G^{-1}\left(A'+\varepsilon \hat{H}'_i\right) \\
	0
\end{pmatrix}
\text{ and }
\pb^0\coloneqq
\begin{pmatrix}
	G^{-1}A' \\
	0
\end{pmatrix}.
\]
Note that $\pb^0$ is an integral vector, since $A$ is integral and $G$ is unimodular. Moreover, by Lemma~\ref{lemma:correct-shift}, we know that either (i) $G^{-1}\hat{H}'_i=\eb^S$ for some $S\subseteq[n']$ with $i\in S$, or (ii) $G^{-1} \hat{H}'_i = -\eb^S$ for some $S\subseteq [n']$ with $i \not \in S$ (here $n'$ is the dimension of $G$). In either case, $\pb^\varepsilon - \varepsilon \eb^i \in \{ \pb^0 \pm \varepsilon \eb^{S'}\}$ for some $S'\subseteq[n]$.
We also see immediately that $\pb^\varepsilon$ is a vector in $F$, and hence lies in $P^1\cap P^{2[\varepsilon,i]}$, so $\pb^\varepsilon$ satisfies \eqref{stacked}. We now use this knowledge to show that $\pb^0 \in P^1 \cap P^2$.

As $\pb^\varepsilon \in P^1$, it satisfies \eqref{eqn:GpA} for $j=1$. Note also that the rows of $H^1$ are strong-substitutes and $(\pb^0 - \pb^\varepsilon) \in \{\pm \varepsilon \eb^S\}$ for some $S \subseteq[n]$, so \eqref{eqn:GpA} and Observation~\ref{obs:Hstep} imply
\[
	H^1\pb^0=H^1\pb^\varepsilon+H^1(\pb^0-\pb^\varepsilon)\geq
A^1-\varepsilon \mathbf{1}.
\]
The integrality of $H^1$, $\pb^0$ and $A^1$, together with $\varepsilon \in (0,1)$, imply $H^1\pb^0 \geq A^1$. Similarly, we argue that $H^2 \pb^0 \geq A^2$. Indeed, $\pb^varepsilon$ satisfies \eqref{eqn:HpB-shift} and we have $(\pb^0 - \pb^\varepsilon + \varepsilon\eb^i) \in \pm \varepsilon \eb^{S}$ for some $S \subseteq [n]$, so \eqref{eqn:HpB-shift} and Observation~\ref{obs:Hstep} imply
\[
H^2 \pb^0 = H^2 (\pb^\varepsilon - \varepsilon \eb^i)
+ H^2(\pb^0 - \pb^\varepsilon + \varepsilon\eb^i)\geq A^2
- \varepsilon \mathbf{1} \geq A^2.
\]
The last inequality again follows due to the integrality of $H^2, \pb^0$ and $A^2$. Thus we have established that $\pb^0 \in P^1 \cap P^2$, which is equivalent to $\rb^j \in D_{u^j}(\pb^0)$ for $j=1,2$, our first claim.

Now let $\rb := \rb^1 + \rb^2$ and let $\pb$ be a price at which $\rb$ is aggregately demanded, so $\rb \in D_{\uagg}(\pb)$. Suppose $\rb = \vec{s}^1 + \vec{s}^2$ is a potentially different equilibrium allocation of bundles across the two agents, so that $\vec{s}^j \in D_{u^j}$ holds. By Proposition 2.F.1 in \citep{MWG95}, we know that $(\bs^j-\rb^j)\cdot(\pb-\pb^0)\leq 0$ for $j=1,2$, with equality holding only if the agent is indifferent between both bundles at both prices. Hence, adding across $j=1,2$, we obtain $((\bs^1+\bs^2)-\rb)\cdot(\pb-\pb^0)\leq 0$, with equality holding only if both agents are indifferent between their two bundles at both prices. But as $\bs^1+\bs^2=\rb$, equality holds, so we arrive at $\rb^j\in D_{u^j}(\pb)$ for $i=1,2$.
\end{proof}

\begin{proof}[Proof of Proposition~\ref{prop:shift-valid}]
By Lemma~{\ref{lemm:local-valuation}}, every bid list that is $\eta$-valid at $\pb$ with $\eta<\frac12$ can be associated with a strong-substitutes integral-valued valuation such that their indirect utility function and demand correspondence coincide in $B(\pb,\eta)$. The bid lists $\bids^j$ and $\bids \setminus \bids^j$ are $1/2$-valid at $\pb$ by Lemma~{\ref{obs:validity-margin}} and the fact that integral bids have a surplus of at least 1 at integral prices. Moreover, the shifted bid list of bidder $j$ and the bid list $\bids'$ are $(1/2-|\varepsilon|)$-valid at $\pb$ (this follows from the definition of local validity). Recalling that $|\epsilon|<\frac14$, we fix $\eta$ with $2|\epsilon|<\eta<\frac12$ and let $u^2$ denote a strong-substitutes integral-valued valuation of bidder $j$ corresponding  in $B(\pb,\eta)$ to $\bids^j$; and let $u^1$ denote a (strong-substitutes integral-valued) valuation corresponding in $B(\pb,\eta)$ to the bids $\bids \setminus \bids^j$ of the other bidders. As above, $\uagg$ denotes the aggregate of these valuations, and it corresponds in $B(\pb,\eta)$ to $\bids$. By assumption, we have $\rb \in D_{\uagg}(\pb)$.

As in Lemma~{\ref{lemma:shift-subs:old}}, we write $\uep$ for the valuation given by $\uep(\xb) = u^2(\xb) + \varepsilon x_i$. Note that the polyhedral complex of $\uep$ is the polyhedral complex of $u^2$ shifted by $\varepsilon<\eta/2$ in direction $\eb^i$. Hence, in $B(\pb;\frac{\eta}2)$ the demand $D_{u^{2[\varepsilon,i]}}$ coincides with the demand of bidder $j$'s bids after shifting, and $D_{u^{1,2[\varepsilon,i]}}$ coincides with $D_{\bids'}$ in $B(\pb;\frac{\eta}2)$.
Since both $u^1$ and $\uep$ are strong-substitutes valuations, there exists a price $\pb^* \in \R^n$ such that $\rb \in D_{u^{1,2[\varepsilon,i]}}(\pb^*)$, and thus there exist $\rb^1,\rb^2\in\Z^n$ with $\rb^1+\rb^2=\rb$ such that $\rb^2\in D_{u^1}(\pb^*)$ and $\rb^2\in D_{\uep}(\pb^*)$. By Lemma~{\ref{lemma:shift-subs:old}}, it follows that $\rb^j \in D_{u^j}(\pb)$ for $j=1,2$.

Like in the proof of Lemma~{\ref{lemma:shift-subs:old}}, we let $P^j := \{ \pb' \geq 0 \mid \rb^j \in D_{u^j}(\pb') \}$ for $j \in \{1,2\}$ and observe that $P^{2[\varepsilon,i]} := \{ \pb' \geq 0 \mid \rb^2 \in D_{\uep}(\pb') \}  = P^2+\varepsilon\eb^i$. We know that $\pb \in P^1 \cap P^2$ (so this intersection is non-empty) and $\pb^* \in P^1 \cap P^{2[\varepsilon,i]}$ (so this intersection is also non-empty). We seek $\pb^\varepsilon\in P^1\cap P^{2[\varepsilon,i]}$ with $\pb^\varepsilon \in \pb \pm \varepsilon \eb^S$ for some $S \subseteq [n]$; this is sufficient to show $\rb \in D_{u^{1,2[\varepsilon,i]}}(\pb^\varepsilon)$.  Since $|\varepsilon|<\frac{\eta}2$, we know $\pb^\varepsilon\in B(\pb,\frac{\eta}2)$, the region in which demand from the valuations coincides with demand from the bids. Our first claim, that $\rb\in D_{\bids'}(\pb^\varepsilon)$ follows.

Analogous to line \eqref{eqn:GpA} above, we express the polyhedron $P^j$ as $H^j \pb' \geq A^j$, for $j \in \{1,2\}$, where the rows of $H^j$ are strong-substitutes vectors. Therefore $P^1 \cap P^{{2[\varepsilon, i]}} = P^1 \cap (P^2 + \varepsilon \eb^i)$ is defined by the set of inequalities
\begin{equation}\label{goal}
H^1 \pb' \geq A^1 \text{ and } H^2(\pb' - \varepsilon \eb^i) \geq A^2.
\end{equation}
Now, for $j=1,2$, let
\begin{equation}
\label{subsystem}
(H^j)'\pb'\geq(A^j)'
\end{equation}
be the subsystem of $H^j\pb'\geq A^j$ consisting of all rows which hold with equality at~$\pb$.  (If both of these subsystems are empty then subsequent arguments will simply set $\pb^\varepsilon=\pb$).  Write $H= \begin{pmatrix}(H^1)'\\(H^2)'\end{pmatrix}$, $A=\begin{pmatrix}(A^1)'\\(A^2)'\end{pmatrix}$ and $\hat{H}=\begin{pmatrix}0\\(H^2)'\end{pmatrix}$.  Thus when we stack the corresponding rows of \eqref{goal}, we obtain the new system
\begin{equation}\label{first shift}
H\pb'\geq A+\varepsilon\hat{H}_i.
\end{equation}
This defines a superset of $P^1\cap(P^2+\varepsilon)$, since we have removed some rows from $H^j\pb'\geq A^j$ for $j=1,2$.  But we have seen that $P^1\cap(P^2+\varepsilon)$ contains $\pb^*$, so \eqref{first shift} defines a non-empty polytope. So let $F=\{\pb' \geq 0 \mid H'\pb'= A' + \varepsilon \hat{H}'_i \}$ be a minimal face of this polytope, where
\begin{equation}\label{subsubsystem}
H'\pb'\geq A' +\varepsilon\hat{H}'_i
\end{equation}
is a subsystem of \eqref{first shift} with linearly independent rows.

Without loss of generality, we can write $H' = \begin{pmatrix}G & \tilde{G}\end{pmatrix}$, where~$G$ is invertible (permuting the columns of $H'$ if necessary, which is equivalent to relabelling goods).
Moreover we can choose this permutation such that the $i$-th coordinate remains within the first square matrix -- that is, if $G$ is $n'\times n'$ then assume when we re-order that $i\leq n'$.  Correspondingly write $\pb' = \begin{pmatrix} \pb'_1 \\ \pb'_2 \end{pmatrix}$. In summary, our system \eqref{first shift} is now written
\begin{equation}\label{broken system}
\begin{pmatrix}
    G & \tilde{G}
\end{pmatrix}
\begin{pmatrix}
    \pb'_1\\
    \pb'_2
\end{pmatrix}
	\geq A' +\varepsilon\hat{H}'_i.
\end{equation}
By definition of the rows identified at \eqref{subsystem}, we have
$
	\begin{pmatrix}G & \tilde{G}\end{pmatrix}
	\begin{pmatrix}\pb_1 \\ \pb_2\end{pmatrix}
	= A'
$
and so
\begin{equation}
\label{eq:pb1}
    \pb_1=G^{-1}A'-G^{-1}\tilde{G}\pb_2.
\end{equation}
We now define $\pb^\varepsilon$ as follows
\begin{equation}
\pb^\varepsilon\coloneqq
	\begin{pmatrix}
		G^{-1}(A' +\varepsilon\hat{H}'_i) - G^{-1}\tilde{G}\pb_2
			\\
		\pb_2.
	\end{pmatrix}
\end{equation}
in order to obtain
\begin{equation}\label{rows that work}
H'\pb^\varepsilon=\begin{pmatrix}G & \tilde{G}\end{pmatrix}\pb^\varepsilon =
		 A' +\varepsilon\hat{H}'_i.
\end{equation}
We wish to show that \eqref{goal} holds for $\pb'=\pb^\varepsilon$. From \eqref{eq:pb1} we observe that $\pb^\varepsilon-\pb=\begin{pmatrix}\varepsilon G^{-1}\hat{H}'_i\\0\end{pmatrix}$.  So, by Lemma~\ref{lemma:correct-shift}, either $\pb^\varepsilon=\pb+\varepsilon\eb^{S}$ for some $S\subseteq[n]$ with $i\in S$, or $\pb^\varepsilon=\pb-\varepsilon\eb^S$ where $S\subseteq[n]$ with $i\notin S$. In either case, $\pb^\varepsilon - \varepsilon \eb^i \in \{ \pb \pm \varepsilon \eb^{S'} \}$ for some $S'\subseteq[n]$.

We also immediately observe that $\pb^\varepsilon$ satisfies \eqref{subsubsystem} with equality, and so $\pb^\varepsilon\in F$. Thus, by definition of $F$, $\pb^\varepsilon$ satisfies \eqref{first shift}. It follows by definition of \eqref{subsystem} that, for every row in the original systems $H^j\pb'\geq A^j$ that holds with equality at $\pb$, the corresponding shifted equation in \eqref{goal} holds at $\pb^\varepsilon$. It remains to show that this is also true for rows that are slack at $\pb$. Let $(H^j)''\pb'\geq (A^j)''$ be the subsystems of such rows, for $j=1,2$, and observe that by integrality it follows that
\begin{equation}\label{slack}
    (H^j)''\pb\geq (A^j)''+\mathbf{1}.
\end{equation}
Now, applying Observation~\ref{obs:Hstep} to $\pb^\varepsilon - \pb$ and \eqref{slack} for $j=1$, we obtain
\[
(H^1)''\pb^\varepsilon=(H^1)''\pb+(H^1)''(\pb^\varepsilon-\pb)\geq
(A^1)''+(1-\varepsilon)\mathbf{1}\geq (A^1)''.
\]
And, applying Observation~\ref{obs:Hstep} to $\pb^\varepsilon-\pb-\varepsilon\eb^i$, and \eqref{slack} for $j=2$,
\[
(H^2)''(\pb^\varepsilon-\varepsilon\eb^i)
	=(H^2)''\pb+(H^2)''(\pb^\varepsilon-\pb-\varepsilon\eb^i)
	\geq (A^2)''+(1-\varepsilon)\mathbf{1}\geq (A^2)''.
\]
So, for every row in the original systems $H^j\pb'\geq A^j$ that is slack at $\pb$, the corresponding shifted equation in \eqref{goal} holds for $\pb^\varepsilon$.  This completes the proof that $\pb^\varepsilon\in P^1\cap (P^2+\{\varepsilon\eb^i\})$. It follows that $\rb \in D_{u^{1,2[\varepsilon,i]}}(\pb^\varepsilon)$, which concludes the proof of our first claim.

To see that we can compute $\pb^\varepsilon$ in polynomial time, recall that the prices at which $\rb$ is demanded are equivalent to the prices that minimise the Lyapunov function $g(\pb) = f_{\bids'}(\pb) + \rb \cdot \pb$.
It is known (see {\citep[Theorem 7.3]{ST15}}) that $f_{\bids'}$ (and thus $g$) satisfies the \textit{generalised submodularity} property
\[
f_{\bids'}(\pb) + f_{\bids'}(\qb) \geq f_{\bids'}(\pb \lor \qb) + f_{\bids'}(\pb \land \qb) \quad (\forall \pb, \qb \in \R^n),
\]
where $\lor$ and $\land$ denote the component-wise maximum and minimum, respectively. Hence
\begin{align*}
	h^+(S) \coloneqq g(\pb + \varepsilon \eb^S) - g(\pb), \\
	h^-(S) \coloneqq g(\pb + \varepsilon \eb^S) - g(\pb),
\end{align*}
are submodular set functions. Hence, in order to determine a price vector $\pb^\varepsilon \in \left \{ \pb + \varepsilon \eb^{S} \mid S \subseteq [n] \right \}$ at which $\rb$ is demanded, we find minimisers $S^+$ and $S^-$ of $h^+$ and $h^-$ using SFM, then let
\[
	\pb^\varepsilon = \argmin \left \{ g \left (\pb + \varepsilon \eb^{S^+}
\right ), g \left (\pb - \varepsilon \eb^{S^-} \right ) \right \}.
\]
Note that we do not require \emph{minimal} submodular minimisers to find $\pb^\varepsilon$, and so this step takes $2T(n)$ time, where $T(n)$ is the time is takes to perform submodular minimisation on $h^\pm$. We also note that finding $\pb^\varepsilon$ is analogous to finding a steepest descent direction in Appendix~\ref{appendix:price-finding}.
\end{proof}

\subsection{Proof of Lemma~\ref{lemma:procedure3-reduce}}
\label{appendix:procedure3-reduce}
Here we prove that applying \textsc{ShiftProjectUnshift} strictly reduces the
number of edges in the marginal bids graph $G_\mathcal{A}$.
\begin{proof}[Proof of Lemma~\ref{lemma:procedure3-reduce}]
	First we show that every edge in $G_{\mathcal{A}'}$ is present in
$G_\mathcal{A}$. To see this, fix some bid $\bid$ and note by
Observation~\ref{obs:no-new-marginals} and Lemma~\ref{lemma:project-invariants}
that Steps 1, 2 and 3 do not make $\bid$ marginal on any new goods.

Secondly we prove that for any multi-bidder cycle $C$ with cycle-link
good $i^*$ and incident label $j^*$, \textsc{ShiftProjectUnshift} removes at
least one of
the edges of $C$ from the marginal bids graph. Re-label the goods going around
cycle $C$ as $1,\ldots,k$, so that $1=i^*$ and so that bidder $j^*$ placed the
bid marginal between goods $1$ and $2$. Also, for convenience, index the
marginal bid that is marginal between goods $i$ and $i+1$ as `$i$', and index
the marginal bid between goods $k$ and $1$ as `$k$'. Thus
$\bid^1\in\bids^{j^*}$ and $\bid^k\notin\bids^{j^*}$. (In general the
differently labelled bids need not be from different bidders). The existence of
a marginal bid between goods $i$ and $i+1$ means that we have an equality

\begin{equation}\label{eqn:1}
b^i_i-p_i=b^i_{i+1}-p_{i+1}
\end{equation}
for $i=1,\ldots,k-1$, and also
\begin{equation}\label{eqn:2}
b^k_k-p_k=b^k_1-p_1
\end{equation}
But if we take the sum of the first $k-1$ equations, and cancel, we find
\[
\sum_{i=1}^{k-1} b^i_i - \sum_{i=1}^{k-1} b^i_{i+1} =
\sum_{i=1}^{k-1}(p_i-p_{i+1})=p_1-p_k
\]
So it must hold that
\begin{equation}\label{eqn:cycle}
\sum_{i=1}^{k-1} (b^i_i - b^i_{i+1}) = b^k_1-b^k_k.
\end{equation}

For any bid $\bid$ in the bid lists $(\bids^j)_{j \in J}$, fixed
$j^* \in J$ and $i^* \in [n]$, let
\begin{align*}
\tilde{\bid}\coloneqq\left\{
\begin{array}{ll}
\bid+\frac{1}{10}\eb^{i^*}	& \bid\in\bids^{j^*}\\
\bid	&\text{otherwise}
\end{array}
\right.
\end{align*}

Suppose we replace each bid $\bid$ by $\tilde{\bid}$ as above and recompute the
prices. Since prices shift by at most $\frac{1}{10}$, a bid that is {\em not}
marginal on a pair of goods cannot {\em become} marginal on them. If we assume
for a contradiction that all edges in $C$ are all still present in $G'$ then we
can write down similar expressions to \eqref{eqn:1} and \eqref{eqn:2} (w.r.t.~new prices), then eliminate those prices obtaining a version of
\eqref{eqn:cycle}, namely
$\sum_{i=1}^{k-1}
\left (\tilde{b}^i_i - \tilde{b}^i_{i+1} \right) =
\tilde{b}^k_1-\tilde{b}^k_k$.
But by definition of $\tilde{\bid}$, we know that
\begin{align*}
\sum_{i=1}^{k-1} (\tilde{b}^i_i - \tilde{b}^i_{i+1}) &=
\sum_{i=1}^{k-1} (b^i_i - b^i_{i+1}) +\frac{1}{10},\\
\text{ and } \tilde{b}^k_1-\tilde{b}^i_k &= b^k_1-b^k_k.
\end{align*}
This is inconsistent with \eqref{eqn:cycle}, so we have the required
contradiction.
\end{proof}

\section{Finding market-clearing prices}
\phantomsection
\label{appendix:price-finding}
\label{apx:price-finding}
We discuss two iterative steepest descent algorithms, \textsc{MinUp} and \textsc{LongStepMinUp}, from {\citep{Shioura17}} that determine a minimal minimiser of an $L^\natural$-convex function. Both algorithms use the SFM subroutine described in Section~\ref{sec:prelims} to find the component-wise minimal discrete steepest descent direction. For the second algorithm, we present two methods of computing step lengths and show that both methods yield a polynomial running time in our bidding-language setting.

In order to apply the steepest descent method to our price-finding problem, we define a Lyapunov function $g$ and note in Proposition~\ref{prop:L-natural-convex} that its restriction to $\Z^n_+$ is $L^\natural$-convex. Moreover, Lemma~\ref{lemma:lyapunov-equivalence} states that the lowest market-clearing price is integral and finding it reduces to determining the minimal minimiser of $g$. This approach generalises an algorithm used by Ausubel's ascending auction design~\citep{Aus06} and \citet{GS00} for the task of finding equilibrium prices in single-unit markets.

Let $\bids$ be a valid bid list. By Theorem~\ref{thm:validity}, the function $f_{\bids}$ defined in \eqref{eq:bid-utility-function} is the indirect utility function $f_u$ of some strong-substitutes valuation $u$. For any strong-substitutes valuation $u$, the \emph{Lyapunov function} with regard to indirect utility function $f_u$ and target bundle $\tb$ is defined as $g_{\tb}(\pb) \coloneqq f_u(\pb) + \tb \cdot \pb$. We suppress the subscript $\tb$ if it is clear from context. We can use our knowledge of the bids in $\bids$ and \eqref{eq:bid-utility-function} to express $g_{\tb}$ for the list $\bids$ as
\begin{equation}\label{eq:lyapunov}
	g_{\tb}(\pb) \coloneqq f_{\bids}(\pb) + \tb \cdot \pb = \sum_{\bid \in \bids} w(\bid) \max_{i \in [n]_0} (b_i -
p_i) + \tb \cdot \pb.
\end{equation}
From~\eqref{eq:lyapunov} it is clear that we can evaluate $g$ at any price $\pb$ in time $O(n|\bids|)$.

A function $f:\Z^n \to \mathbb{R}$ is $L^\natural$-convex if it satisfies the \emph{translation submodularity} property,
\begin{equation}\label{eq:translation-submodularity}
	f(\pb) + f(\qb) \geq f( (\pb - \alpha \mathbf{1}) \vee \qb) + f(\pb \wedge (\qb + \alpha \mathbf{1})) \quad (\forall \pb, \qb \in \Z^n_+, \forall \alpha \in \Z_+).
\end{equation}
Here, $\vee$ and $\wedge$ denote the component-wise maximum and minimum, respectively.

The following proposition and lemma demonstrate that we can use algorithms for minimising $L^\natural$-convex functions to find the minimal price $\pb^*$ at which $\tb$ is demanded.

\begin{proposition}[\citep{MSY2013, MSY2016}]\label{prop:L-natural-convex}
	The Lyapunov function $g_{\tb}$ restricted to $\Z^n_+$ is $L^\natural$-convex.
\end{proposition}
\begin{proof}
It is known that an indirect utility function $f_u$ restricted to $\Z^n$ is $L^\natural$-convex if and only if the valuation function $u$ is strong-substitutes (cf.~\citep[Theorem 7.1]{ST15}). Secondly, it is easy to verify that adding a linear term to an $L^\natural$-convex function preserves $L^\natural$-convexity. As the bid list $\bids$ is valid, we have $f_{\bids} = f_u$ for some strong-substitutes valuation $u$ and it follows that the function $f_{\bids}(\pb) + \tb \cdot \pb$ is $L^\natural$-convex.
\end{proof}

\begin{lemma}\label{lemma:lyapunov-equivalence}
	The Lyapunov function $g_{\tb}$ with regard to a valid (integral) bid list
and any target bundle $\tb$ is convex, and $\pb$ is a minimiser of $g$ if and
only if it is a market-clearing price for target bundle $\tb$. Moreover, the
minimal minimiser of $g$ is integral.
\end{lemma}
\begin{proof}
	The first statement is immediate from the fact that $f_\bids$ is convex
and $\tb \cdot \pb$ is a linear term. To see the second statement, note that for any
market-clearing price $\pb$ of target bundle~$\tb$, we have $g(\pb) = u(\tb)$,
whereas for any price $\pb$ at which $\tb$ is not demanded, we have
$g(\pb) = \max_{\xb \in D(\pb)} (u(\pb) - \xb \cdot \pb) + \tb \cdot \pb > u(\tb)$. Here
$u$ is the valuation function as defined in Section~\ref{sec:prelims} and we
use \eqref{eq:iuf}.
Integrality of the minimal minimiser of $g$ was shown in \citep[Corollary~4.4]{MSY2016}. Alternatively, we can see that this follows from the fact that if $\tb$ is demanded at $\pb$, then $\tb$ is also demanded at $\lfloor \pb \rfloor$. To see this, fix a bid $\bid$ and note that if $\bid$ demands good $i$ at~$\pb$, it still demands $i$ at $\lfloor \pb \rfloor$.
Indeed, as $\bid$ demands $i$ at $\pb$, we have $b_j - p_j \leq b_i - p_i$ for
all goods~$j$, which implies
	\[
		b_j - \lfloor p_j \rfloor < b_j - p_j + 1 \leq b_i - p_i + 1
\leq b_i - \lfloor p_i \rfloor + 1.
	\]
	Due to the integrality of the bids and prices in $\lfloor p \rfloor$,
this implies $b_j - \lfloor p_j \rfloor \leq b_i - \lfloor p_i \rfloor$.
\end{proof}

Let $\pb$ be a point that is dominated by some minimiser of $g$. \Citet{Shioura17} noted that $\pb$ minimises~$g$ if and only if $g(\pb) \leq g(\pb + \eb^S)$ for every $S \subseteq [n]$. For any integral point $\pb \in \Z^n_+$ and $S \subseteq [n]$, let the \emph{slope function} $g'(\pb;S) \coloneqq g(\pb + \eb^S) - g(\pb)$ denote the amount by which $\pb$ decreases when moving in the direction of $\eb^S$. If $S$ minimises $g'(\pb;S)$, we call $\eb^S$ a \emph{steepest descent direction}. For any integral vector $\pb$, the $L^\natural$-convexity of $g$ implies that $g'(\pb;S)$ is an integral submodular function \citep[Theorem 7.2]{ST15} and hence there exists a unique component-wise \emph{minimal steepest descent direction}~$\eb^{S_0}$.

Let $\pb^*$ be the minimal minimiser of $g$. If $\pb$ is dominated by $\pb^*$ and we move in the minimal steepest descent direction, the point $\pb + \eb^{S_0}$ is also dominated by $\pb^*$. This suggests Algorithm~\ref{alg:minup} (called \textsc{GreedyUpMinimal} in \citep{Shioura17}), which iterates a point $\pb$ by moving by some step $\eb^{S_0}$, all the while remaining dominated by $\pb^*$, until $\pb = \pb^*$.

\begin{algorithm}[hbt!]
\caption{\textsc{MinUp}}
\label{alg:minup}
\begin{algorithmic}[1]
    \State Pick a point $\pb \leq \pb^*$ (e.g.~$\pb = \bf{0}$).
    \State\label{step:sfm}Find the inclusion-wise minimal set $S_0 \subseteq [n]$ minimising $g'(\pb;S)$.
    \If{$S_0 = \emptyset$}
      \State \Return $\pb$.
    \Else
      \State Set $\pb = \pb + \eb^{S_0}$ and go to line \ref{step:sfm}.
    \EndIf
\end{algorithmic}
\end{algorithm}

By the existence of the auctioneer's reserve bids, we can initialise $\pb$ to
$\mathbf{0}$. For this starting point, running time analysis by \citet{MS14}
implies that Algorithm~\ref{alg:minup} iterates exactly $\| \pb^* \|_\infty$ times. As we
know that $\pb^*$ is bounded from above by the component-wise maximum over all
bids~$\bids$, the number of iterations is at most $M \coloneqq \max_{\bid \in \bids}
\| \bid \|_\infty$. In each iteration of Step 2, we can determine $S_0$ using
the SFM subroutine described in Section~\ref{sec:prelims} that finds a minimal
minimiser of a given function in time $T(n)$. This leads to the following
running time for \textsc{MinUp}.

\begin{theorem}[cf.~\citep{MS14}]\label{thm:minup-running-time}
	The algorithm \textsc{MinUp} finds the component-wise minimal
market-clearing price in time $O(MnT(n))$, where $T(n)$ denotes the time it
takes to find a minimiser of $g'$.
\end{theorem}

We note that the running times of the two practical SFM algorithms mentioned in Section~\ref{sec:prelims} are given with respect to an upper bound on the
absolute value of the objective function. In order to provide such an upper
bound for our slope function $g'$, observe that the two points $\pb$ and $\pb +
\eb^S$ share a demanded bundle $\xb$, for any $\pb$ and $S \subseteq [n]$. As
every bid contributes at most one item to $\tb$ and $\xb$, this implies
$g'(\pb;S) = (\tb-\xb) \cdot \eb^S \leq |\bids|$.

\subsection{Longer step sizes}
The analysis by \citet{MS14} shows that \textsc{MinUp} performs optimally for
an iterative algorithm that is constrained to steps with an $L_\infty$-size of
at most~1. As described by \citet{Shioura17}, we can, however, exploit
monotonicity properties of the function $g'(\pb; S_0)$ in order to increase
step sizes without changing the trajectory of $\pb$ as the algorithm runs. This
reduces the number of SFM subroutine calls, the most expensive part of the
algorithm. In particular, if $\eb^{S_0}$ denotes the minimal steepest descent
direction at $\pb$, we take a single long step $\lambda \eb^{S_0}$ for some
$\lambda \in \Z_+$ that is equivalent to several consecutive steps of
\textsc{MinUp} in the same direction~$\eb^{S_0}$.

We follow \citep{Shioura17} in choosing step length
\begin{equation}
\label{eq:steplength}
\lambda(\pb, S_0) = \max \{ \lambda \in \Z_+ \mid g'(\pb;S_0) =
g'(\pb+(\lambda-1)\eb^{S_0};S_0) \},
\end{equation}
that is, the farthest distance we can move before the slope
\begin{equation*}
g'(\pb+(\lambda-1)\eb^{S_0};S_0) = g(\pb + \lambda \eb^{S_0}) - g(\pb +
(\lambda-1)\eb^{S_0})
\end{equation*}
in the direction of $\eb^{S_0}$ changes. We give two methods to compute \eqref{eq:steplength} in Section~\ref{subsubsec:step-length}. This leads to the following algorithm
(referred to as \textsc{GreedyUp-LS} in~\citep{Shioura17}).

\begin{algorithm}[hbt!]
\caption{\textsc{LongStepMinUp}}
\label{alg:longstepminup}
\begin{algorithmic}[1]
    \State Pick a point $\pb \leq \pb^*$ (e.g.~$\pb = \bf{0}$).
    \State\label{step:longstep-sfm}Compute inclusion-wise minimal minimiser $S_0
  \subseteq [n]$ of $g'(\pb;S_0)$ using SFM.
	 	\State Determine $\lambda(\pb,S_0)$, as defined by
  \eqref{eq:steplength}, using a method from Section~\ref{subsubsec:step-length}.
    \If{$S_0 \neq \emptyset$}
      \State Set $\pb = \pb + \lambda(\pb, S_0) \eb^{S_0}$ and go to line \ref{step:longstep-sfm}.
    \EndIf
	\State \Return $\pb$.
\end{algorithmic}
\end{algorithm}

Note that the values of $\pb$ follow the same trajectory for $\textsc{MinUp}$
and $\textsc{LongStepMinUp}$. This is an immediate consequence of
Lemma~\ref{lemma:same-trajectory}, which rests on the monotonicity properties
of $g'(\pb, S_0)$ stated in Proposition~\ref{prop:monotonicity}.

\begin{proposition}[{\citep{Shioura17}, Theorem 4.16}]
\label{prop:monotonicity}
Let $S_0$ and $S_0'$ be \emph{minimal} steepest descent directions at $\pb$ and
$\pb + \eb^{S_0}$, respectively. Then we have
	\begin{enumerate}
		\item \label{prop:monotonicity-one} $g'(\pb + \eb^{S_0};S'_0) >
g'(\pb;S_0)$ or
		\item \label{prop:monotonicity-two} $g'(\pb + \eb^{S_0};S'_0) =
g'(\pb;S_0)$ and $S_0 \subseteq S'_0$.
	\end{enumerate}
\end{proposition}

\begin{lemma}\label{lemma:same-trajectory}
	Let $S_0$ denote the minimal steepest descent at $\pb$ and let
$\lambda(\pb, S_0)$ be defined by \eqref{eq:steplength}. Then $\eb^{S_0}$ is
the minimal steepest descent at $\pb + (\lambda-1) \eb^{S_0}$ for any $1 \leq
\lambda \leq \lambda(\pb, S_0)$.
\end{lemma}

\citet{Shioura17} bounds the number of iterations of $\textsc{LongStepMinUp}$
as follows.
\begin{theorem}[{\citep{Shioura17}, Theorem 4.17}]
	The number of iterations of \textsc{LongStepMinUp} is at most $n \max
\{-g'(\bm{0};S) \mid S \subseteq [n] \}$.
\end{theorem}

Note that $\bf{0}$ and $\eb^{S_0}$ share a demanded bundle $\xb$, so for any $S
\subseteq [n]$ we have
\[
	g'({\bm{0}}; S) = g(\eb^S) - g(\bm{0}) = f_u(\eb^S) + \tb \cdot \eb^S -
f_u(\bm{0})
	= (\tb-\xb) \cdot \eb^S \geq - \sum_{i \in [n]} x_i \geq -|\bids|,
\]
as each bid contributes at most one unit to the demanded bundle. This implies
that the $\textsc{LongStepMinUp}$ algorithm takes at most $n|\bids|$ iterations to
find component-wise minimal equilibrium prices.

\subsubsection{Computing the step length}
\label{subsubsec:step-length}
Fix a price $\pb$ and let $\eb^{S_0}$ be the component-wise minimal steepest
descent direction at $\pb$. We describe two methods to compute the step length
defined by \eqref{eq:steplength}. The first method uses binary search and is
also suggested in \citep{Shioura17}, while second method exploits our knowledge
of the bids to determine $\lambda(\pb; S_0)$. Note that we can evaluate
$g'(\pb;S_0)$ in time $O(n|\bids|)$, as $g'(\pb;S_0) = g(\pb+\eb^{S_0}) -
g(\pb)$.

\begin{theorem}
	The \textsc{LongStepMinUp} algorithm in combination with the
\emph{binary search} and \emph{demand change} methods has a respective running
time of $O(n^2|\bids|^2\log M + n|\bids|T(n))$ and $O(n^2|\bids|^3 +
n|\bids|T(n))$.
\end{theorem}

A description of the two methods, as well as a proof of this theorem, is
provided below. Note that as the two methods have different running time
guarantees, the best method in practice is context-specific.

\paragraph{Binary search.}
Note that $\lambda(\pb; S_0)$ can be bounded by the total number of unit steps
in $\textsc{MinUp}$, which in turn is bounded by $M$. By
Proposition~\ref{prop:monotonicity}, we have that $g'(\pb;S_0) < g'(\pb+\lambda
\eb^{S_0}; S_0)$ implies $g'(\pb;S_0) < g'(\pb+\lambda' \eb^{S_0}; S_0)$ for
all $\lambda' > \lambda$. Hence we can apply binary search to find
$\lambda(\pb;S_0)$ in time $O(n|\bids|\log M)$.

\paragraph{Demand change.}
Alternatively, we can exploit our knowledge of the individual bids to determine
$\lambda(\pb;S_0)$. We proceed by performing a demand-change procedure, which
repeatedly determines the highest value $\mu$ for which every bid $\bid \in
\bids$ demands the same goods (or a superset thereof) at prices
$\pb+\mu\eb^{S_0}$ that it demands at $\pb + \eb^{S_0}$ and updates $\pb$ to
$\pb + \mu \eb^{S_0}$.

Fix a bid $\bid$ and let $I$ denote the goods it demands at $\pb$. If $I \not
\subseteq S_0$, then $\bid$ demands the same set of goods $I_\bid = I \setminus
S_0$ at all prices $\pb + \mu'\eb^{S_0}$ with $\mu' \geq 1$. On the other hand,
suppose $I \subseteq S_0$. We define $\mu_\bid \coloneqq \min_{j \in [n]_0 \setminus {S_0}} ((b_i
- p_i) - (b_j - p_j))$, where $i$ is any good in $I$, and consider the set of
goods that $\bid$ demands at prices $\pb + \mu' \eb^{S_0}$ with $\mu' \geq 1$.
If $1 \leq \mu' < \mu_{\bid}$, then $\bid$ demands $I$, if $\mu' = \mu_{\bid}$
then $\bid$ demands a superset of $I$ and if $\mu' > \mu_{\bid}$, then $\bid$
demands none of the goods in $I$.
We define $\mathcal{C}$ to be the list of bids that demand a subset of $S_0$ at
$\pb$, and let $\mu(\pb, S_0) \coloneqq \min_{\bid \in \mathcal{C}} \mu_\bid$. Note
that we can determine the value of $\mu(\pb, S_0)$ in time $O(n|\bids|)$.

The demand-change procedure takes as input a price $\pb^0 \coloneqq \pb$ and direction
$\eb^{S_0}$, and consists of the following steps. Initially, we set $\lambda$
to $0$. Compute $\mu(\pb, S_0)$, and increment $\lambda$ by $\mu(\pb, S_0)$. If
we have $\lambda = \lambda(\pb^0, S_0)$, return $\lambda$. Otherwise, increment
$\pb$ by $\mu \eb^{S_0}$ and repeat the above with the same value for
$\eb^{S_0}$.

Note that in order to check whether $\lambda = \lambda(\pb^0;S_0)$, we can
compute $g'(\pb^0+\mu\eb^{S_0};S_0)$ and verify that
$g'(\pb^0+\mu\eb^{S_0};S_0) = g'(\pb^0;S_0)$, which takes time $O(n|\bids|)$.
Lemma~\ref{lemma:demand-change} proves that the demand-change procedure
correctly computes the value of $\lambda(\pb;S_0)$ in time $O(n|\bids|^2)$.

\begin{lemma}\label{lemma:demand-change}
	The demand-change procedure returns $\lambda(\pb;S_0)$ in at most
$|\bids|$ iterations.
\end{lemma}
\begin{proof}
	Let $K$ denote the number of iterations in the demand-change procedure
and let $\pb^k, \mu^k$ and $C^k$ be the values of $\pb, \mu(\pb, S_0)$ and $C$
after the $k$-th iteration. For notational convenience, let $\pb^0 = \pb$ and
$\mu^0 = 0$.
	Proposition~\ref{prop:monotonicity} implies that it suffices to show
	\begin{equation}\label{eq:lambda-conditions}
		g'(\pb;S_0) = g'(\pb^K - \eb^S; S_0) \text{ and } g'(\pb;S_0) <
g'(\pb^K; S_0)
	\end{equation}
	in order to prove the first claim that $\lambda = \lambda(\pb,S_0)$.

	Firstly, note that there is a bundle $\xb$ that is demanded at
$\pb^{k-1}$ and $\pb^{k}$, as well as all prices in between these two points.
Indeed, if a bid $\bid$ demands good $i$ at $\pb^{k-1}+\eb^{S_0}$, then $\bid$
also demands~$i$ at the two prices $\pb^{k-1}$ and $\pb^k - \eb^{S_0} =
\pb^{k-1} + (\mu^{k}-1)\eb^{S_0}$, by construction of $\mu^k$. Hence, making
use of \eqref{eq:iuf}, we get
	\[g'(\pb^{k-1};S_0) = (\tb-\xb) \cdot \eb^{S_0} = g'(\pb^k - \eb^{S_0};S_0).
	\]
	Secondly, note that for every $0 \leq k < K$, we have $g'(\pb^k -
\eb^{S_0};S_0) = g'(\pb^k;S_0)$, and for the last iteration $K$ we have
$g'(\pb^K - \eb^{S_0};S_0) < g'(\pb^K;S_0)$. This implies
\eqref{eq:lambda-conditions}.

	Now we turn to the second claim, that $K \leq |\bids|$. This follows
from the fact that $C^{k-1} \supsetneq C^k$ and $|C^0| \leq |\bids|$. Indeed,
if $\bid \not \in C^{k-1}$, then $\bid$ demands goods $I \not \subseteq S_0$ at
$\pb^{k-1}$ and $I_\bid = I \setminus S_0$ at any prices $\pb^{k-1} + \mu
\eb^{S_0}$ for any $\mu \geq 1$, so $b \not \in C^k$. Now suppose $\bid \in
C^{k-1}$ is a bid for which $\mu^{k-1} = \mu_{\bid}$ (at $\pb^{k-1}$). Then we
claim that $\bid \not \in C^k$. Indeed, the demanded goods $I'$ of $\bid$ at
$\pb + (\mu_{\bid}+1)\eb^{S_0}$ satisfy $I' \not \subseteq S_0$ by construction
of $\mu_\bid$. By the same argument as above, $\bid$ demands goods not in $S_0$
at all prices $\pb + \mu \eb^{S_0}$ with $\mu \geq \mu_\bid +1 $.
\end{proof}

\subsection{Some practical improvements}
The computation time of \textsc{MinUp} and \textsc{LongStepMinUp} is dominated
by the task of finding a minimal set $S_0$ minimising $g'(\pb;S)$ using SFM. In
practice, we can exploit our direct access to the bids to speed up the
computation of $S_0$ by decreasing the dimensionality of the submodular
function to be minimised. The following observations can be seen as a special
case of observations by~\citet{Aus06}. Fix $\pb \in \Z^n_+$ and let $\eb^{S_0}$
be the minimal steepest descent direction at $\pb$. Note that $\pb$ and $\pb +
\eb^{S_0}$ share a demanded bundle $\xb$ due to the structure of our price
space. Hence
	\[
		g'(\pb; S_0) = g(\pb+\eb^{S_0}) - g(\pb) = (\tb-\xb) \cdot \eb^{S_0},
	\]
and $S_0$ minimises this term if and only if $S_0$ contains all indices $i \in
[n]$ with $t_i < x_i$ and no indices $j \in [n]$ with $t_j > x_j$. In
particular, we have $S_0 \coloneqq \{ i \in [n] \mid t_i < x_i \}$ due to the minimal
minimiser property of $S_0$.

If $\pb$ has a unique demanded bundle $\xb$, which is easy to verify, it is
straightforward to compute $\xb$. Hence, in this case we can determine the
minimal steepest descent direction $\eb^{S_0}$ without performing SFM. In the
case that $\pb$ has at least two demanded bundles, the demanded bundle $\xb$
shared by $\pb$ and $\pb + \eb^{S_0}$ is unknown. However, if we define index
sets $I_{\pb}$ and $J_{\pb}$ as
	\begin{align*}\label{eq:Ip-Jp-def}
		I_{\pb} &\coloneqq \{ i \in [n] \mid x_i > t_i, \forall \xb \in D(\pb)
\}, \\
		J_{\pb} &\coloneqq \{ i \in [n] \mid x_i \leq t_i, \forall \xb \in
D(\pb) \},
	\end{align*}
we have $I_{\pb} \subseteq S_0$ and $J_{\pb} \cap S_0 = \emptyset$.
(Note that the inequality in the definition of $I_{\pb}$ and $J_{\pb}$ is strict and not strict, respectively.)
Hence we can restrict ourselves to minimising the submodular function $h \colon
[n]\setminus (I_{\pb} \cup J_{\pb}) \to \Z$ defined by $h(T) = g'(\pb;T \cup
I_{\pb})$ and reduce the dimensionality of the SFM problem from $n$ to $n -
|I_{\pb} \cup J_{\pb}|$. The following lemma shows that computing $I_{\pb}$ and
$J_{\pb}$ is cheap.

\begin{lemma}\label{lemma:Ip-Jp}
	$I_{\pb}$ and $J_{\pb}$ can be computed in time $O(n|\bids|)$.
\end{lemma}
\begin{proof}
	Note that $I_{\pb} = \{i \in [n] \mid \min_{\xb \in D(\pb)} x_i > t_i
\}$ and $J_{\pb} = \{i \in [n] \mid \max_{\xb \in D(\pb)} x_i \leq t_i \}$,
where we compute the minimum and maximum component-wise. Fix $i \in [n]$. The
minimum $\min_{x \in D(\pb)} x_i$ is attained if no marginal bid selects good
$i$. Hence
	\[
	\min x_i = \sum \{ w(\bid) \mid \bid \in \bids \text{ is non-marginal
and demands }i \}.
	\]

	Similarly, we claim that $\max_{\xb \in D(\pb)} x_i$ is attained if
good $i$ is selected whenever possible and thus $\max_{\xb \in D(\pb)} x_i$ is
the sum of the weights of the bids for which $i$ is demanded. Consider the
price perturbation $\pb'$ defined by $p'_i = p_i$ and $p'_j = p_j +
\varepsilon$ for $j \not = i$. Here $\varepsilon > 0$ is chosen sufficiently
small so that every bid that is marginal on $i$ at $\pb$ is non-marginal and
demands $i$ at $\pb'$, while every non-marginal bid demanding good $i$ at $\pb$
is also non-marginal and demands $i$ at $\pb'$. This perturbation corresponds
to our proposed rule to compute $\max x_i$ by selecting good $i$ whenever
possible.

Note that all bundles $\xb'$ at price $\pb'$ have the same number of
items of $i$ and are also demanded at~$\pb$. As $\pb \leq \pb'$ and $p_i =
p'_i$, the strong-substitutes property implies that there exists $\xb' \in
D(\pb')$ such that $x_i \leq x'_i$ for all $\xb \in D(\pb)$. In other words,
any demanded bundle $\xb'$ at $\pb'$ maximises~$x_i$.
\end{proof}

\section{Hardness of testing validity of unrestricted bid lists}
\label{sec:testing}
Here we show that the question of whether a list of bids is valid is
\conp-complete, even when we restrict the problem to unit bids (bids with weights $\pm 1$). We also present a simple algorithm verifying the validity of a given list of positive and negative unit bids that runs in polynomial time if the number of goods, or the number of negative bids, is bounded by a constant. Let $\bids$ denote a list of unit bids, and let $\posbids$ and $\negbids$ denote the positive and negative bids.

\subsection{Checking validity is \conp-complete}
\label{subsec:coNP-complete}
\begin{definition}
Given a list of positive and negative bids $\bids$, the problem \validbids\ is
to decide whether $\bids$ is valid.
\end{definition}

\begin{theorem}[Theorem~\ref{thm:np}]\label{thm:conp}
\validbids\ is \conp-complete.
\end{theorem}

The proof of Theorem~\ref{thm:conp} uses an equivalent definition of validity
for the list of bids. Define regions in $\mathbb{R}^n_+$ that are generated by
a point $\pb$ and coordinates $i,j \in [n]$ as follows.
\begin{align}
	{H^p_i} &\coloneqq \{ \xb \in \mathbb{R}^n_+ \mid \xb \leq \pb, x_i = p_i \}
\label{eq:hod}\\
	F^p_{ij} &\coloneqq \{ \xb + \beta \eb^{[n]} \mid \beta \in \mathbb{R}_+, \xb
\leq \pb, x_i = p_i \text{ and }x_j = p_j. \} \label{eq:flange}
\end{align}

A bid $\bid = (b_1, \ldots, b_n; b_{n+1})$ is \textit{contained} in a region $H^p_i$ (or $F^p_{ij}$) if the `valuation' vector $(b_1, \ldots, b_n)$ consisting of the first $n$ components lies in it. Moreover, we say that $H^p_i$ or $F^p_{ij}$ is \emph{negative} for a list of bids if it
contains more negative than positive bids, otherwise it is \emph{non-negative}.

\begin{observation}\label{obs:containment-condition}
	For any $\xb, \yb \in \mathbb{R}^n_+$, we have $\xb \in F_{ij}^y$ if and only if $x_i-y_i = x_j-y_j$ and $\xb - (x_i-y_i) \mathbf{1} \leq \yb$.
\end{observation}

\begin{observation}\label{obs:transitive-inclusion}
	Fix $i,j \in [n], i \not = j$. Then containment in the regions given in
\eqref{eq:hod} and \eqref{eq:flange} is transitive in the sense that, for any
$\xb,\yb,\zb \in \mathbb{R}^n_+$,
	\begin{itemize}
		\item $\xb \in H^y_{i}$ and $\yb \in H^z_{i}$, implies $\xb \in
H^z_{i}$.
		\item $\xb \in F^y_{ij}$ and $\yb \in F^z_{ij}$, implies $\xb
\in F^z_{ij}$.
	\end{itemize}
\end{observation}

The following definition of valid bids restates the second part of Theorem~{\ref{thm:validity}} in terms of regions {\eqref{eq:hod}} and {\eqref{eq:flange}}.

\begin{definition}[Valid bids]\label{def:valid-bids}
A list of positive and negative bids is \emph{valid} if, for any point $\pb \in \mathbb{R}^n_+$ and two coordinates $i,j \in [n]$, $H^p_i$ and $F^p_{ij}$ are non-negative.
\end{definition}

On the basis of Definition~\ref{def:valid-bids}, the proof of Theorem
\ref{thm:np} works by reducing the well-known NP-complete problem \threecnf\ to
\validbids\ by means of an intermediate NP-complete decision problem
\mnd, which we define in Definition~\ref{def:mnd}. The NP-completeness of \mnd\
is established in Theorem~\ref{def:mnd-np} and the reduction from \mnd\ to
\validbids\ is given below as the proof of Theorem~\ref{thm:np}. For any two
points $\xb,\yb \in \mathbb{R}^n$, we say that $\xb$ dominates $\yb$ if $x_i
\geq y_i$ for all $i \in \{1,\ldots, n \}$.

\begin{definition}\label{def:mnd}
	Given two lists $P$ and $N$ of vectors in $\mathbb{R}^n_+$, \mnd\ is
the problem of deciding whether there exists a vector $\zb \in \mathbb{R}^n_+$
that dominates more vectors in $N$ than in $P$.
\end{definition}

\begin{theorem}\label{def:mnd-np}
	The problem \mnd\ is NP-complete.
\end{theorem}
\begin{proof}
	Note that verifying whether a given point $\zb \in \mathbb{R}^n_+$
dominates more vectors in $N$ than in $P$ can be done in polynomial time. This
establishes membership of \mnd\ in the class NP. We show completeness by
reducing from \threecnf. Recall that a boolean formula $\phi$ defined on $n$
variables $x_1, \ldots, x_n$ is 3-CNF if
	\[
		\phi = \bigwedge_{i=1}^m C_i,
	\]
	where $C_i = L^i_1 \vee L^i_2 \vee L^i_3$ are its clauses and $L^i_j
\in \{x_1, \ldots, x_n, \overline{x_1}, \ldots, \overline{x_n}\}$ are its
literals. Here $\overline{x_i}$ denotes the negation of $x_i$. Without loss of
generality, we assume that the three variables appearing in each clause are
distinct. Given such a 3-CNF formula $\phi$, first construct a pair of lists
$(P_i, N_i)$ of vectors in $\mathbb{R}^n_+$ for each clause $C_i$ as follows.
For notational convenience, $\eb^{a \ldots z}$ denotes the characteristic
vector $\eb^{\{a,\ldots, z\}}$ of $\{a,\ldots, z\}$.

	\begin{itemize}
		\item If $C_i = x_a \vee x_b \vee x_c$, let $N_i = \{
\eb^\emptyset \}$ and $P_i = \{ \eb^{abc}\}$.
		\item If $C_i = x_a \vee x_b \vee \overline{x_c}$, let $N_i =
\{ \eb^\emptyset, \eb^{abc} \}$ and $P_i = \{ \eb^{ab} \}$.
		\item If $C_i = x_a \vee \overline{x_b} \vee \overline{x_c}$,
let $N_i = \{ \eb^\emptyset, \eb^{ab}, \eb^{ac} \}$ and $P_i = \{ \eb^{a},
\eb^{abc} \}$.
		\item If $C_i = \overline{x_a} \vee \overline{x_b} \vee
\overline{x_c}$, let $N_i = \{\eb^{a}, \eb^{b}, \eb^{c}, \eb^{abc} \}$ and $P_i
= \{ \eb^{ab}, \eb^{ac}, \eb^{bc} \}$.
	\end{itemize}

	Finally, let $N$ be the sum of the $N_i$ and let $P$ be the sum of the
$P_i$ together with $m-1$ copies of $\eb^\emptyset$. This completes our
reduction from \threecnf\ to \mnd. Clearly, $P$ and $N$ can be constructed from
$\phi$ in polynomial time. To show correctness of the reduction, associate with
every vector $\zb \in \mathbb{R}^n_+$ a truth assignment $\beta_z$ (on boolean
variables $x_1, \ldots, x_n$) that sets $x_i$ to \true\ if $z_i < 1$ and
\false\ otherwise.

	\begin{observation}\label{obs:z-x-correspondence}
		Let $\zb \in \mathbb{R}^n_+$ be a point and $\beta_z$ be its
associated truth assignment. For any $i \in \{1, \ldots, m\}$, $\zb$ dominates
one more point in $N_i$ than in $P_i$ if $\beta_z$ satisfies $C_i$ and an equal
number of points in $N_i$ and $P_i$ if $\beta_z$ does not satisfy $C_i$.
	\end{observation}

	Suppose $\phi$ is a satisfiable 3-CNF formula with satisfying truth
assignment $\beta$. Then define $\zb \in \mathbb{R}^n_+$ by
	\[
		z_i =
		\begin{cases}
		0 & \text{ if }\beta[x_i] = \true, \\
		1 & \text{ else.}
		\end{cases}
	\]
	Hence $\beta$ is the truth assignment associated with $z$. By
Observation~\ref{obs:z-x-correspondence}, $\zb$ dominates one more point in
$N_i$ than in $P_i$ for each $i$, so due to the additional $m-1$ origin points
$\eb^\emptyset$ added to $P$, $\zb$ dominates exactly one more point in $N$
than in $P$.
Conversely, suppose there exists a point $\zb \in \mathbb{R}^n_+$ that
dominates more points in $N$ than in $P$, and let $\beta_z$ be its
corresponding truth assignment. Then $\zb$ dominates at least $m$ more points
in $N$ than in $\sum_{i=1}^m P_i$. By Observation~\ref{obs:z-x-correspondence},
we know that $\zb$ dominates at most one more point in $N_i$ than in $P_i$,
which implies that this is the case for each $i$. Hence by the same
observation, $\beta_z$ satisfies all clauses of $\phi$ and thus $\phi$
itself.
\end{proof}

\begin{proof}[Proof of Theorem~\ref{thm:np}]
	In order to show that a list of bids $\bids$ is not valid, it suffices
to provide a certificate in the form of a point $\pb$ and coordinates $i,j \in
[n]$. We can verify in polynomial time whether ${H_i^p}$ or ${F_{ij}^p}$ is
negative. This establishes membership of \validbids\ in \conp.

	Let $(N,P)$ be an instance of \mnd\ and $n$ be the dimension of its
vectors. We construct an $(n+1)$-dimensional instance $\bids$ of \validbids\
such that $(N,P) \in$ \mnd\ if and only if $\bids \not \in$ \validbids\ as
follows. Let the negative and positive bids of $\bids$ be given by prepending a component to the vectors as follows.
	\begin{align*}
		\negbids & \coloneqq \{ (1,\vb) \mid \vb \in N \}\\
		\text{ and } \posbids & \coloneqq \{ (1,\wb) \mid \wb \in P \} + \{ (0,\vb) \mid \vb
\in N \} + \{ (1,\vb) + \eb^{[n+1]} \mid \vb \in N \}.
	\end{align*}

Clearly, $\bids = \posbids + \negbids$ can be constructed efficiently.
Note that all $F^x_{ij}$ and all $H^x_i$ with $i \not = 0$ are non-negative
by construction. Indeed, fix ${F^x_{ij}}$ and suppose it contains $r$ negative
bids. Then for each such bid $(1,\vb)$, where $\vb \in N$, there exists a
positive bid $(1,\vb)+\eb^{[n+1]}$ that is contained in $F^x_{ij}$, so that the region is non-negative. Analogously, for any $H^x_i$ with $i \not = 0$, each negative
bid $(1,\vb)$ has a corresponding positive bid $(0,\vb)$. Further, as $H_{0}^x$ contain no negative points unless $x_{0} = 1$,
the correctness of our reduction thus follows from the fact that $\zb$
dominates $r$ points from $N$ and $s$ points from $P$ if and only if
$H^x_{0}$ with $\xb = (1,\zb)$ contains $r$ negative and $s$ positive
bids.
\end{proof}

\subsection{A simple algorithm for testing validity}
In Section~\ref{subsec:coNP-complete}, we show that checking the validity of a bid list in the general case in \conp-complete. Hence we cannot expect an efficient algorithm for the task of checking validity of a list of bids in the general case. Here we present a simple algorithm that tests validity of a given list of bids in polynomial time if the number of goods or the number of negative bids is bounded by a constant. Such constraints may be reasonable in certain economic settings.

\subsubsection{The algorithm}
Our procedure rests on Theorem~\ref{thm:finite-valid-bids}, which reduces the validity condition from Definition~\ref{def:valid-bids} to a finite number of checks. Lemma~\ref{lemma:small-subsets} reduces the number of checks further. Together, the two results immediately yield Algorithm~\ref{alg:check-validity}, whose running time is given in Theorem~\ref{thm:list-validity-running-time}.
For any list of bids $U$, define the minimal dominating vector $md(U)$ of $U$ as the component-wise maximum over the valuation vectors of the bids in $U$, so that $md(U)_i = \max_{\bid \in U} b_i$ for all $i \in [n]$.
We note that if all bids of $U$ agree on some coordinate $i$ (that is, if $b_i = b'_i$ for all $\bid,\bid' \in U$), they lie in $H_i^x$ for large enough $\xb$, and $md(U)$ is the minimal such $\xb$ so that $H_i^x$ contains all points in $U$.
Similarly, for any list of bids $U$ and $i \in [n]$, define $mdF(i,U)$ as
$ mdF(i,U) \coloneqq \min_{\bid \in U} b_i \mathbf{1} + md(\{ \bid -b_i \mathbf{1} \mid \bid \in U \})$. We note that if, for some $i \not = j$, the set $U$ satisfies $b_i-b_j = b'_i-b'_j$ for all $\bid,\bid' \in U$, then the bids in $U$ lie in $F^x_{ij}$ for small enough $\xb$, and it follows from Observation~\ref{obs:containment-condition} that $mdF(i,U) = mdF(j,U)$ is the maximal such $\xb$.

\begin{theorem}\label{thm:finite-valid-bids}
	A list of bids $\bids$ is valid if and only if the following two conditions hold.
	\begin{enumerate}
		\item For every set $U \subseteq \negbids$ of negative bids that agree on the $i$-th coordinate, that is $b_i = b'_i\ \forall b,b' \in U$, the region $H_i^{md(U)}$ is non-negative.
		\item For every set $U \subseteq \negbids$ of negative bids that satisfies $b_i-b'_i = b_j-b'_j$ for all $b,b' \in U$ and some $i \not = j$, the region $F_{ij}^{mdF(i,U)}$ is non-negative.
	\end{enumerate}
\end{theorem}
\begin{proof}
The implication is immediate by Definition~\ref{def:valid-bids}.
Conversely, suppose $\bids$ satisfies the second condition. First we show that
condition 1 of Definition~\ref{def:valid-bids} is satisfied. Fix $H^x_i$, let
$U = \negbids \cap H^x_i$ be the set of negative bids in $H^x_i$ and $\yb =
md(U)$. Then by construction, we have $\bid \in H^y_i$ for all $\bid \in U$ and
$\yb \in H^x_i$. As $H^y_i$ has $|U|$ negative bids, it also contains at least
$|U|$ positive bids by assumption and by Observation
\ref{obs:transitive-inclusion}, these positive bids are also in ${H^x_i}$.
Condition 2 is shown analogously. Fix ${F^x_{ij}}$, let $U$ be the negative
bids in this region and $\zb = mdF(U)$. Suppose $\bid \in {F^z_{ij}}$ for
every $\bid \in U$ and $\zb \in {F^x_{ij}}$. Then
Observation~\ref{obs:transitive-inclusion} implies that ${F^x_{ij}}$ is non-
negative. To see that $\bid \in {F^z_{ij}}$ for every $\bid \in U$, we verify
the conditions of Observation~\ref{obs:containment-condition}. Firstly,
\[
	z_i - z_j = \min_{\bid \in U} b_i - (\min_{\bid \in U} b_i + \max_{\bid \in U}(b_j - b_i)) = \max_{\bid \in U}(b_j - b_i) = b_j - b_i
\]
for all $\bid \in U$, as the difference $b_j - b_i$ is the same for all $\bid \in {F_{ij}^z}$. Secondly, for any $k \in [n]$ and $\bid \in U$, we have
\[
	(\bid - (b_i-z_i) \mathbf{1})_k = z_i + b_k - b_i \leq z_i + \max_{\bid \in U} (b_k - b_i) = z_k.
\]
We can verify $\zb \in F_{ij}^x$ similarly.
\end{proof}

\begin{lemma}\label{lemma:small-subsets}
	For any list of negative bids $U \subseteq \negbids$ there exists a list $U' \subseteq U$ with $|U'| \leq n$ so that $md(U') = md(U)$. Moreover, for each coordinate $i \in [n]$, there exists a list $U'' \subseteq U$ with $|U''| \leq n+1$ so that $mdF(i, U'') = mdF(i, U)$.
\end{lemma}
\begin{proof}
Let $U'$ be a list of bids $\vec{u}^1, \ldots, \vec{u}^n$, where $\vec{u}^k$ is a bid from $U$ that maximises the $k$-th component, i.e.~$\vec{u}^k \in \argmax_{\vec{v} \in U} v_k$. Then $md(U') = md(U)$ by construction. Secondly, fix $i \in [n]$ and let $U''$ be a list of bids
$\vec{u}^1, \ldots, \vec{u}^{n+1}$, where $\vec{u}^1 \in \argmin_{\bid \in U} b_i$ and $\vec{u}^k \in \argmax_{\bid \in U} (b_k-b_i)$ for $k \in \{2, \ldots, n+1\}$. Then $mdF(U'') = mdF(U)$ by construction and we are done.
\end{proof}

\begin{algorithm}[htb]
	\caption{Checking the validity of a bid list $\bids$}
	\label{alg:check-validity}
	\begin{algorithmic}[1]
		\ForAll {subsets $U \subseteq \bids^-$ of negative bids with $|U| \leq n+1$}
			\ForAll {$i,j \in [n]$}
				\State Verify the two conditions given in Theorem~\ref{thm:finite-valid-bids}.
			\EndFor
		\EndFor
	\end{algorithmic}
\end{algorithm}

\begin{theorem}\label{thm:list-validity-running-time}
	Algorithm~\ref{alg:check-validity} runs in time $O \left ( {\binom{|\negbids|}{n}} n^3 |\bids| \right)$. Moreover, if $|\negbids| \leq k$ or $n \leq k$, the algorithm's running time is polynomial in the input size $n|\bids|$.
\end{theorem}
\begin{proof}
There are $\binom{|\negbids|}{n+1}$ subsets $U$ of $\negbids$ for the algorithm to iterate over, while there are $n^2$ possibilities for the index pair $i,j$. For each $U$ and $i,j$, checking the conditions in Theorem~\ref{thm:finite-valid-bids} takes $n|\bids|$ time. This implies the running time.
Suppose the number of negative bids $|\negbids| \leq k$ is bounded by some $k \in \mathbb{N}$. Then ${\binom{|\negbids|}{n+1}}$ is constant and the running time of our algorithm is $O(n^3 (|V| + |W|))$. Secondly, if the number of goods $n \leq k$ is bounded, we have ${\binom{|\negbids|}{n+1}} \leq |\negbids|^{k+1}$ and hence a running time of $O(|\negbids|^{k+1}|\bids|)$.
\end{proof}

\end{document}